\documentclass[onecolumn,pra,eqsecnum,aps,10pt]{revtex4-2}
\usepackage{mathrsfs,units}
\usepackage{amssymb, amsmath, amsthm, dsfont, graphicx,mathrsfs,color,graphicx,physics,bm}
\usepackage[hidelinks]{hyperref}
\makeatletter
\long\def\@makecaption#1#2{%
  \vskip\abovecaptionskip
  \sbox\@tempboxa{#1. #2}%
  \raggedright #1. #2\par
  \vskip\belowcaptionskip}
\makeatother
\usepackage{upgreek}
\usepackage{algorithm}
\usepackage{placeins}
\usepackage{algpseudocode}
\usepackage{orcidlink}
\usepackage{natbib} 
\usepackage{dcolumn}

\newcommand{\indep}{\perp\!\!\!\perp}

\DeclareFontFamily{U}{mathx}{}
\DeclareFontShape{U}{mathx}{m}{n}{<-> mathx10}{}
\DeclareSymbolFont{mathx}{U}{mathx}{m}{n}
\DeclareMathSymbol{\bigtimes}{1}{mathx}{"91}

\algrenewcommand\algorithmicrequire{\textbf{Input:}}
\algrenewcommand\algorithmicensure{\textbf{Output:}}

\newtheorem{definition}{Definition}

\newtheorem{example}{Example}
\newtheorem{problem}{Problem}
\newtheorem{remark}{Remark}
\newtheorem{fact}{Fact}
\newtheorem{theorem}{Theorem}
\newtheorem{lemma}[theorem]{Lemma}
\newtheorem{proposition}[theorem]{Proposition}

\def\bitcoin{%
  \leavevmode
  \vtop{\offinterlineskip 
    \setbox0=\hbox{B}%
    \setbox2=\hbox to\wd0{\hfil\hskip-.03em
    \vrule height .3ex width .15ex\hskip .08em
    \vrule height .3ex width .15ex\hfil}
    \vbox{\copy2\box0}\box2}}

\algrenewcommand\algorithmicprocedure{\textbf{Procedure}
}

\newcolumntype{d}[1]{D{.}{.}{#1}}

\begin{document}
\title{Strategies for quantum-enabled Bitcoin miners}
\author{Zach Manson\, \orcidlink{0000-0001-6462-4101}}
\affiliation{Institute for Quantum Science and Technology, University of Calgary, Alberta, Canada T2N 1N4}
\author{Barry C. Sanders\, \orcidlink{0000-0002-8326-8912}}
\affiliation{Institute for Quantum Science and Technology, University of Calgary, Alberta, Canada T2N 1N4}

\begin{abstract}
We study the impact that two miners equipped with quantum computers purpose-built for quantum Bitcoin mining will have on the 51\% attack threshold of the Bitcoin network, given that the miners are playing a competitive game against each other to be the first to mine a block.
We extend an existing game-theoretic framework for Bitcoin mining and compute the resultant payoff matrices. 
From these payoff matrices, we determine optimal quantum mining strategies for two non-colluding and aggressive quantum miners with multiple opportunities at finding a valid block in an otherwise classical Bitcoin network.
We show that these optimal quantum mining strategies have a negligible effect on the 51\% attack threshold. 
The novelty of our work is the inclusion of the Aggressive Quantum Mining Strategy and the realistic approach of allowing the quantum miners to restart their search if their measurements do not yield a valid block when determining the optimal quantum mining strategies. 
Our result is important for evaluating quantum-mining threats on cryptocurrencies based on Proof-of-Work,
e.g. Bitcoin.
\end{abstract}
\maketitle

\tableofcontents

\section{Introduction}
\label{sec:introduction}

Bitcoin is a decentralized digital currency whose security
relies on the assumption that it is prohibitively expensive for a consortium to control the majority of the Bitcoin network's computational power in a ``51\% attack"~\cite{Ant14}. 
However, if quantum computers are built to scale, this assumption will no longer hold, as Grover's algorithm could be used to reduce the cost of executing a 51\% attack~\cite{Sat20}. 
Prior works have studied this threat, but only under narrow and simplistic adversarial models~\cite{Sat20, LRS19}.
These works show that, in theory, Grover's algorithm could significantly lower the cost of executing a 51\% attack, but the magnitude of this threat within a more practical setting remains unknown.

We study this same threat in a broader and more realistic adversarial model and quantify the impact that Grover's algorithm has on the cost of executing a 51\% attack.
First we define a quantum miner and then state the problem we address.
\begin{definition}[Quantum miner]
\label{def:quantumminer}
A quantum miner is a Bitcoin miner with access to a quantum computer.
\end{definition}
\begin{remark}
\label{remark:weak}
Prior work restricts a quantum miner to undertaking a peaceful mining strategy and just one quantum computation per race~\cite{LRS19},
whereas we permit Sattath's aggressive quantum mining strategy~\cite{Sat20}
and up to~$c$ quantum computations and up to an arbitrary~$\kappa$ Grover iterations over all quantum computations,
compared to prior work for which
\begin{equation}
\label{eq:prior}
c=1,\,
\kappa=k_\text{opt}
\end{equation}
for~$k_\text{opt}$
the optimal number of Grover iterations that maximizes the success probability of Grover's algorithm
for partially inverting the $\operatorname{SHA-256}$ hash function to yield a valid proof-of-work (PoW).
\end{remark}
\noindent
Specifically, we solve the following problem.
\begin{problem}
\label{prob:quantumthreat}
Given two quantum miners in the Bitcoin network, determine their optimal quantum mining strategies and the resulting threshold for the 51\% attack.
\end{problem}
\noindent
A quantum mining strategy is the number of Grover iterations performed by the quantum miners for each of their measurements.

Bitcoin's (capital ``B'' referring to the network/protocol) security against 51\% attacks relies on network-wide consensus~\cite{Ant14}.
To achieve this consensus, Bitcoin relies on the consensus mechanism of PoW, where participants, called miners, compete to solve a computational problem based on unstructured search, with the winning miner earning the right to extend the blockchain and receiving bitcoins (lowercase ``b" referring to the currency) as a reward.
As a large number of miners participate in this competition, and because the difficulty of the underlying computational problem is adjusted periodically based on the total computational power of the network, it is prohibitively expensive for a consortium to dominate the mining process and execute a 51\% attack.
However, if quantum computers are built to scale, miners could use Grover's algorithm to solve the Bitcoin mining problem with a theoretical quadratic advantage over classical methods~\cite{ABL+17}.
Furthermore, quantum Bitcoin miners (quantum miners) may adopt a quantum-specific technique called the aggressive quantum mining strategy~\cite{Sat20}, which increases the rate at which temporary forks in the blockchain occur.
Classically, this is a known security issue, as a higher fork rate decreases the cost to successfully execute a 51\% attack~\cite{Sat20}.

Whereas prior work concerns the theoretical possibility of the threat arising from quantum miners decreasing the cost of successfully executing a 51\% attack~\cite{Sat20, LRS19},
practically the threat level to the Bitcoin network is unknown. 
The quantum threat to Bitcoin in general has been studied from several different angles. 
Kiktenko et al.~\cite{KPA+18} study the quantum threats to both Bitcoin's digital signature scheme and the hash function used in Bitcoin mining, and propose a quantum key distribution-based replacement for Proof-of-Work.
More recently, Babbush et al.~\cite{BZG+26} show that Shor's algorithm can break Bitcoin's digital signature scheme with substantially fewer resources than previously estimated, and Dallaire-Demers and BTQ Technologies~\cite{DDBTQ} show that quantum Bitcoin mining requires an extraordinary amount of power.
No existing work, however, studies the effect of quantum miners on Bitcoin's stale rate under a realistic adversarial model. 
In this work, we fill in this gap by extending an existing game-theoretic model of quantum Bitcoin mining to incorporate behaviour that more accurately reflects how quantum mining would be performed in practice.
Our model accounts for the quantum miners' ability to perform multiple quantum measurements between successive blocks, as well as their use of the aggressive quantum mining strategy. 
We determine the optimal quantum mining strategies for the quantum miners under these conditions, and analyze their effect on the 51\% attack threshold.

In this work, we consider a setting where two quantum miners, Alice and Bob, are hidden within an otherwise classical Bitcoin network.
We extend the game-theoretic approach of Lee et al.~\cite{LRS19} to determine optimal quantum mining strategies for quantum miners who perform multiple quantum measurements between consecutive blocks and use the aggressive quantum mining strategy.
These optimal quantum mining strategies correspond to Nash equilibria, which are strategies for which no quantum miner can increase their expected payoff alone.
These Nash equilibria are computed via the Lemke-Howson algorithm, which computes an exact Nash equilibrium for two-player games~\cite{LH64}.
We simulate the employment of these optimal quantum mining strategies within the Bitcoin network to determine their effect on the 51\% attack threshold, and finally determine the conditions for which executing a 51\% attack becomes feasible in the presence of quantum miners.

Our research is structured as follows.
Here we provide a table of symbols and abbreviations for the reader's convenience.
In~\S\ref{sec:background}, we provide information on the background of our system. 
In~\S\ref{sec:approach}, we present our extended model, the mathematics which describes our model, and the methods used to determine our results.
\S\ref{sec:results} details our key results, which are discussed in~\S\ref{sec:discussion}. 
Finally, we present our conclusions in~\S\ref{sec:conclusion}. 
\clearpage
\begin{table}[ht!]
\caption{Symbols and abbreviations.}
\label{tab:symbols}
\renewcommand{\arraystretch}{1.2}
\begin{ruledtabular}
\begin{tabular}{ll}
\textbf{Symbol / Abbreviation} & \textbf{Definition} \\
\hline
UTXO & Unspent transaction output \\
BTC  & Bitcoin currency \\
TXID & Transaction ID \\
PoW  & Proof-of-Work \\
ASIC & Application-Specific Integrated Circuit \\
DAG  & Directed Acyclic Graph \\
FQS & Finite-Budget Quantum Search \\
PQMS & Peaceful Quantum Mining Strategy \\
AQMS & Aggressive Quantum Mining Strategy \\
S2QR & Symmetric Two-Player Quantum Races \\
TFNP & Total Function Nondeterministic Polynomial \\
PPAD & Polynomial Parity Arguments on Directed Graphs \\
\hline
$V$ & Bitcoin protocol version \\
$\operatorname{hash}$ & Cryptographic hash function \\
$\operatorname{SHA-256}^{2}(s)$ & Double-application of SHA-256 on input $s$ \\
$\mathcal{M}$ & Mempool \\
$T$ & Set of all valid but unconfirmed transactions \\
$R$ & Merkle root \\
$H$ & Block header \\
$P$ & Previous block hash \\
$t$ & Timestamp \\
$\tau$ & Difficulty target \\
$N$ & Nonce \\
$\bar{t}_{11}$ & Median timestamp of the last 11 blocks \\
$\bar{t}$ & Network-adjusted time \\
$\mathcal{S}(X)$ & Set of all subsets whose elements follow the partial ordering of $X$ \\
$f_R$ & Merkle root computation oracle \\
$f_{t}^{\bar{t},\,\overline{t_{11}}}$ & Timestamp validation oracle \\
$\tilde{\tau}$ & Uncompressed difficulty target \\
$D$ & Network difficulty \\
$\gamma$ & Propagation parameter \\
$\gamma_\text{S}$ & Sattath's propagation parameter \\
$\lambda$ & Block arrival rate \\
$p_{\text{stale}}$ & Stale rate \\
$q$ & Hash rate of the Bitcoin network \\
$q_c$ & Fractional hash rate \\
\end{tabular}
\end{ruledtabular}
\end{table}
\begin{table}[ht!]
\caption[]{Symbols and abbreviations (continued).}
\renewcommand{\arraystretch}{1.2}
\begin{ruledtabular}
\begin{tabular}{ll}
\textbf{Symbol / Abbreviation} & \textbf{Definition} \\
\hline
$f$ & 32-bit Boolean oracle \\
$M$ & Subset of marked inputs \\
$\mathscr{H}$ & Hilbert space \\
$\operatorname{uni}(\mathscr{H})$ & Set of unitary operators on $\mathscr{H}$ \\
$\operatorname{her}(\mathscr{H})$ & Set of Hermitian operators on $\mathscr{H}$ \\
$U_f$ & 32-qubit quantum phase oracle \\
$\mathcal{H}$ & Single-qubit Hadamard gate \\
$U_0$ & 32-qubit Householder reflection \\
$G$ & 32-qubit Grover iterate \\
$x$ & Label in computational basis returned by the FQS algorithm \\
$\operatorname{mea}$ & Quantum measurement \\
$[x]$ & $\{1, 2, \dots, x\}$ \\
$[x]_0$ & $\{0, 1, 2, \dots, x\}$ \\
$p(k)$ & Success probability for the FQS algorithm after $k$ iterations \\
$k_{\text{opt}}$ & Optimal number of Grover iterations \\
$k_{\text{opt}}^{32}$ & Optimal number of Grover iterations for 32-bit search \\
\hline
$\wp$ & Finite set of players \\
$n_\wp$ & Number of players \\
$S_i$ & Set of actions for player $i$ \\
$\sigma_i$ & Mixed strategy for player $i$ \\
$\operatorname{supp}(\sigma)$ & Support of mixed strategy $\sigma$ \\
$\Sigma_i$ & Set of mixed strategies for player $i$ \\
$s^*$ & Pure strategy profile \\
$\sigma^*$ & Mixed strategy profile \\
$u_i$ & Payoff function for player~$i$\\
$u$ & Alice's payoff function \\
$v$ & Bob's payoff function \\
$U(\sigma^*)$ & Expected payoff for mixed strategy profile $\sigma^*$ \\
$A$ & Alice's payoff matrix \\
$B$ & Bob's payoff matrix \\
\hline
$D_0$ & Network difficulty threshold for at least one expected solution to the Bounded PoW problem \\
$c$ & Number of quantum-state opportunities for Alice and Bob \\
$\kappa$ & Grover iteration budget \\
$S_0$ & Alice's and Bob's pure strategy set for a single quantum measurement \\
$S_0^c$ & Alice's and Bob's pure strategy set \\
$\mathbf{k}^{(j)}$ & $j$th pure strategy of~$S_0^c$ in lexicographic ordering \\
$\alpha/\beta$ & Row and column indices of Alice's and Bob's payoff matrices \\
$X_{\text{A}}^i,\, X_{\text{B}}^i$ & Outcome of Alice's / Bob's $i$th Grover state measurement \\
$A_i / B_i$ & Alice's / Bob's payoff matrix for their $i$th quantum measurement \\
$A_i^{\text{meas}} / B_i^{\text{meas}}$ & Payoff matrix for Alice's/Bob's $i$th quantum measurement alone\\
$A_i^{\text{AQMS}} / B_i^{\text{AQMS}}$ & Payoff matrix for Alice's/Bob's $i$th quantum measurement from Sattath's AQMS \\
$\mathcal{A}_i / \mathcal{B}_i$ & Cumulative Grover iterations performed by Alice / Bob \\
$\ell_{\text{A}} / \ell_{\text{B}}$ & Index of Alice's / Bob's most recent measurement before the other's $i$th \\
$S_0'$ & Sampled subset of $S_0$ used for computation \\
$n_q$ & Number of sampled Grover iteration counts from $S_0$ \\
$(S_0^c)'$ & Reduced pure strategy set of Alice and Bob \\
$n_s$ & Number of simulated days \\
$b_i$ & Number of blocks mined on simulated day~$i$ \\
$f_i$ & Number of forks observed on simulated day~$i$ \\
$P_{95}/P_{99}$ & 95th/99th percentiles \\
$P_{> \nicefrac{1}{3}}$ & Empirical probability of $p_\text{stale} > \nicefrac{1}{3}$ on any given day \\
$P_\text{det}$ & Empirical detection probability \\
$q_i$ & p-value for simulated day~$i$ under null hypothesis that no quantum miners are present \\
$F_{\text{A}}^i / F_{\text{B}}^i$ & Cumulative probability that Alice / Bob fails their first $i$ measurements \\
\end{tabular}
\end{ruledtabular}
\end{table}
\clearpage
\section{Background}
\label{sec:background}

In this section, we provide the essential background to frame the context of our research.  We begin with an explanation of Bitcoin mining: detailing what Bitcoin is, how the Bitcoin network operates, and the mechanics of mining Bitcoin. Then we discuss quantum mining: detailing how Grover's algorithm could be utilized to mine Bitcoin, the implications of quantum Bitcoin mining over the classical case, and introducing the aggressive quantum mining strategy. Finally, we discuss quantum races: summarizing the notion of quantum races~\cite{LRS19} and explaining the relevant game-theoretic concepts.
\subsection{Bitcoin protocol}
\label{subsec:bitcoinmining}

Here we give an overview of the Bitcoin protocol.
We begin by describing how the Bitcoin network operates by detailing how ``blocks'' store transactions and how the block headers are linked to form the blockchain. 
Then we describe what Proof-of-Work is and how it is quintessential to the mining process. 
Finally, we discuss temporary forks in the blockchain and their relation to the 51\% attack. 
\subsubsection{Structure of the blockchain}

We describe the structure of the blockchain.  
We begin with a brief overview of the Bitcoin network.
Then, we explain cryptographic hash functions, which 
play an important role in the Bitcoin protocol.
We then move on to a discussion on what a bitcoin is, and how the ownership of bitcoin is recorded. 
Next, we discuss how the chain of ownership is recorded in transactions, and then explain how these transactions are aggregated into a data structure called a block, discussing all components.
Finally, we discuss how the block headers link to form a digital ledger called the blockchain.

We begin by providing a brief overview of the network that Bitcoin operates on.
Bitcoin (capital ``B'' referring to the network/protocol) is structured as a decentralized peer-to-peer network~\cite{Ant14}.
In this network, the computers that participate, called nodes, are peers to each other and interconnect to form a mesh network, meaning that each node connects to multiple others with no hierarchy.
These nodes communicate and synchronize with each other by following the current Bitcoin protocol labeled~$V$,
which is the system of rules that governs how the Bitcoin network operates.
Whereas all nodes share the responsibility of routing information and data, some nodes can take on additional responsibilities.
Nodes that take on all possible responsibilities are called ``full nodes''.
A user interacts with the Bitcoin network to send and receive bitcoin currency through a wallet application.

Now we give an overview of cryptographic hash functions.
Cryptographic hash functions are an essential component to Bitcoin, used for secure identification and ensuring the integrity of data~\cite{Ant14, LanCaw20}.
\begin{definition}[Hash function~\cite{GGF17} ]
\label{def:hashfunction}
The $m$-bit hash of a bit string of arbitrary size (with $^*$ for wildcard notation) is the uniformly distributed mapping
\begin{equation}
\label{eqn:hashheader}
\operatorname{hash}:\{0,1\}^*\to\{0,1\}^m:s\mapsto\operatorname{hash}(s)
.\end{equation}
\end{definition}
\begin{remark}[Approved hash function~\cite{GGF17} ]
\label{remark:approvedhash}
An approved hash function such as $\operatorname{SHA-256}^2$ is efficiently computable, one-way (computationally difficult to invert), and collision-resistant (computationally difficult to find two inputs that map to the same output).  
\end{remark}
We define the partial inversion of a hash function as follows.
\begin{definition}[Partial inversion of a hash function ]
\label{def:partialinversion}
Given an m-bit hash function and a threshold $\Theta \in \{0, 1\}^m$, the partial inversion of the hash is a bit string $s \in \{0, 1\}^*$ such that 
\begin{equation}\label{eqn:partialinversion}
\operatorname{hash}(s) < \Theta
.\end{equation}
\end{definition}
\noindent In the Bitcoin protocol, the accepted hash function SHA-256 is double applied
\begin{equation}\label{eqn:doubleSHA}
    \operatorname{SHA-256}^2(s) := \operatorname{SHA-256}(\operatorname{SHA-256}(s))
\end{equation}
when hashing~\cite{Ant14, NIST15}.

Now we discuss what a bitcoin is from a conceptual standpoint.
At its core, a bitcoin (lowercase ``b'' referring to the currency) is not a physical or digital entity but a unit of account (an economic term for quantifying the value of goods or services) recorded on a public, decentralized ledger~\cite{LanCaw20}.
The ownership of bitcoin is defined by the right to control a collection of Unspent Transaction Outputs (UTXOs), which can conceptually be thought of as discrete, indivisible quantities of bitcoin that are recognized by the entire network as valid~\cite{Ant14}. 
The term ``indivisible'' means that a UTXO must be spent in full, and cannot be divided or partially spent~\cite{Ant14}.
\begin{definition}[UTXO~\cite{Ant14}]
\label{def:UTXO}
A UTXO is a 2-tuple of an integer representing an amount of bitcoin denominated in satoshis ($1\, \text{satoshi}=10^{-8}\,\text{BTC}$) and a byte stream (a sequence of bytes) of variable size called the locking script, which sets the conditions for that amount of bitcoin to be spent and also includes the length of the script itself.
\end{definition}
\noindent The conditions of the locking script are written in a simple Turing-incomplete scripting language, allowing for a plethora of requirements to be set for spending a UTXO~\cite{Ant14}.
However, the most common conditions involve requiring proof of ownership through cryptographic keys~\cite{Ant14}.
\begin{definition}[Private key~\cite{LanCaw20}]
\label{def:privatekey}
A private key is a 256-bit integer chosen uniformly at random.
\end{definition}
\begin{definition}[Public key~\cite{Ant14}]
\label{def:publickey}
A public key is derived from the private key using a cryptographic one-way function.
\end{definition}
\noindent Currently, the one-way function used by the Bitcoin protocol is elliptic curve cryptography on the secp256k1 curve~\cite{Ant14}.
\begin{definition}[Address~\cite{Ant14}]
\label{def:address}
An address is a human-readable string derived from a public key, specifying the destination for receiving bitcoin.
\end{definition}
\noindent Thus, ownership of bitcoin is conferred by controlling private-public key pairs and fulfilling the conditions of the locking scripts~\cite{Ant14}.
When bitcoin is exchanged or transacted, UTXOs are spent by the sender, and new UTXOs are created, for which only the receiver can unlock with their key pair(s).

Here we explain the purpose and structure of Bitcoin transactions.
Transactions record the creation and consumption of UTXOs, and also record the proof of ownership of bitcoin~\cite{Ant14}.
Each transaction consumes existing UTXOs as input, and creates new UTXOs as outputs. 
Each input references a UTXO to be spent and includes the transaction ID
\begin{equation}
\label{eqn:txid}
    \textsc{txid} := \operatorname{SHA-256}^2(\texttt{bytes}(x))
\end{equation}
where $\texttt{bytes}(x)$ is the byte representation of a transaction~$x$, the output index (specifying which UTXO from that transaction is being referenced), an unlocking script (satisfying the conditions for spending that UTXO), and a sequence number (for relative time-locks and enabling additional features)~\cite{Ant14}.
We now define a Bitcoin transaction formally in a way that differs slightly from Antonopoulos's definition~\cite{Ant14}
but is congruent.
\begin{definition}[Bitcoin Transaction ]
\label{def:transaction}
A Bitcoin transaction is a 2-tuple comprising a list of inputs and a list of outputs,
with each input being a 4-tuple referencing a specific UTXO and comprising
\begin{itemize}
\item[] a 256-bit transaction identification (\textsc{txid}),
\item[] a 32-bit index pointing to a specific output from a previous transaction,
\item[] a variable-size (up to 10\,000 bytes) unlocking script (\textsc{scriptSig}) that provides the necessary information for the locking script, and
\item[] a 32-bit sequence number (\textsc{nSequence}) for establishing complex time locks
and to set flags to enable additional transaction features,
\end{itemize}
and each output being a new UTXO created by the transaction.
\end{definition}
\noindent 
If the total value of the UTXOs being spent exceeds the amount being sent to the recipient(s), then a new UTXO is created in the output sending the change back to the sender~\cite{Ant14}.
Although the change is sent back to the sender, the value of the input UTXOs is typically not equal to the value of the output UTXOs; the difference represents the transaction fee that the sender agrees to pay as a rule of the protocol~\cite{Ant14}. 
Transactions are grouped and stored in a decentralized and distributed digital ledger, which is a digitally-maintained record of transactions that is shared, replicated, and synchronized across multiple network nodes~\cite{Ant14, LanCaw20}.
All full nodes in the network maintain a set of all unspent outputs from all transactions to construct transaction inputs quickly and to efficiently validate transactions to ensure immutability~\cite{Ant14}.
Similarly, all nodes except those running a lightweight version of the Bitcoin protocol maintain a partially-ordered set (poset) of unconfirmed transactions known as mempool, which we now define~\cite{Ant14}.
\begin{definition}[Mempool]
The mempool 
\begin{equation}
    \mathcal{M} := (T, \prec)
\end{equation}
is a poset where~$T$ is the set of all valid but unconfirmed Bitcoin transactions and $\prec$ is a binary relation on~$T$ such that for $x, y \in T$, $x \prec y$ iff (if and only if) the bitstring corresponding to the \textsc{txid} of~$x$ is in the inputs of~$y$. 
\end{definition}

Now we explain how blocks are structured and the nature of their components.
Bitcoin transactions are permanently recorded in a digital ledger by first grouping them into container data structures called blocks~\cite{Ant14, Nak08}, which are illustrated in Fig.~\ref{fig:blockstructure}.
In this figure, the Merkle tree is a binary tree where individual transactions are hashed at the leaves, arranged in an arbitrary order.
Parent nodes are formed by recursively hashing the concatenation of each pair of children until a single hash remains, called the Merkle root~$R$~\cite{Ant14}. 
\begin{figure}
    \centering
    \includegraphics[width=\textwidth]{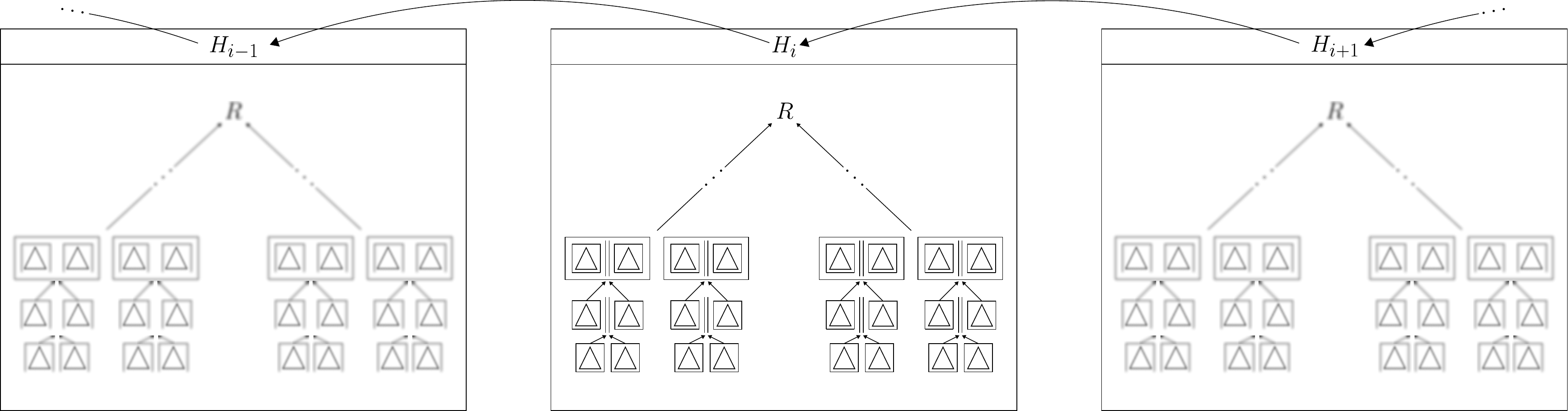}
    \caption{The structure of a block. 
    Blocks are depicted as squares that are partitioned into two rectangles. 
    The upper rectangle is labeled~$H_i$, which is the block header of the $i^\text{th}$ block, and the lower rectangle is the Merkle tree with root~$R$.
    Within the Merkle tree, a triangle ($\triangle$)  represents the hash of an individual transaction.
    A square ($\square$) represents the application of a hash function whose inputs are the contents inside the square, and $||$ denotes concatenation.
    The Merkle tree construction is shown for an even number of transactions.
    For an odd number of transactions, the last transaction is duplicated to make the number of transactions even.
    The arrows between blocks represent the hash of the previous block's header.}
    \label{fig:blockstructure}
\end{figure}
\begin{definition}[Block~\cite{Ant14} ]
\label{def:blockheader}
A block is a container data structure comprising a block header
\begin{equation}
\label{eqn:Hdef}
H:=(V,P,R,t,\tau,N)\in\{0,1\}^{32}\times \{0, 1\}^{256} \times \{0, 1\}^{256} \times \{0, 1\}^{32} \times \{0, 1\}^{32} \times \{0, 1\}^{32}
\end{equation}
which is an 80-byte sextuple containing metadata and a list of transactions aggregated into a Merkle tree for efficient verification of membership and summarization of the included transactions. 
The elements of the block header are 
\begin{itemize}
    \item[] a Bitcoin protocol version~$V\in\{0,1\}^{32}$, 
    \item[] the double $\operatorname{SHA-256}$ hash~$P \in \{0, 1\}^{256}$ of the preceding block's header~$H$, 
    \item[] a Merkle root~$R\in\{0,1\}^{256}$,
    \item[] a timestamp $t\in\{0,1\}^{32}$ that records the approximate Unix time when the block was created,
    \item[] a difficulty target~$\tau \in \{0, 1\}^{32}$ encoding a threshold that~$H$ must satisfy to be considered valid (see Def.~\ref{def:powcondition}), and
    \item[] an arbitrary nonce~$N\in\{0,1\}^{32}$ set by the block's creator. 
\end{itemize}
\end{definition}
\begin{remark}[Compression of the difficulty target~$\tau$~\cite{Ant14} ]
To reduce the size of the block header, the difficulty target~$\tau$ is stored as a 4-byte compressed encoding of a 256-bit string~$\tilde{\tau}$.
Let
\begin{equation}
\label{eq:bytesbig-endianorder}
b_{0,1,2,3}\in\{0, 1\}^8
\end{equation}
be the bytes of $\tau$ in big-endian order. 
The full 256-bit difficulty target is computed via
\begin{equation}
    \label{eqn:targetcompression}
    \tilde{\tau} := 2^{8(b_0 - 3)} \left(2^{16}b_1 + 2^8 b_2 + b_3\right).
\end{equation}
The value of~$\tilde{\tau}$ is updated every 2016 blocks according to the recurrence relation
\begin{equation}
    \label{eqn:difficultyupdate}
    \tilde{\tau}_{i+1} \gets \frac{t_{2016}}{20160\,\mathrm{min}}\tilde{\tau}_i,
\end{equation}
where~$t_{2016}$ is the time to mine the last 2016 blocks.
\end{remark}
\noindent The Merkle root~$R$ is stored in the block header to summarize the transactions included and enable efficient verification of membership, requiring only $\mathcal{O}(\log_2 n_T)$ computations for~$n_T$ transactions~\cite{Ant14}.
The previous block hash~$P$ in each block's header uniquely and unambiguously identifies the previous block, and serves as a back link in a linked-list structure, linking all blocks chronologically.

Now we explain what the blockchain is. 
The blockchain is structured as an ordered back-linked list of blocks that extends all the way to the first block, which is known as the genesis block~\cite{Ant14}.
The genesis block does not have a predecessor, and contains a trivial hash in its header.
This structure is illustrated in Fig.~\ref{fig:blockchain}.
\begin{figure}
    \centering
    \includegraphics[width=\textwidth]{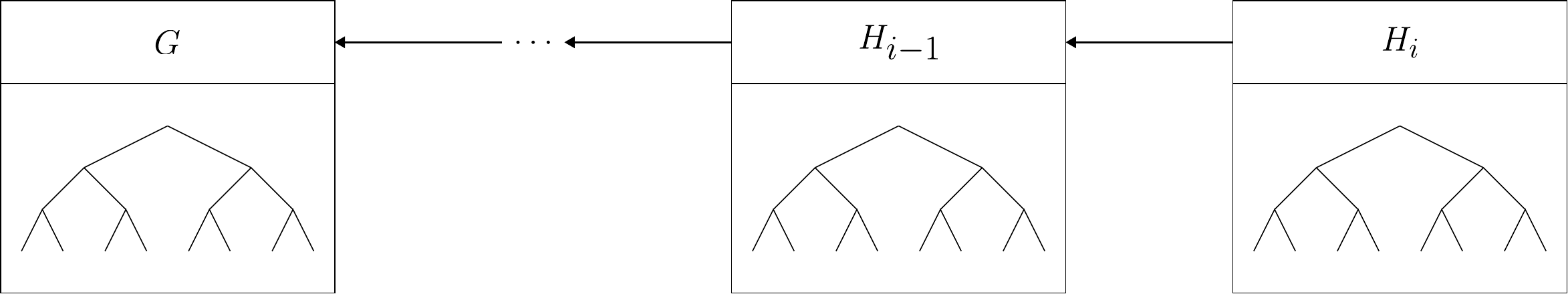}
    \caption{The blockchain depicted as a back-linked list of blocks.
    Each block is shown as a square divided into two rectangles. 
    The upper rectangle represents the block header~$H$ of the $i^\text{th}$ block, and the lower rectangle contains the Merkle tree of transactions, depicted as a tree structure.
    Back-links are depicted by arrows between block headers, representing the previous block hash stored in each block's header.
    The ellipses denote further preceding and succeeding blocks in the blockchain, and the block with header labeled~$G$ represents the genesis block, which has no predecessor.} 
    \label{fig:blockchain}
\end{figure}
\noindent Any alteration to a previous block would change its hash, invalidating the block hash of all subsequent blocks.
Correcting this invalidation requires significant computational effort, making the blockchain resistant to tampering~\cite{Ant14}.
Every node in the network maintains a local copy of the blockchain and individually verifies it against the copies held by other nodes to ensure consistency and maintain integrity~\cite{Ant14}.
In cases where nodes have blockchains that conflict with each other, each node independently determines which blockchain to consider authoritative by choosing the blockchain with the largest block height (longest length)~\cite{Ant14}. 
\begin{definition}[Block height~\cite{LanCaw20}]
\label{def:blockheight}
The block height is the number of blocks between a given block and the genesis block, which has a block height of 0. 
\end{definition}
\noindent As a result, this structure forms a decentralized and immutable digital ledger that records all Bitcoin transactions. 
\subsubsection{Proof of work}

Here we describe the mechanics of Proof-of-Work.
We begin by discussing consensus in the context of blockchains and the importance of achieving consensus among the decentralized network. 
Then, we explain the specifics of Proof-of-Work and how it functions to achieve consensus.
Afterwards, we describe the process of mining Bitcoin and how it relates to Proof-of-Work. 
Finally, we describe the role of the coinbase transaction and discuss how its structure is exploited when mining.

We explain the concept of consensus in the Bitcoin network and how it emerges from the protocol. 
In general, consensus refers to a general agreement among a group. 
We start by giving an accepted definition that describes how consensus is achieved in the Bitcoin network~\cite{Ant14}, which we refer to as operational consensus because it focuses on the mechanism rather than the concept of a general agreement itself.
\begin{definition}[Operational consensus ]
\label{def:consensusoperational}
Operational consensus is achieved when a majority of nodes in the network have the same blocks in the blockchain that they each validate as requiring the most total hashes.
\end{definition}
\noindent While this definition captures the mechanism by which consensus emerges, we also offer a definition that aligns more directly with the conceptual meaning of a general agreement.
\begin{definition}[Consensus agreement ]
\label{def:consensusconceptual}
A consensus agreement is a general agreement that all nodes in the network accept operational consensus.
\end{definition}
\noindent In the Bitcoin network, consensus emerges through a mechanism called Proof-of-Work (PoW)~\cite{Ant14}, which enables nodes to achieve a consensus agreement without a fixed election period or central coordination.
PoW enforces the rules of the Bitcoin protocol that determine which blockchain is most authoritative and which blocks can be appended to it.
At the same time, PoW also provides the necessary incentives
for nodes to participate in maintaining and extending the blockchain.

Now we discuss the technical details of PoW. 
PoW is a consensus mechanism that performs two functions~\cite{Ant14}.
One function is to establish the eligibility of blocks to be appended to the blockchain.
The second function is to enforce the verification of blocks and transactions by nodes.
To determine which blocks can be appended, participating nodes compete to solve a cryptographic problem based on the partial inversion of the $\operatorname{SHA-256}$ hash function.
As part of this process, the timestamp included in each block's header must satisfy a validity condition~\cite{Ant14}, which we now define.
\begin{definition}[Valid timestamp ]
\label{def:validtimestamp}
A valid timestamp~$\mathfrak{t}$ satisfies 
\begin{equation}
\label{eqn:validtimestamp}
\bar{t}_{11} < \mathfrak{t} < \bar{t} + 2\,\textnormal{hours}
\end{equation}
where $\bar{t}_{11}$ is the median timestamp of the last 11 blocks, and~$\bar{t}$ is the network-adjusted time, which is the median of a node's local system time and the times reported by all connected peers.
\end{definition}
\begin{definition}[PoW condition ]
\label{def:powcondition}
The PoW condition is
\begin{equation} \label{eqn:miningcondition}
    \operatorname{SHA-256}^2 \left(V \mathbin\Vert P \mathbin\Vert R \mathbin\Vert \mathfrak{t} \mathbin\Vert \tau \mathbin\Vert N\right) < \tilde{\tau},    
\end{equation}
where $\mathbin\Vert$ denotes concatenation, $R$ is a valid Merkle root computed from the ordered transactions included in the block, $\mathfrak{t}$ is a valid timestamp, and $\tau \in \{0, 1\}^{32}$ is the difficulty target.
\end{definition}
\noindent A block may be appended to the blockchain iff it satisfies the PoW condition.
The problem of finding a block header that satisfies the PoW condition is now formalized as a computational problem.
Given any set~$X$, the power set $2^X$ is the set containing all subsets of~$X$.
If~$X$ is a poset, we define $\mathcal{S}(X)$ as the set of all subsets whose elements are partially ordered according to the partial ordering of~$X$.
\begin{problem}[PoW problem ]
\label{prob:proofofwork}
Given a version~$V$, a previous block hash~$P$, a difficulty target~$\tau$, a mempool of valid transactions~$\mathcal{M}$, a Merkle root computation oracle
\begin{equation}
\label{eq:fR}
f_R: \mathcal S(\mathcal M) \to \{0, 1\}^{256}
\end{equation}
the network-adjusted time~$\bar{t}$, the median timestamp of the last 11 blocks~$\bar{t}_{11}$,  and a timestamp validation oracle
\begin{equation}
\label{eq:ft}
f_t^{\bar{t}, \bar{t}_{11}} : \{0, 1\}^{32} \to \{0 ,1\},
\end{equation}
find a timestamp~$t$, a sequence of transactions $T_x \in \mathcal{S}(\mathcal{M})$, and a nonce~$N$ such that the resultant block header~$H$ satisfies the PoW condition.
\end{problem}
\begin{definition}[Network difficulty~\cite{LanCaw20} ]
The network difficulty function
\begin{equation}
    \label{eqn:networkdifficulty}
    D(\tilde{\tau}) := \nicefrac{2^{256}}{\tilde{\tau}}
\end{equation}
represents the expected number of hashes required to find a block header satisfying the PoW condition, where $\tilde{\tau}$ is the full 256-bit difficulty target.
\end{definition}
\begin{remark}[Rationale for~$D$]
As the $2^{256}$ outputs of $\operatorname{SHA-256}$ are uniformly distributed (Definition~\ref{def:hashfunction}), the probability that an output of $\operatorname{SHA-256}$ is less than a given~$\tilde{\tau}$ is $\nicefrac{\tilde{\tau}}{2^{256}}$. 
Thus, the expected number of hashes computed before finding a header that satisfies the PoW condition is the network difficulty function (\ref{eqn:networkdifficulty}).
Although we define the network difficulty~$D$ as a function of $\tilde{\tau}$, we refer to the evaluated value~$D(\tilde{\tau})$ simply as~$D$.
\end{remark}
\noindent
In practice, solving Problem~\ref{prob:proofofwork} by iterating over all $2^{320}$ possible values of the bitstrings $(f_R(T_x), t, N)$ is infeasible due to the overhead required for recomputing the Merkle root and validating the timestamp. 
Instead of solving Problem~\ref{prob:proofofwork} all at once, nodes fix all bitstrings except the nonce and execute a brute-force search over the $2^{32}$ possible nonce values.
This process is repeated over multiple instances of this bounded problem until a solution to Problem~\ref{fig:powproblem} is found.
\begin{problem}[Bounded PoW problem ]
\label{prob:boundedproofofwork}
Given a version~$V$, a previous block hash~$P$, a Merkle root~$R$, a valid timestamp~$\mathfrak{t}$, and a difficulty target~$\tau$, find a nonce~$N$ such that the resultant~$H$ satisfies the PoW condition.
\end{problem}
\begin{figure}
    \centering
    \includegraphics[width=\textwidth]{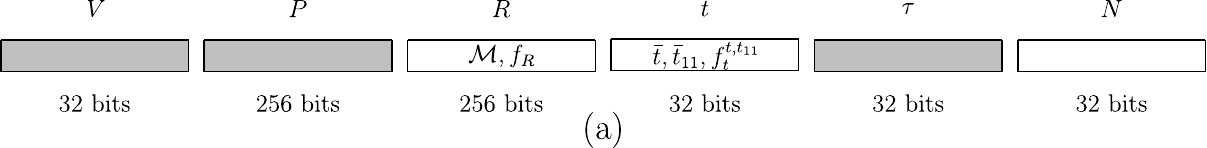}
    \includegraphics[width=\textwidth]{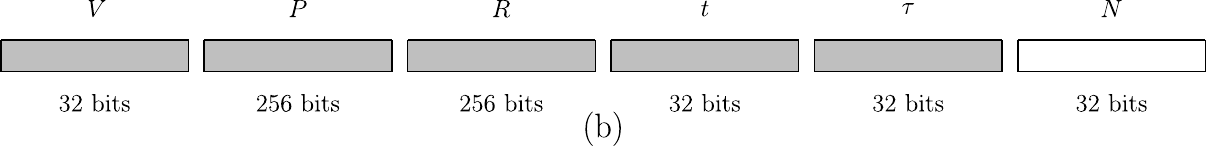}
    \caption{ A visualization of finding a proof-of-work for (a) the PoW problem and (b) the bounded PoW problem.
    The concatenated block header is represented by six rectangles.
    A rectangle represents a specific bitstring in the concatenation of the block header.
    Each rectangle is labeled above by the symbol of the element in the block header.
    The length of the bitstring for each label is indicated below the rectangle.
    A filled rectangle represents a bitstring whose value remains constant.
    An empty rectangle denotes a bitstring whose value is computed by the node finding a proof-of-work.
    In (a), for the Merkle Root~$R$, the mempool~$\mathcal{M}$ and the oracle~$f_R$ are inputs, as are $\bar{t}$, $\bar{t}_{11}$, and $f_t^{\bar{t}, \bar{t}_{11}}$ for the timestamp~$t$.
    In (b), only the nonce~$N$ is computed. 
    The nonce is represented as an empty rectangle in both figures, as it does not require any inputs.}
    \label{fig:powproblem}
\end{figure}
\noindent The probability that a valid nonce exists within a search space of size $2^{32}$ is approximately
\begin{equation}
    \label{eqn:validnonceprob}
    \nicefrac{2^{32}}{D}=\nicefrac{\tilde{\tau}}{2^{224}}.
\end{equation}
The solution to Problems~\ref{prob:proofofwork} and~\ref{prob:boundedproofofwork} is called the proof-of-work (lowercase).
Whereas Antonopoulos uses ``Proof-of-Work" to refer to both the mechanism and any solution to Problem~\ref{prob:proofofwork}~\cite{Ant14}, we adopt the lower case to distinguish a solution from the consensus mechanism.
A proof-of-work demonstrates that a node has expended sufficient computational effort based on the PoW condition.
Satisfying the PoW condition grants that node the right to append a block to the blockchain~\cite{Ant14}.
In return, the node receives a reward for its efforts, which encourages other nodes to participate in the competition.

We describe the process of mining Bitcoin and outline the responsibilities of nodes that participate as miners.
Bitcoin mining refers to the process of extending the blockchain by constructing candidate blocks and searching for a proof-of-work~\cite{Ant14}.
Bitcoin miners, the nodes that participate in this process, are responsible for~\cite{Ant14}
\begin{itemize}
    \item[] validating new blocks and transactions when they are received by broadcast,
    \item[] placing validated transactions into their mempool,
    \item[] constructing the coinbase transaction, which pays the reward to the miner,
    \item[] selecting transactions from their mempool and constructing a Merkle tree with root $R$, which includes the coinbase transaction,
    \item[] constructing a candidate block with values for $V$, $P$, and $\tau$ using both public sources of information and local data,
    \item[] solving Problem~\ref{prob:proofofwork} to find a proof-of-work and form a valid block header, and
    \item[] broadcasting the newly-mined block to the network. 
\end{itemize}
\noindent As a reward for extending the blockchain, a miner receives the transaction fees of all included transactions, and a fixed amount of newly minted bitcoins, with this fixed amount being halved approximately every four years~\cite{Ant14, Nak08}.
As a result, miners choose transactions from their mempool commencing with the highest transaction fee available when forming a candidate block~\cite{LanCaw20}.
Additionally, miners are responsible for independently updating the value of~$\tau$ every 2016 blocks (\ref{eqn:difficultyupdate}).
As all miners compute the value of~$\tau$ based on the publicly available blockchain, all miners will compute the same result, and any block with an incorrect value of~$\tau$ in its header is considered invalid~\cite{Ant14}.
\begin{remark}[Average blocks per day~\cite{Ant14}]
\label{remark:averageblocks}
On average, the Bitcoin blockchain produces 144 blocks per day. 
\end{remark}

We describe the role and structure of the coinbase transaction.
The mining reward is paid to the miner by what is known as the coinbase transaction~\cite{Ant14}.
The coinbase transaction is the first transaction in a block and is constructed by a miner as part of the mining process. 
Unlike other transactions, the coinbase transaction does not consume any UTXO as input.
Instead, the coinbase transaction includes trivial data in place of the transaction input, and the transaction output is constructed by the miner so that the funds go to that miner's wallet. 
Importantly, the unlocking script is replaced by coinbase data, which is a field containing arbitrary data that must contain between 2 and 100 bytes.
For blocks with version 2 or greater, the first few bytes of this field are reserved to contain the block height~\cite{Ant14}.

We describe how miners exploit the structure of the coinbase transaction to extend their search space when mining.
The 32-bit nonce within the block header allows miners to quickly generate up to $2^{32}$ distinct block headers for each candidate block~\cite{Ant14}. 
However, because the network computes hashes at an extremely high rate (primarily due to specialized mining devices called Application-Specific Integrated Circuits, or ASICs), the mining difficulty has been adjusted such that this limited nonce space is often exhausted before a proof-of-work is found.
Rather than waiting for another miner to succeed and then reset the competition, miners require a way to seed new search spaces once the standard range of nonces is exhausted. 
Although adjusting the timestamp or reordering the transactions in the Merkle tree can alter the block header, these methods have limitations: only a small range of timestamps are considered valid, and reordering transactions requires a full recomputation of the Merkle tree.
The practical solution has been to exploit the structure of the coinbase transaction. 
Miners utilize the 2 to 100 bytes of arbitrary data permitted in the coinbase field to seed new search spaces in a technique called \texttt{extraNonce}.
This technique requires recomputing the left-most branch of the Merkle tree for each iteration of the \texttt{extraNonce}. 
Although \texttt{extraNonce} is not part of the Bitcoin protocol, it is widely adopted by miners in practice.
\subsubsection{Temporary forks and the 51\% attack}

Now we discuss what temporary forks are and how they relate to the 51\% attack. 
We begin with an explanation of what a temporary fork is, and how they occur both naturally and intentionally.
Then, we introduce the so-called ``propagation parameter", and discuss how it can be used for modelling forks in the blockchain.
Finally, we discuss the 51\% attack and how it relates to the rate at which forks appear in the blockchain.

We explain what temporary forks are and how they arise in the blockchain.
In the blockchain, temporary forks occasionally occur when multiple blocks are mined nearly simultaneously~\cite{Ant14}.
These blocks are broadcast through the network, and miners will mine atop of the block they receive first, causing the network to temporarily lose consensus.
In this situation, the temporary forks form a directed acyclic graph (DAG), which is illustrated in Fig.~\ref{fig:temporaryfork}.
\begin{figure}
    \centering
    \includegraphics[width=0.9\textwidth]{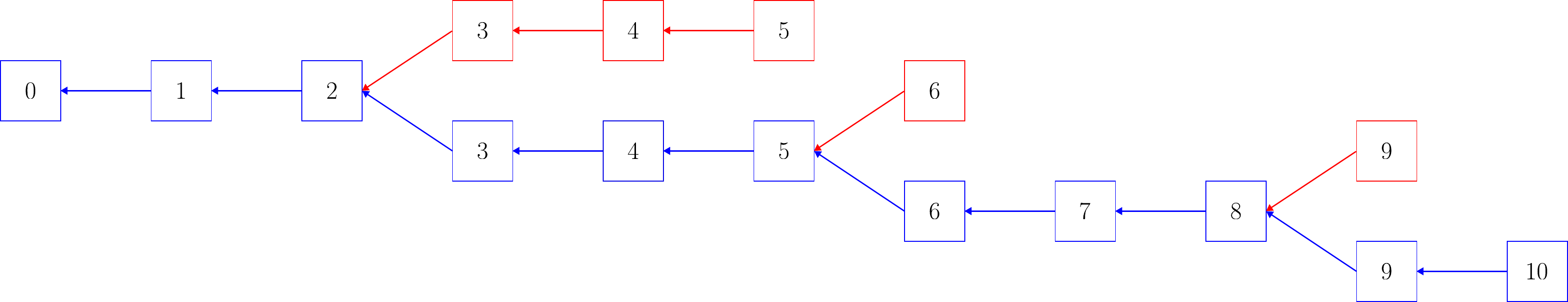}
    \caption{ The blockchain represented as a DAG during a temporary fork.
    A vertex ($\square$) represents a block, each labeled with its block height.
    Directed edges represent the block hash of preceding blocks.
    A blue square (\textcolor{blue}{$\square$}) indicates a block that has become part of the longest blockchain.
    A red square (\textcolor{red}{$\square$}) indicates a block that is not part of the longest blockchain.
    Each vertex (except the genesis block with height 0)
    has one outgoing edge to a preceding block, but can have multiple incoming edges from succeeding blocks during a temporary fork.
    This figure represents a temporary fork because not all nodes have the same symbol.}
    \label{fig:temporaryfork}
\end{figure}
\noindent Eventually, one branch of the temporary fork will achieve more cumulative proof-of-work than the others, becoming the authoritative branch of the blockchain~\cite{Ant14}.  
Blocks in the losing branches become ``stale''.
\begin{definition}[Stale block~\cite{Ant14} ]
A stale block is a block that was successfully mined but has not been included in the longest branch of the blockchain.
\end{definition}
\noindent Miners who have mined blocks that become stale do not receive any mining rewards~\cite{LanCaw20}.
Temporary forks can occur naturally as a result of network effects~\cite{Ant14} as depicted in Fig.~\ref{fig:temporaryfork_network}.
However, temporary forks can also be created deliberately by a consortium of miners (a singlet or a collection of miners exercising control over a cumulative mining power) who could adopt a strategy to increase their expected revenue.
One such strategy is known as the selfish mining strategy, where a consortium temporarily withholds the broadcast of 
newly-mined blocks to gain an advantage~\cite{EyaSir18}.
\begin{figure}
    \centering
    \includegraphics[width=0.75\textwidth]{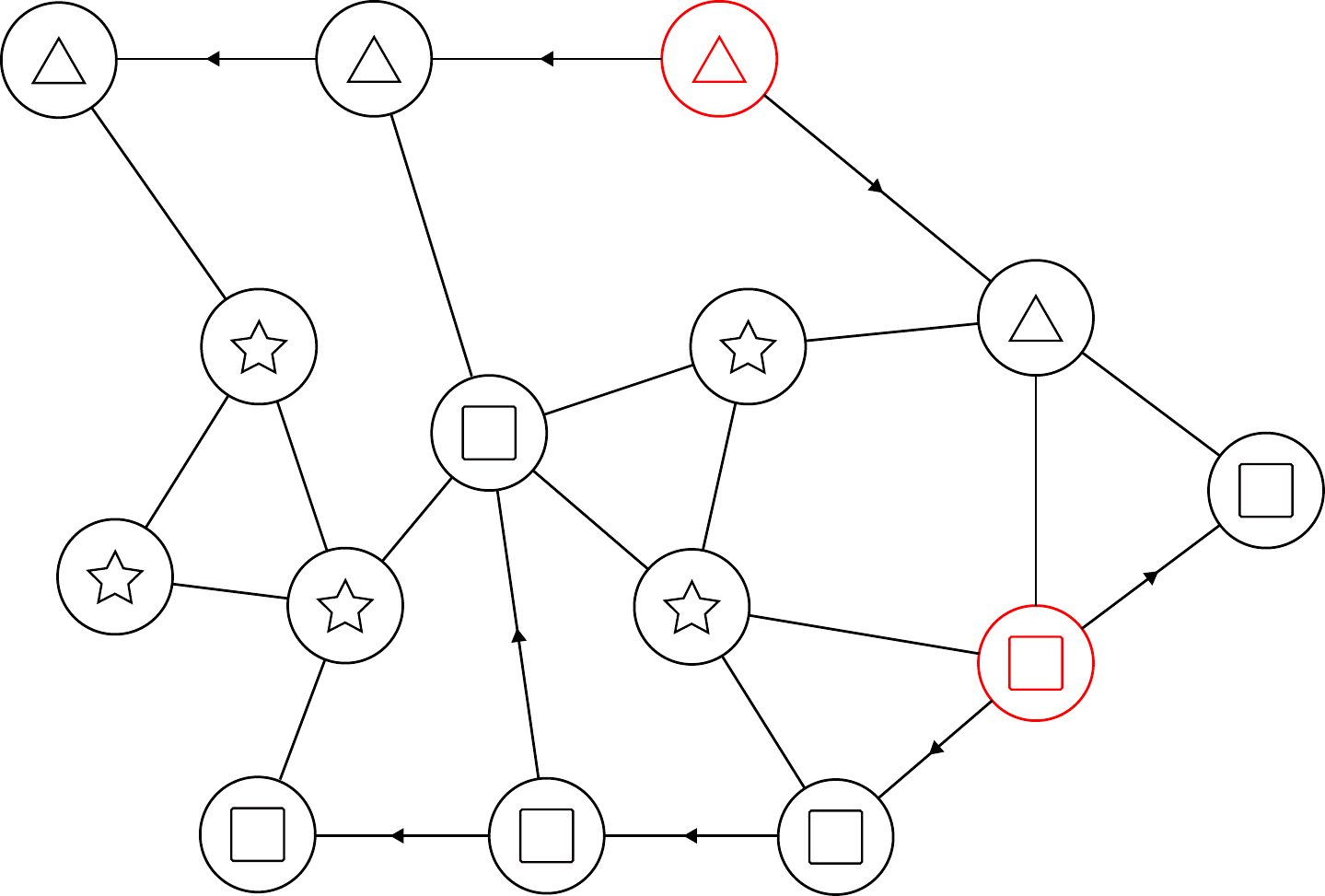}
    \caption{ A temporary fork occurring due to near-simultaneous block discovery. 
    A vertex ($\bigcirc$) represents a node in the Bitcoin network.
    A red vertex (\textcolor{red}{$\bigcirc$}) represents a node that has just mined a new block.
    Each vertex contains either a star ($\bigstar$), representing the state of the network before the fork, or a triangle ($\triangle$) or a square ($\square$) representing one of the new branches created by the temporary fork.
     An arrow within an edge indicates the broadcast of a new block to a peer.
    An edge without an arrow represents a peer-to-peer connection where no new block is being broadcast.}    
    \label{fig:temporaryfork_network}
\end{figure}

Now we discuss the propagation parameter.
To capture the complexities that determine the outcome of temporary forks, we introduce the propagation parameter~$\gamma$.
\begin{definition}[Propagation parameter~\cite{EyaSir18, SSZ17}]
\label{def:propparameter}
The propagation parameter~$\gamma$ is the fraction of honest miners who will mine on top of a given miner's block in the event of a temporary fork.
\end{definition}
This parameter abstracts the plethora of factors that influences the resolution of a temporary fork such as the propagation time, network architecture, and unique strategies employed by miners that influence block announcements~\cite{Sat20, EyaSir18, SSZ17}.
However, our use of~$\gamma$ differs from the prior definition~(Def.~\ref{def:propparameter}), as we adopt Sattath's definition~\cite{Sat20}.
\begin{definition}[Sattath's propagation parameter~\cite{Sat20}]
\label{def:sattathspropparameter}
Sattath's propagation parameter~$\gamma_\text{S}$ is the probability that a designated miner will successfully extend their branch of the blockchain first.
\end{definition}

We describe the 51\% attack~\cite{Ant14, Sat20}
formally and explain the 51\% attack threshold and the relative hash rate, which quantifies the authority of a consortium of miners.
\begin{definition}[Hash rate~\cite{LanCaw20} ]
\label{def:hashrate}
The hash rate~$q$ is the number of hashes that can be computed per second by a consortium comprising all miners in the network,
and the fractional hash rate~$q_\text{c}$ is the proportion of the hash rate that can be performed by a given consortium.
\end{definition}
\begin{remark}[Block arrival times~\cite{BKKT18} ]
\label{remark:arrivaltimes}
The block arrival times in Bitcoin are modelled as a Poisson process, meaning that they follow exponential distribution with a rate parameter of $\lambda=\nicefrac1{600}$ seconds.
However, this model is only valid if the hash rate and network difficulty remain constant over the period being modelled. 
In practice, the hash rate and network difficulty evolve as functions of time, as miners are constantly joining and leaving the Bitcoin network and the technology for mining Bitcoin improves.
Thus, the block arrival times are more accurately modelled as a nonhomogeneous Poisson process with a rate parameter~$\lambda$ that is proportional to the ratio of hash rate to the network difficulty.
\end{remark}
\begin{definition}[The 51\% attack~\cite{Sat20} ]
The 51\% attack means that some consortium controls over 50\% of the hash rate of the network,
enabling the consortium to double-spend coins,
reverse transactions,
or create their own branch of the blockchain with a high probability of success. 
\end{definition}
\begin{remark}[Duration of a 51\% attack~\cite{Ant14}]
\label{remark:51duration}
Executing a 51\% attack is not instantaneous, as the consortium must mine many consecutive blocks faster than the rest of the Bitcoin network.
In other words, the conditions enabling a 51\% attack must hold for an extended period, rather than at a single moment.
\end{remark}
\noindent
Although the terminology is 
51\%\ attack,
any percentage above 50\%\ suffices for this attack~\cite{Ant14}.
In fact, such an attack can be attempted with a smaller percentage of the hash rate, but the probability of success is reduced as other honest miners control the generation of some blocks with their relative hash rate.
The 51\% threshold is simply the level at which the probability of success approaches certainty. 
Consortiums with a larger fractional hash rate can deliberately create longer temporary forks, meaning that they can control and invalidate more blocks than smaller consortiums. 
The only known instance of a consortium reaching this threshold occurred in 2014, when the GHash.io consortium temporarily breached a fractional hash rate of 50\%, causing widespread panic among the Bitcoin community~\cite{guardian2014}.
In response, GHash.io established a reduced fractional hash rate of 
\begin{equation}
    \label{eqn:ghashfrachashrate}
    q_c=40\%,
\end{equation}
and committed to maintaining their fractional hash rate below this level.

We explain what the stale rate of the network is and discuss how it relates to the 51\% attack~\cite{Sat20}.
The 51\% attack threshold is not static and depends on a key property of the network known as the stale rate. 
\begin{definition}[Stale rate~\cite{Sat20}]
\label{def:stalerate}
The stale rate~$p_\text{stale}$ is the ratio of the number of blocks outside the longest blockchain to the total number of blocks in all branches of the blockchain.
\end{definition}
\begin{remark}[Binomial distribution of stale blocks~\cite{Ant14}]
    \label{remark:binomial}
    Since block arrivals follow a Poisson process (Remark~\ref{remark:arrivaltimes}), block events are both memoryless and independent of each other.
    Therefore, each block independently results in a stale block with probability~$p_\text{stale}$. 
    Thus, the number of stale blocks out of any fixed number of blocks is approximately binomially distributed.
\end{remark}
\noindent Importantly, the value of the stale rate of the network directly impacts the effective threshold for a 51\% attack.
A consortium of miners with a fractional hash rate~$q_\text{c}$ could execute a 51\% attack with certain probability if~\cite{Sat20}
\begin{equation}
\label{eqn:stalerate51attack}
q_\text{c} > \frac{1 - p_\text{stale}}{2 - p_\text{stale}}
\end{equation}
with the right hand side being effective threshold to execute a 51\% attack with certain probability.
Classically, the 51\% attack is not an important security concern as it is infeasible. 
A study in 2018 estimated a network stale rate of $p_\text{stale} \approx 0.24\%$~\cite{SSJ+18}, corresponding to approximately zero forks per day on average and meaning that a consortium will execute a 51\% attack if $q_c \gtrapprox 0.499$.  
However, the risk to Bitcoin's security would become far more potent if factors arise that significantly increase the stale rate of the network~\cite{Sat20}.
\subsection{Quantum mining}
\label{subsec:quantummining}

Now we describe how quantum mining works. 
First, we discuss the nature of Grover's algorithm and its generalization to multiple marked elements and how it could be applied to mining Bitcoin.
Then we explore the contrasts between classical and quantum bitcoin mining. 
Finally, we discuss aggressive versus peaceful quantum mining strategies.
\subsubsection{Grover's algorithm}

Now we discuss how Grover's algorithm and its generalization could be employed to mine Bitcoin.
We begin by defining what an algorithm is.
Then, we explain the problem that Grover's algorithm solves, namely, the probabilistic unstructured search problem, and then describe the algorithm. 
Next, we show how the problem of finding a proof-of-work reduces to an instance of this problem, allowing the use of Grover's algorithm for mining Bitcoin. 
Finally, we discuss the feasibility of employing Grover's algorithm to mine Bitcoin.

We first define what an algorithm is both in the classical and quantum settings. 
An algorithm is the means by which a computer solves a computational problem and comprises three parts: an input, an output, and a procedure that transforms the input into the output~\cite{San25}.
A procedure is a finite set of formally defined instructions interpretable by a computer.
A quantum algorithm is an algorithm whose procedure uses quantum logic.

We now turn to a specific computational problem known as probabilistic unstructured search.
This problem formalizes the task of finding a ``marked" item in a large unsorted dataset for which no underlying structure can be exploited. 
Since we ultimately wish to relate this problem to mining Bitcoin, we assume that the dataset has size~$2^{32}$~\cite{Grov96, Ant14}.
The condition that identifies a marked item is defined using a function that returns 1 for a single input and 0 otherwise.
We formally define the 32-bit Boolean oracle inspired by seminal work on the exact quantum lower bound for Grover's problem~\cite{DohHoy09}.
\begin{definition}[32-bit Boolean oracle]
\label{def:booleanoracle}
A 32-bit Boolean oracle is 
\begin{equation}
\label{eqn:booleanoracle}
f : \{0, 1\}^{32} \to \{0, 1\} : x \mapsto 
\begin{cases}
    1 & \text{if } x \in M \\
    0 & \text{else}
\end{cases}
\end{equation}
for some subset 
\begin{equation}
    \label{eqn:solutionsubset}
    \{0, 1\}^{32} \supseteq M :=  \{x_0 : f(x_0)=1\}
\end{equation}
 being the subset of marked inputs,
 i.e., all~$x_0$
that map to~$1$ under the action of $f$~(\ref{eqn:booleanoracle}).
\end{definition}
\begin{remark}
\label{remark:multimarked}
As we undertake a search for multiple marked elements~(\ref{eqn:solutionsubset}),
Grover's algorithm, which technically is designed for finding the only marked element, is not the appropriate name for the generalization of Grover's algorithm to multiple marked elements.
Building on Grover's and Shor's algorithm,
Brassard, H{\o}yer, Mosca and Tapp presented an algorithm~\cite{BHMT02} for
a quantum search for multiple marked elements.
In both cases, the budget,
i.e., the number of iterations in the search loop, is unbounded,
whereas we explicitly treat the finite-budget case as the quantum miners have limited time to execute the search.
We refer to our variant of quantum search as a finite-budget quantum search (FQS).
\end{remark}
\begin{problem}
\label{prob:unstructuredsearchproblem}
[Probabilistic unstructured search problem~\cite{BHMT02}]
Given probability~$p$ and a 32-bit Boolean oracle~$f$~(\ref{eqn:booleanoracle}),
and a promise of an unstructured search
(each input is equally likely to yield~$1$ at the output), determine, 
with probability~$p$, 
an input $x_0 \in M$ that is mapped to~$1$ by~$f$.
\end{problem}
\noindent Problem~\ref{prob:unstructuredsearchproblem} can be solved by the brute-force search algorithm, which is guaranteed to succeed within $2^{32}$ calls to the Boolean oracle~\cite{BHMT02}.

We now explain our FQS algorithm,
which solves Problem~\ref{prob:unstructuredsearchproblem} with no more than
\begin{equation}
\label{eq:nomorethan}
\left\lfloor \frac{\pi}{4} \sqrt{2^{32} / |M|}\right\rfloor
\end{equation}
queries to a quantum-phase oracle~\cite{KLM06}.
We define the 32-qubit quantum-phase oracle.
We consider a Hilbert space 
\begin{equation}
\label{eq:H32}
\mathscr{H} := \mathscr{H}_{2^{32}},
\end{equation}
where the subscript denotes the dimension of the space, equipped with an orthonormal computational basis $\{\ket{x} :  x \in \{0, 1\}^{32}\}$.
We denote uni$\left(\mathscr{H}\right)$
and her$\left(\mathscr{H}\right)$
as the set of unitary and hermitian operators on~$\mathscr H$,
respectively.
\begin{definition}[32-qubit quantum-phase oracle~\cite{KLM06}]
\label{def:quantumoracle}
Given a 32-bit Boolean oracle~$f$, the corresponding 32-qubit quantum-phase oracle is
\begin{equation}
\label{eqn:quantumoracle}
\operatorname{uni}
\left(\mathscr{H}\right) \ni U_{f} :\mathscr{H} \to \mathscr{H}: \sum_x \alpha_x \ket{x} \mapsto \sum_x \alpha_x (-1)^{f(x)}\ket{x}.
\end{equation}
\end{definition}
\noindent The phase flip introduced by the quantum-phase oracle ``marks" the solution by flipping the sign of the probability amplitude associated with the corresponding basis state~\cite{KLM06}.
The FQS algorithm exploits this marking by iteratively amplifying the probability amplitude of the basis states corresponding to solutions.
This amplitude amplification is achieved by iteratively applying a quantum operator known as a Grover iterate. 
We first define the quantum operators that serve as the components of the Grover iterate and then define the Grover iterate itself.
\begin{definition}[$32$-qubit Hadamard gate~\cite{KLM06}]
The $32$-qubit Hadamard gate $\mathcal{H}^{\otimes 32}$ is the tensor product of~$32$ single-qubit Hadamard gates
\begin{equation}
\label{eqn:hadamardgate}
\mathcal{H}=\frac1{\sqrt{2}}
\begin{bmatrix}
1 & 1 \\
1 & -1
\end{bmatrix}
\in\operatorname{uni}\left(\mathscr{H}_2\right)\cap\operatorname{her}\left(\mathscr{H_2}\right).
\end{equation}
\end{definition}
\begin{definition}[32-qubit Householder reflection~\cite{ANBH+22}]
The $32$-qubit Householder reflection is 
\begin{equation}
    \label{eqn:phaseshiftoperator}
    \operatorname{uni}(\mathscr H) \cap \operatorname{her}(\mathscr H) \ni U_0 := 2 \left(\ket{0}\bra{0}\right)^{\otimes 32} - \mathds1,
\end{equation}
which maps
\begin{equation}
\label{eqn:Householdermaps}
\ket{0}^{\otimes 32} \mapsto \ket{0}^{\otimes 32},\,
\ket{x} \mapsto - \ket{x} \forall x \neq 0.
\end{equation}
\end{definition}
\begin{definition}[32-qubit Grover iterate~\cite{KLM06}]
\label{def:groveriterate}
The 32-qubit Grover iterate is
\begin{equation}
\operatorname{uni}(\mathscr H) \ni G := \mathcal{H}^{\otimes 32}U_{0}\mathcal{H}^{\otimes 32}U_{f}. 
\end{equation}
\end{definition}
\noindent We now outline the FQS algorithm (Remark~\ref{remark:multimarked}) for the 32-qubit multi-solution case, utilizing the Grover iterate.
\begin{definition}[Quantum measurement]
\label{def:qmeas}
A quantum measurement is the destructive mapping
\begin{equation}
\label{eqn:measurement}
\text{mea}:\mathscr{H}\to\{0,1\}^{32}:
\underbrace{\sum_{x=0}^{2^{32}-1}\bra{x}\ket{\psi}\ket{x}}_{\ket{\psi}}
\mapsto x
\end{equation}
in the computational basis, where each label~$x \in \{0, 1\}^{32}$ is an outcome with probability~$\left|\bra{x}\ket{\psi}\right|^2$.
\end{definition}
\begin{remark}[Measurement timing]
\label{remark:timing}
We treat measurement as instantaneous.
\end{remark}
We now introduce the 32-qubit version of the FQS algorithm.
Henceforth,
we use the notation
\begin{equation}
\label{eqn:[x]}
[x] := \{1, \dots, x\}, \;\;[x]_0 := \{0\} \cup [x]
\end{equation}
being all natural numbers from~1 (or~0, respectively) up to~$x$ inclusive.
Now we present Alg.~\ref{alg:grover}, which formalizes the mapping that yields
\begin{equation}
\label{eqn:Gmapping}
x\gets\operatorname{mea}\left(G^{k}\mathcal{H}^{\otimes32}\ket0^{\otimes32}\right)
\end{equation}
with~$x$ being a marked input~$x_0$~(\ref{eqn:solutionsubset}) to the Boolean oracle~(\ref{eqn:booleanoracle}) with probability that depends on the number of Grover iterations~$k$ performed before the quantum measurement~(Def.~\ref{def:qmeas}).
\begin{algorithm}[H]
\caption{FQS algorithm (32-qubit)~\cite{Grov96, BHMT02, KLM06}}
\label{alg:grover}
\begin{algorithmic}[1]
\Require $G \in \operatorname{uni}(\mathscr{H})$ \Comment{32-qubit Grover iterate} 
\Statex \hspace{\algorithmicindent} $k \in [51471]_0$ \Comment{The number of Grover iterations}
\Ensure $x \in \{0, 1\}^{32}$ \Comment{Measured bitstring} 
\Procedure{}{}
\State $\mathscr{H} \ni \ket{\psi} \gets G^k \mathcal{H}^{\otimes 32} \ket{0}^{\otimes 32}$ \Comment{Apply Grover operator~$k$ times on superposition of basis states}
\State \Return $x \gets \operatorname{mea}(\ket{\psi})$ \Comment{Measure the resultant state and return the basis state label with probability~$|\bra{x}\ket{\psi}|^2$}
\EndProcedure
\end{algorithmic}
\end{algorithm}
\begin{definition}
\label{def:kopt}
The optimal number~$k_\text{opt}$ of Grover iterations for the 32-qubit FQS algorithm is the least number of Grover iterations that maximizes the success probability;
i.e., the maximum probability for which the returned label is correctly the inverse of~1 for the marked element.
\end{definition}
\noindent The FQS algorithm returns a valid solution with a probability 
\begin{equation}\label{eqn:groversuccess}
p(k, |M|) := \sin^2\left[2(k + 1/2)\arcsin\left(\sqrt{\nicefrac{|M|}{2^{32}}}\right)\right],
\end{equation}
that scales with the number of Grover iterations~$k$ performed before measurement and the number of marked elements~$|M|$~\cite{NieChu10, BHMT02}.
In the context of mining Bitcoin, Eq.~(\ref{eqn:groversuccess}) is rewritten in terms of the network difficulty~$D$ by substituting
\begin{equation}
    \label{eqn:Dquantum}
    D=\nicefrac{2^{32}}{|M|},
\end{equation}
giving
\begin{equation}
    \label{eqn:groversuccessbitcoin}
    p(k)=p(k,D) := \sin^2\left[2(k + 1/2)\arcsin\left(\sqrt{\nicefrac1{D}}\right)\right],
\end{equation}
where we adopt the shorthand notation~$p(k)$, as~$D$ remains fixed for the values considered in this work.
Similarly, the optimal number of Grover iterations to perform, 
\begin{equation}
    \label{eqn:koptbitcoin}
    k_\text{opt}(D) := \left\lfloor \frac{\pi}{4} \sqrt{D}\right\rfloor,
\end{equation}
is also a function of the network difficulty~$D$ in the context of Bitcoin.
Since each Grover iteration queries the quantum-phase oracle (Def.~\ref{def:quantumoracle}) exactly once, this is also the optimal number of queries to perform.
However, for simplicity, we refer to the evaluated value~$k_\text{opt}(D)$ simply as~$k_\text{opt}$ and specify the value of~$D$ when necessary.
\begin{remark}[Zero Grover iterations]
\label{remark:zeroiterations}
The optimal number of Grover iterations is zero if
\begin{equation}
\label{eq:k0opt}
|M| > \frac{\pi^2}{16} \times 2^{32} \approx 2.65 \times 10^{9} \impliedby \frac{\pi}{4} \sqrt{\frac{2^{32}}{|M|}} < 1.
\end{equation}
The case in which zero Grover iterations are performed is equivalent to selecting a 32-bit string uniformly at random.
\end{remark}
\begin{remark}[Maximum number of Grover iterations]
\label{remark:maxiterations}
As the quantum-phase oracle encodes only the 32-bit nonce~$N$, for~$D > 2^{32}$, the success probability~$p(k)$~(\ref{eqn:groversuccessbitcoin}) is maximized after performing
\begin{equation}
    \label{eqn:kopt32}
    k_\text{opt}^{32} := \left\lfloor \frac{\pi}{4} \sqrt{2^{32}}\right\rfloor=51471 
\end{equation}
Grover iterations.
In contrast, $k_\text{opt}$~(\ref{eqn:koptbitcoin}) only applies for~$D \leq 2^{32}$ or in the setting where the quantum-phase oracle encodes the entire block header.
\end{remark}

Now we present essential considerations for implementing the FQS algorithm in a realistic setting.
The FQS algorithm offers a quadratic speedup in solving Problem~\ref{prob:unstructuredsearchproblem}, requiring only
\begin{equation}
\label{eqn:qoraclecomplexity}
\mathcal{O}\left(\sqrt{\nicefrac{2^{32}}{|M|}}\right)
\end{equation}
queries to the quantum-phase oracle~$U_{f}$ compared to the
\begin{equation}
\label{eqn:classicaloraclecomplexity}
\mathcal{O}\left(\nicefrac{2^{32}}{|M|}\right)
\end{equation}
oracle queries required by classical methods~\cite{KLM06}.
In this respect, the standard analyses of quantum search algorithms are often focused on minimizing the number of quantum oracle queries, treating it as a black box without consideration of the complexity of the quantum oracle itself.
If the complexities of all operations are considered, then the overall complexity of the FQS algorithm depends on the number of oracle queries, the cost of implementing an oracle query, and the computational overhead of all non-query operations such as state preparation, measurement, and error correction.
Furthermore, if the construction of a quantum oracle is considered in the overall speedup offered by the FQS algorithm, then an efficient and reversible quantum circuit implementing the quantum oracle must be constructed, which can be hard depending on the problem being solved.
Therefore, while the FQS algorithm offers a quadratic speedup in the number of oracle calls, the overall advantage could be diminished if the overhead incurred by querying the oracle or non-query operations is too large.
\subsubsection{Applying the FQS algorithm to Bitcoin mining}

Now we discuss how the FQS algorithm can be applied to mining Bitcoin.
We begin by demonstrating the theoretical possibility of utilizing the FQS algorithm to mine Bitcoin by showing that the problem of finding a proof of work reduces to the probabilistic unstructured search problem.
Then, we discuss the key differences between quantum and classical Bitcoin mining.
Finally, we discuss a particular competition that arises between miners in the quantum setting.

We prove that the FQS algorithm could theoretically be applied to the Bitcoin mining problem.
We begin by establishing that Problem~\ref{prob:boundedproofofwork} reduces to  Problem~\ref{prob:unstructuredsearchproblem}.
\begin{lemma}[Reduction to probabilistic unstructured search ]
 \label{lemma:miningproblemisunstructures}
The bounded PoW problem can be reduced to the probabilistic unstructured search problem.
\end{lemma}
\noindent To prove this reduction, we first construct a Boolean
oracle that encodes the PoW condition for a specific instance of the bounded PoW problem. 
Our Boolean oracle is used as input into Problem~\ref{prob:unstructuredsearchproblem}. 
Thus, solving this instance of Problem~\ref{prob:unstructuredsearchproblem} yields a solution to an instance of Problem~\ref{prob:boundedproofofwork}.
\begin{figure}
    \centering
    \includegraphics[width=0.5\textwidth]{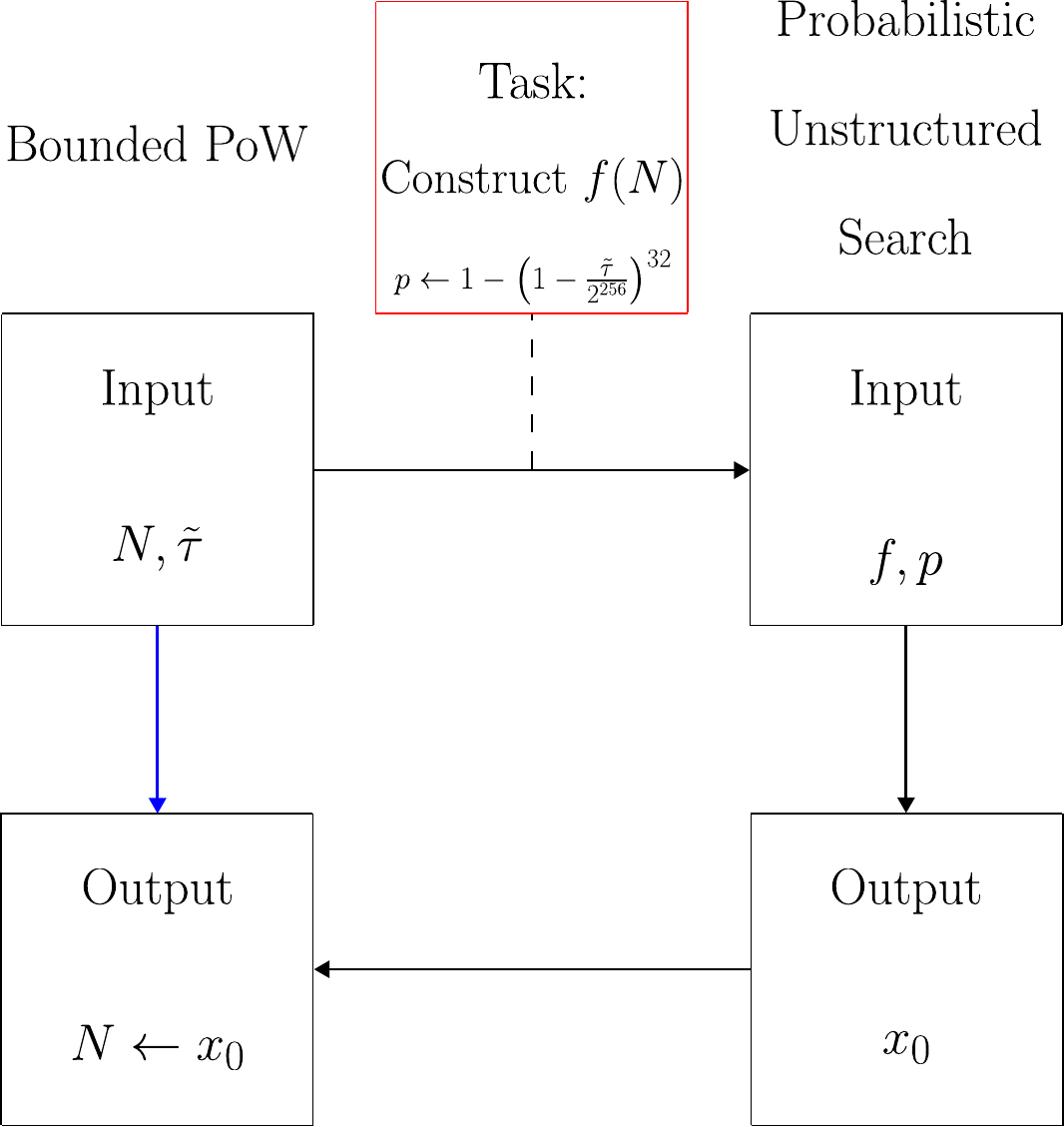}
    \caption{A visualization of the reduction from the bounded PoW problem to probabilistic unstructured search.
    A square ($\square$) represents an input or output of the respective problem, while a red square (\textcolor{red}{$\square$}) represents a task that must be completed during the reduction step.
    A black arrow indicates the mapping between components in the reduction.
    The blue arrow represents the mapping between the input to the transformed output for the bounded PoW problem.
    The dashed line indicates an intermediate step.
    For the bounded PoW problem, $N$ represents the 32-bit nonce and~$\tilde{\tau}$ is the 256-bit difficulty target.
    The Boolean oracle~$f$ takes~$N$ as input and is constructed for a particular instance of the bounded PoW problem with success probability~$p$.
    Finally, $x_0$ represents a nonce such that the resultant block header is a proof-of-work.
    }
    \label{fig:reduction}
\end{figure}
\begin{proof} 
Let~$f$ be the 32-bit Boolean oracle 
\begin{equation}
    \label{eqn:reductionoracle}
    f: \{0, 1\}^{32} \to \{0, 1\}: N \mapsto
    \begin{cases}
        1 & \operatorname{SHA-256}^2(V \mathbin\Vert P \mathbin\Vert R \mathbin\Vert \mathfrak{t} \mathbin\Vert \tau \mathbin\Vert N) < \tilde{\tau} \\
        0 & \text{else}.
    \end{cases}
\end{equation}
The constant bitstrings $V, P, R, \mathfrak{t}$, and $\tau$ defining an instance of the bounded PoW problem determine the construction of~$f$.
Let~$M$ (\ref{eqn:solutionsubset}) be the set of nonces that satisfy the PoW condition.
As the outputs of $\operatorname{SHA-256}$ are uniformly distributed (Definition~\ref{def:hashfunction}), the probability that at least one nonce in the 32-bit space satisfies the PoW condition (i.e., the probability that $M \neq \emptyset$) is
\begin{equation}
\label{eqn:probreduction}
p=1 - \left(1 - \nicefrac{\tilde{\tau}}{2^{256}}\right)^{2^{32}}.
\end{equation}
By Remark~\ref{remark:approvedhash}, $f$ is computable in polynomial time.
Thus, solving the bounded PoW problem reduces to solving an instance of the probabilistic unstructured search problem with a 32-bit Boolean oracle~$f$ and probability~$p$.
Any algorithm that solves the probabilistic unstructured search problem can therefore be used to solve the bounded PoW problem using the constructed Boolean oracle~$f$.
\end{proof}
\begin{remark}[Assuming at least one solution]
For typical values of the network difficulty function~$D(\tilde{\tau})$, the probability that at least one valid nonce exists within a 32-bit space is small. 
If no valid nonce exists~$(|M| = 0)$, the FQS algorithm cannot yield a valid proof-of-work with any number of Grover iterations~(\ref{eqn:groversuccessbitcoin}).
In this case, a quantum miner would opt out of the quantum search entirely.
However, we employ the procedure for the case that a valid nonce exists to study the advantage offered by the FQS algorithm when a valid nonce is present in the 32-bit search space.
\end{remark}
\noindent If an efficient and reversible quantum circuit capable of computing the $\operatorname{SHA-256}$ hash function existed, then by building on Lemma~\ref{lemma:miningproblemisunstructures} and Remark~\ref{remark:approvedhash}, an efficient quantum circuit implementing the quantum oracle for mining Bitcoin could be constructed~\cite{NieChu10,ABL+17}.

Now we discuss the key differences between mining Bitcoin in the classical and quantum settings.
Applying the FQS algorithm to Bitcoin mining introduces new techniques that are not present in classical mining.
Classically, miners check block headers sequentially and immediately determine whether a given header satisfies the PoW condition (Def.~\ref{def:powcondition}).
In contrast, quantum miners must commit a specific amount of quantum information processing time before making a measurement~\cite{LRS19, Sat20}.
As a result, quantum miners cannot continually verify block headers; instead, they only learn of the outcome of the search after measurement. 
This fundamental difference leads to a unique competition between the miners that is only present in the quantum setting.

Now we discuss a unique behaviour among miners that only arises when employing quantum mining techniques. 
Quantum miners must commit time to quantum information processing, aiming to perform as many Grover iterations as they can, or as close to the optimal number of Grover iterations as possible, to maximize their success probability. 
However, the quantum miners want to measure before each other to be the first to claim the mining reward~\cite{LRS19}.
If only one quantum miner is in the network, they would focus on performing~$k_\text{opt}$ Grover iterations before measuring~\cite{NerGau23}.
However, if multiple quantum miners in the network are present, determining when to measure is much more complicated.
For example, a quantum adversary could perform a quantum measurement (Def.~\ref{def:qmeas}) after performing fewer than~$k_\text{opt}$ Grover iterations, thereby reducing their success probability but gaining a temporal advantage and a higher chance of being the first to broadcast a valid block to the network~\cite{LRS19, Sat20}.
Therefore, quantum miners must strike a balance between maximizing their success probability and measuring before their adversaries. 
\subsubsection{Aggressive and peaceful quantum mining}

Now we explore the difference between aggressive and peaceful quantum mining following Sattath's explanation~\cite{Sat20}.
We begin by discussing the choice that quantum miners are faced with when they learn of a new block. 
Then we explain what the peaceful quantum mining strategy is. 
Finally, we explain the aggressive mining strategy and its consequences on the Bitcoin network.

We now turn to the decision that quantum miners must make when they learn of a newly-mined block while performing a Grover search.
Upon receiving a broadcast of a new block, quantum miners face a crucial decision: either discard their search and begin a new one for the next block, or halt the search and immediately measure the resulting quantum state in hopes of producing a valid block.
The first option is known as the peaceful quantum mining strategy and mirrors the classical approach of abandoning a block once a new block is found.
The second is called the aggressive quantum mining strategy, which is an opportunistic approach unique to quantum miners that minimizes the waste of computational resources already invested in solving Problem~\ref{prob:proofofwork}.

We discuss the peaceful quantum mining strategy (PQMS). 
In the PQMS, the quantum miners discard their current Grover search entirely and start anew upon receiving a broadcast of a newly-mined block.
This behaviour mirrors classical Bitcoin miners, who immediately begin mining on top of newly-mined blocks as soon as they receive a broadcast of them.
By discarding their Grover search, quantum miners employing the PQMS avoid intentionally forking the blockchain.
However, unlike in the classical setting, discarding a Grover search wastes significant computational resources already invested in executing the FQS algorithm.

Next, we turn to the aggressive quantum mining strategy (AQMS), as proposed by Sattath~\cite{Sat20}. 
In his strategy, quantum miners halt their search and measure upon receiving a broadcast of a newly-mined block, attempting to maximize their chances of temporarily forking the blockchain.
By employing Sattath's AQMS, quantum miners avoid wasting the computational effort already invested in executing the FQS algorithm.
If the quantum measurement (Def.~\ref{def:qmeas}) is successful, the quantum miner will attempt to temporarily fork the blockchain. 
If all quantum miners in the network employ Sattath's AQMS, the quantum miners will halt and measure their quantum states nearly simultaneously when they receive a broadcast of a newly-mined block.
As a result, the times at which quantum miners perform measurements and the times at which blocks are found become highly correlated. 
This correlation, which only occurs in the quantum setting due to Sattath's AQMS, leads to an increased stale rate~$p_\text{stale}$ in Def.~\ref{def:stalerate} compared to the classical setting.
\begin{remark}
\label{remark:reducedhashrate}
Sattath's AQMS increases the Bitcoin network's stale rate~$p_\text{stale}$. 
As a consequence, the effective threshold for a 51\% attack~(\ref{eqn:stalerate51attack}) decreases, meaning that a consortium with fractional hash rate~$q_c$ less than one half would be able to execute a 51\% attack.
\end{remark}
\begin{remark}[Regulation of peaceful mining]
\label{remark:peacefulregulation}
As a network regulation that all quantum miners be peaceful is not enforceable, aggressive mining should be the default assumption.
\end{remark}
\begin{figure}[htbp]
    \centering
    \includegraphics[width=0.75\textwidth]{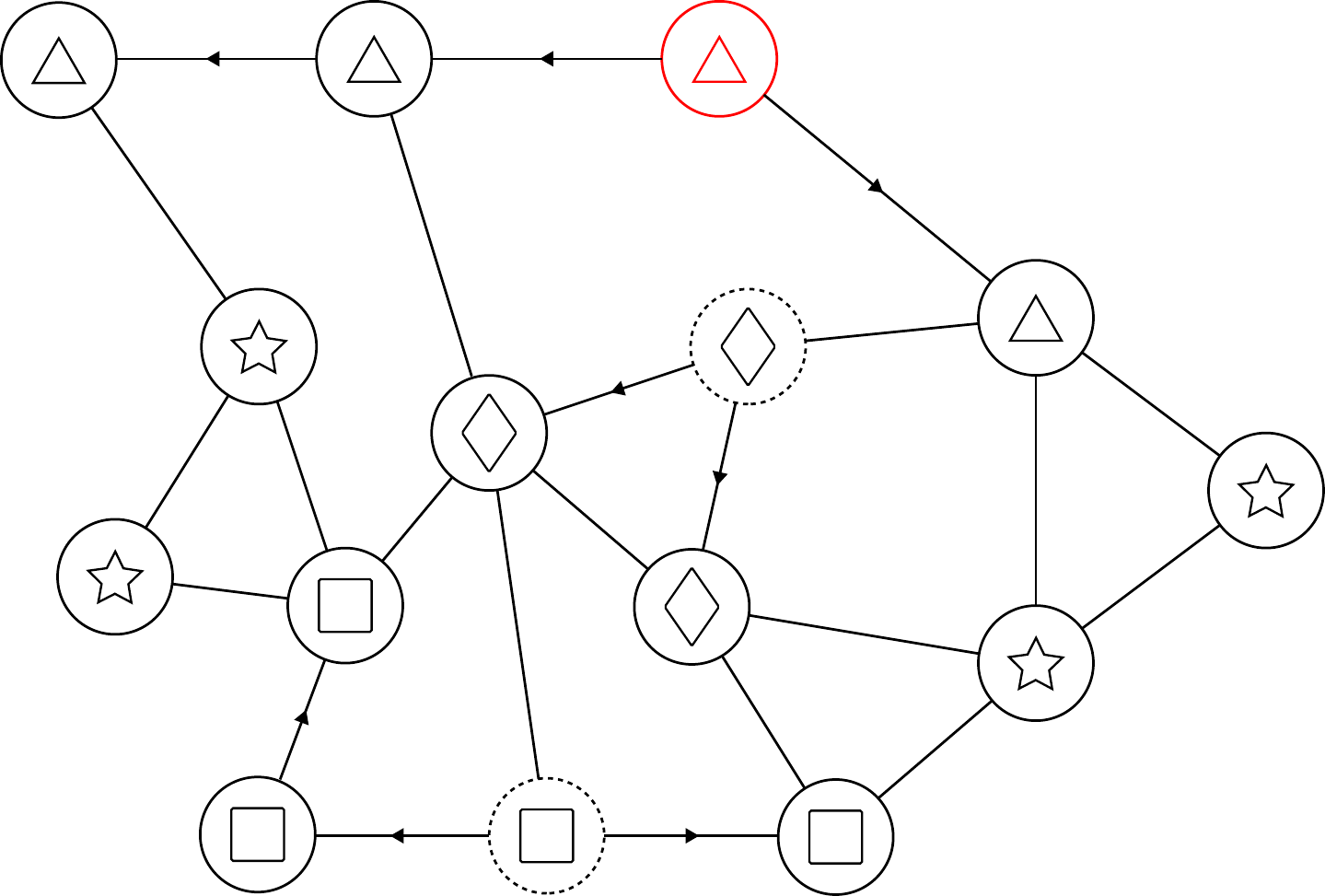}
    \caption{Blockchain forking due to Sattath's AQMS. 
    A circle ($\bigcirc$) represents a node in the Bitcoin network, with a dashed circle representing a quantum miner employing Sattath's AQMS.
    The red circle (\textcolor{red}{$\bigcirc$}) represents a classical mining node that has just mined a new block.
    Each vertex contains either a star ($\bigstar$), which represents the blockchain state before the new block was discovered, or a triangle ($\triangle$), square ($\square$), or diamond ($\lozenge$), representing the new branches of the blockchain during a temporary fork.
    An arrow within an edge indicates the broadcast of a new block to a peer.
    An edge without an arrow represents a peer-to-peer connection where no new block is being broadcast.
    This figure represents a temporary fork because all of the nodes do not have the same symbol.
    }
    \label{fig:quantum_fork_network}
\end{figure}
\subsection{Quantum races}
\label{subsec:quantumraces}

We explain the quantum races model and its underlying game-theoretic foundation.
We begin with some preliminary game-theoretic concepts that are essential to understanding the quantum races model.
Next, we describe some methods used for computing a Nash equilibrium and discuss the computational complexity of this problem.
Finally, we introduce the quantum races model and focus on the case of two players.
\subsubsection{Preliminaries on game theory}

We discuss essential game theory concepts for understanding the quantum races model. 
We begin with a top-down explanation of what a strategic game is, noting the essential assumptions about the participants. 
Then, we explain what strategies are in strategic games. 
Next, we discuss what payoff functions are and how they quantify the players' preferences over the outcomes of the game. 
Finally, we explain what a Nash equilibrium is and discuss how its existence is guaranteed.

We begin with a top-down explanation of what a strategic game is, introducing each of its essential components.
A strategic game is a mathematical model of interactive decision making, where each participant chooses a plan of action simultaneously and once and for all~\cite{OsRu94}.
Formally, a strategic game is defined as follows.
\begin{definition}[Strategic game~\cite{Tad13}]
\label{def:strategicgame}
A strategic game is a 3-tuple comprising
\begin{itemize}
    \item[] a set of players~$\wp$ that is finite:~$|\wp| =: n_\wp$,
    \item[] a set of strategies for each player~$S_i$, and 
    \item[] a set of payoff functions for each player that give a payoff value for each combination of the players' chosen strategies. 
\end{itemize}
\end{definition}
\noindent Although several different types of games exist~\cite{Tad13, MSZ13}, we focus exclusively on strategic games in this work. 
In the standard framework of strategic games, players are assumed to be both rational and intelligent.
\begin{definition}[Rational player~\cite{Tad13}]
\label{def:rationalplayer}
A rational player chooses their action to maximize their payoff consistent with their beliefs of what is going on in the game.
\end{definition}
\begin{definition}[Intelligent player~\cite{Tad13}]
\label{def:intelligentplayer}
An intelligent player possesses complete knowledge about a game, including the actions available to themselves and their opponents, the possible outcomes of the game, and the preferences of all players.
\end{definition}
\noindent These assumptions, along with the structure of the game itself, are assumed to be common knowledge~\cite{Tad13}.

Now we explain the notion of strategies in the context of strategic games. 
Each player $i \in \wp$ is associated with a set of actions~$S_i$ permitted by the rules of the game.
If this set of actions is finite, then the game is called finite~\cite{MSZ13}.
A strategy is an action chosen from $S_i$ deterministically (pure), or a probability distribution over~$S_i$ (mixed).
\begin{definition}[Pure strategy~\cite{OsRu94}]
\label{def:purestrategy}
A pure strategy is a single action $s_i \in S_i$ that a player~$i \in \wp$ selects deterministically. 
\end{definition}
\begin{definition}[Mixed strategy~\cite{MSZ13}]
A mixed strategy~$\sigma_i$ for player~$i \in \wp$ is a probability distribution over $S_i$.
\end{definition}
\begin{definition}[Support of mixed strategy~\cite{MSZ13}]
The support of a mixed strategy 
\begin{equation}
    \label{eqn:support}
    \operatorname{supp}(\sigma_i) := \{s_i \in S_i : \sigma_i(s_i) > 0\}
\end{equation}
is the set of pure strategies that are played with positive probability by a mixed strategy~$\sigma_i$.
\end{definition}
\noindent The set of mixed strategies for player $i$ is~\cite{MSZ13}
\begin{equation}
\label{eqn:setofmixedstrat}
\Sigma_i := \left\{\sigma_i : S_i \to [0, 1]: \sum_{s_i \in S_i} \sigma_i(s_i)=1\right\}
.\end{equation}
\noindent A pure strategy can also be considered a mixed strategy with a degenerate distribution that selects a single element of~$S_i$ with probability 1~\cite{Tad13}.
\begin{definition}[Pure strategy profile~\cite{Tad13}]
A pure strategy profile is an $n_\wp$-tuple
\begin{equation}
\label{eqn:s*}
s^*=(s_1, s_2, \dots, s_{n_\wp}),
\end{equation}
where each $s_i \in S_i$ is the pure strategy chosen by player $i \in \wp$.
The set of all pure strategy profiles is denoted~$S^*$.
\end{definition}
\noindent A mixed strategy profile~$\sigma^*$ is defined analogously, where each $\sigma_i$ is the mixed strategy chosen by player~$i \in \wp$~\cite{Tad13}.
\begin{remark}[Concise strategy profile notation~\cite{Tad13}]
A pure strategy profile is
\begin{equation}
\label{eqn:s*concise}
s^* := (s_i, s_{-i}^*),
\end{equation}
where $s_{-i}^*$ is the $(n_\wp-1)$-tuple denoting the strategies chosen by all players except player $i$.
Similarly, a mixed strategy profile is denoted as $\sigma^*=(\sigma_i, \sigma_{-i}^*)$.
\end{remark} 
\noindent We note that Sattath's AQMS~\cite{Sat20} does not correspond to a strategy in the formal game-theoretic sense, and only specifies how a quantum miner will react to receiving a broadcast of a newly-mined block.

Now we discuss how a player's preferences are quantified in payoff. 
In a strategic game (Def.~\ref{def:strategicgame}), each player $i \in \wp$ is associated with a payoff function that assigns a real number to every possible strategy chosen by all players~\cite{Tad13}.
The output of the payoff function is called the payoff, and represents the player's preference over the game's outcome that has been determined by the strategies played by all players.
\begin{definition}[Payoff function~\cite{MSZ13}]
\label{def:payofffunction}
A payoff function $u_i : S^* \to \mathbb{R}$ represents player~$i$'s preference on the outcome of a game led to by the strategies chosen in the pure strategy profile $s^* \in S^*$.
\end{definition}
\noindent We adopt a definition of a payoff matrix based on standard textbooks on game theory~\cite{NRTV07, MSZ13}.
\begin{definition}[Payoff matrix]
\label{def:payoffmatrix}
For two-player finite strategic games with~$S_{1,2}$ the sets of actions available to the players, the payoff functions
\begin{equation}
\label{eqn:payoffmatrixmapping}
u_i : S_1 \times S_2 \to \mathbb{R} : (s_1, s_2) \mapsto u_i(s_1, s_2)
\end{equation}
determine a pair of payoff matrices,
\begin{equation}
\label{eqn:payoffmatrixdim}
A \in \mathbb{R}^{|S_1| \times |S_2|}, \;\; B \in \mathbb{R}^{|S_2| \times |S_1|},
\end{equation}
with payoff-matrix entries
\begin{equation}
\label{eqn:payoffmatrixentriesmapping}
(s_1, s_2) \mapsto u_1(s_1, s_2), \;\; (s_1, s_2) \mapsto u_2(s_1, s_2),
\end{equation}
respectively.
\end{definition}
\noindent
\begin{remark}
    \label{remark:uvnotation}
    For two player games, we adopt the notation 
    \begin{equation}
    \label{eqn:uvpayoffnotation}
    u := u_1 \;\; v := u_2
    \end{equation}
    as the payoff functions for both players. 
\end{remark}
\noindent 
Thus,
the payoff matrix entries are
\begin{equation}
\label{eqn:payoffmatrixentries}
A_{s_1s_2} =  u(s_1,s_2)
\end{equation}
and similarly for~$B$.
\begin{remark}[Expected payoff for mixed strategies~\cite{MSZ13}]
\label{remark:expectedpayoff}
For a mixed strategy profile~$\sigma^*$, the expected payoff is given by 
\begin{equation}
\label{eqn:expectedpayoff}
U(\sigma^*) := \mathbb{E}_{\sigma^*}[u(\sigma^*)]=\sum_{s^* \in S^*} u(s^*) \prod_{i \in \wp} \sigma_i (s_i)
.\end{equation} 
\end{remark}
\noindent Note that payoff is an ordinal construct.
The value of payoff itself has no meaning and is used to order the desirability of alternatives~\cite{Tad13}.
Importantly, the payoff function is used to define a best-response.
\begin{definition}[Best response~\cite{Tad13}]
\label{def:bestresponse}
The strategy $\sigma_i \in \Sigma_i$ is player~$i$'s best response to their opponents' strategies $\sigma_{-i} \in \Sigma_{-i}$ if 
\begin{equation}
\label{eqn:bestresponse}
u_i(\sigma_i, \sigma_{-i}^*) \geq u_i(\sigma_i^\prime, \sigma_{-i}^*) \; \forall \sigma_i^\prime \in \Sigma_i
.\end{equation}
\end{definition}
\begin{definition}[Strictly dominated strategy~\cite{MSZ13}]
\label{def:dominatedstrategy}
A strategy~$s_i \in S_i$ of player~$i \in \wp$ is strictly dominated if 
\begin{equation}
    \label{eqn:strictlydominated}
    \exists s_i^\prime \in S_i : u_i(s_i, s_{-i}^*) < u_i(s_i^\prime, s_{-i}^*) \forall s_{-i} \in S_{-i}.
\end{equation}
A rational player will not play a strictly dominated strategy.
\end{definition}
\noindent In two-player finite strategic games~(Def.~\ref{def:strategicgame}), the payoff functions are represented as payoff matrices~\cite{Tad13}.
Each player's payoff matrix lists their payoff for each combination of pure strategies played by both players: rows correspond to their own actions, and columns correspond to the actions of their opponent.
\begin{definition}[Symmetric game~\cite{MSZ13}]
\label{def:symmetricgame}
A game is symmetric if the players share the same pure strategy set
\begin{equation}
    \label{eqn:purestrategysetsame}
    S_i=S_j \;\forall i, j \in \wp,
\end{equation}
and if,
for any mixed-strategy profile
\begin{equation}
\label{eq:mixedstrategyprofile}
\sigma^*=(\sigma_1, \sigma_2, \dots , \sigma_{n_\wp}),
\end{equation}
then,
for any permutation 
\begin{equation}
    \label{eqn:permutation}
    \pi : \wp \to \wp 
\end{equation}
of the players~$i \in \wp$, 
\begin{equation}
    \label{eqn:payoffpermutation}
    u_i(\sigma^*)=u_{\pi(i)}(\pi(\sigma^*)),
\end{equation}
where
\begin{equation}
    \label{eqn:permutatestrategyprofile}
    \pi(\sigma^*) := \left(\sigma_{\pi(1)}, \sigma_{\pi(2)}, \dots, \sigma_{\pi\left(n_\wp\right)}\right).
\end{equation}
\end{definition}
\begin{remark}[Payoff for symmetric two-player strategic games~\cite{MSZ13}]
\label{remark:payofftwoplayer}
In a two-player symmetric strategic game~(Def.~\ref{def:strategicgame}), each player's payoff matrix is the transpose of the other's.
\end{remark}

Here we describe what a Nash equilibrium is.
A Nash equilibrium is a strategy profile where each player plays a best response to the strategies of all other players, and no player has an incentive to unilaterally deviate~\cite{Tad13}.
\begin{definition}[Nash equilibrium~\cite{Tad13}]
\label{def:nasheq}
A strategy profile~$\sigma^*=(\sigma_1, \sigma_2, \dots, \sigma_n)$ is a Nash equilibrium if 
\begin{equation}
\label{eqn:nasheq}
u_i(\sigma_i, \sigma_{-i}^*) \geq u_i (\sigma_i^\prime, \sigma_{-i}^*) \; \forall \sigma_i^\prime \in \Sigma_i, \; \forall i \in \wp.
\end{equation}
\end{definition}
\noindent As pure strategies are a degenerate case of mixed strategies, this definition of the Nash equilibrium applies to both pure and mixed strategies.
In non-cooperative strategic games where all players are rational and intelligent, the Nash equilibrium is self-enforcing because all players know that deviating from their equilibrium strategy will result in a lower expected payoff~\cite{Tad13}.

Importantly, the existence of a mixed-strategy Nash equilibrium is guaranteed by Nash's theorem.
\begin{theorem}[Nash's existence theorem~\cite{Nash51, Tad13}]
\label{theorem:nash}
Any $n_\wp$-player finite strategic game has at least one mixed-strategy Nash equilibrium.
\end{theorem}
\begin{theorem}[Nash's existence theorem for symmetric games~\cite{Nash51}]
\label{theorem:nashsymm}
Any finite symmetric game has at least one symmetric Nash equilibrium.
\end{theorem}
\noindent Whereas a game can have multiple Nash equilibria, the existence of additional equilibria is not guaranteed. 
If multiple Nash equilibria exist, some can result in a better expected payoff for the players than others.
\subsubsection{Computing a Nash equilibrium}

We discuss the problem of computing a Nash equilibrium in finite strategic games. 
We begin by giving an overview of different methods used to compute a Nash equilibrium.
Then, we focus on the Lemke-Howson algorithm, a standard method for computing an exact Nash equilibrium in two-player games.
Finally, we discuss the computational complexity of computing a Nash equilibrium.

We begin by giving an overview of methods used for computing a Nash equilibrium.
The method used to find a Nash equilibrium depends on several properties of the game, including, but not limited to, the number of players, whether the game is symmetric, whether players make decisions sequentially or in turns, and whether the goal is to find an exact or $\epsilon$-approximate Nash equilibrium~\cite{NRTV07}.
As a result, a wide range of methods have been developed for computing a Nash equilibrium, including analytical approaches, graphical techniques~\cite{KLS01}, and algorithmic approaches based on linear programming~\cite{LH64}. 
No single technique exists that is universally applicable to all types and properties of games.

Now we discuss the Lemke-Howson algorithm.
The Lemke-Howson algorithm is a well-established and standard method for finding an exact Nash equilibrium in two-player, finite, strategic games~\cite{LH64, NRTV07}.
As input, the Lemke-Howson algorithm requires the payoff matrices of both players, along with an integer specifying the initial direction of the search~\cite{LH64}.
The Lemke-Howson algorithm outputs a mixed-strategy profile corresponding to an exact Nash equilibrium~(Def.~\ref{def:nasheq}).
This algorithm works by traversing through the edges of a pair of polytopes representing the strategy sets of both players, using a method called complementary pivoting, until a pair of vertices corresponding to a Nash equilibrium is found~\cite{LH64, NRTV07}.
Note that the Lemke-Howson algorithm finds only a single Nash equilibrium, though additional Nash equilibria can be found by enumerating through different initial directions to search.
All Nash equilibria found by Lemke-Howson must lie within the polytopes representing the strategy sets of both players.

We discuss the complexity of computing a Nash equilibrium.
The complexity of computing a Nash equilibrium is essential for understanding the limitations of the Lemke-Howson algorithm.
The problem of computing a Nash equilibrium for mixed-strategy games, denoted \texttt{Nash}, is heavily studied in computational complexity theory~\cite{DGP09}.
Formally, \texttt{Nash} is a function problem, which means the task is to evaluate a function that returns a Nash equilibrium for each game~\cite{Papa94}.
Since a solution to \texttt{Nash} is always guaranteed to exist by Nash's theorem~(Theorem \ref{theorem:nash}), \texttt{Nash} belongs to the complexity class called Total Function Nondeterministic Polynomial (TFNP)~\cite{NRTV07}.
More specifically, in 2006, it was shown that \texttt{Nash} is complete for the Polynomial Parity Arguments on Directed Graphs (PPAD) complexity class, which is a subclass of TFNP~\cite{DGP09}.
Although being PPAD-complete does not prove that \texttt{Nash} is intractable, it does provide strong evidence of intractability~\cite{Pap94}.
This is because any efficient algorithm that solves a PPAD-complete problem like \texttt{Nash} would yield efficient algorithms for all PPAD-complete problems, including computing Brouwer fixed points, which has been resilient to efficient solutions for decades~\cite{HPV89}.
Thus, although the Lemke-Howson algorithm computes an exact Nash equilibrium, it has worst-case exponential runtime, reflecting that \texttt{Nash} is PPAD-complete~\cite{DGP09}.
Additionally, many variants of \texttt{Nash} are known to be NP-complete, such as computing a second Nash equilibrium given a first, or computing a Nash equilibrium that uses particular strategies.
\subsubsection{Two-player quantum races}

We follow the 2019 Lee-Ray-Santha approach (LRS19)~\cite{LRS19} and describe the quantum races model in the case where there are two players, Alice and Bob.
We begin with a description of the quantum races model, noting its purpose and its components.
Then, we discuss the pure strategies for the players and the associated success probabilities.
Next, we specialize to the case where there are only two players, Alice and Bob, and discuss their payoff and the structure of their payoff matrices.
Finally, we explain how the quantum races model can be applied to Bitcoin.

We explain the quantum races model with a model comprising descriptions of the players and their resources, the protocols they follow, the actions they are allowed, and the payoffs associated with those actions.
We formally define a quantum race based on the description by LRS19~\cite{LRS19}.
\begin{definition}[Quantum race]
\label{def:quantumrace}
A quantum race is a competition between two or more players, each equipped with a quantum computer, who are racing to be the first to solve the same computational problem.
\end{definition}
\noindent The quantum races model is introduced as a means to study such competitions using a game-theoretic approach. 
The purpose of the quantum races model is to determine the optimal strategies for the players. 
Before beginning a quantum race, each player~$i \in \wp$ chooses the times at which they will measure their respective quantum states from the set 
\begin{equation}
\label{eqn:quantumracepurestrat}
S_i := \{1, 2, \dots, K\},
\end{equation}
which is the set of pure strategies available to player~$i$, and $K$ is the maximum allowed measurement time. 
The quantum races model assumes that each player performs exactly one quantum measurement for each quantum race.
\begin{remark}[Weak players]
\label{remark:weakplayers}
The assumption that each player performs exactly one quantum measurement per quantum race makes the players ``weak", as if both players' quantum measurements do not yield a solution, then there is no winner.
In contrast, a ``strong" player would maximize their resources by performing multiple quantum measurements until either quantum miner succeeds.
\end{remark}

The set of pure strategies for each player is associated with a success probability that increases with the measurement time. 
For each time $t^\prime \in S_i$, each player has a quantum algorithm that they can run for time~$t^\prime$ with an associated success probability $p(t^\prime)$.
It is assumed that these probabilities form an increasing sequence
\begin{equation}
\label{eqn:probseq}
0 < p(1) < p(2)< \dots < p(K) \leq 1.
\end{equation}
Whereas LRS19 refer to these actions as ``times", we clarify that they can also represent the number of computational steps. 
The quantum races model applies to any quantum algorithm whose success probability can be characterized by an increasing sequence of probabilities at discrete points in its execution.
Each quantum race is fully specified by an increasing sequence of probabilities (\ref{eqn:probseq}) for each player, as well as their respective payoff functions, which we describe next.
\begin{remark}[Two players]
\label{remark:twoplayers}
LRS19 first considered a setting with two players, although they did not justify this choice. 
Presumably, the two-player case was analyzed first because the quantum race is fully represented by a pair of payoff matrices. 
A quantum race involving three or more players requires a higher-dimensional tensor, which makes the analysis much more difficult. 
\end{remark}

We now follow LRS19 and specialize to the symmetric two-player quantum races model~(Remark~\ref{remark:twoplayers}), which we refer to as the S2QR model for short, with Alice and Bob as the players.
\begin{definition}[S2QR]
\label{def:symmetric2QR}
A symmetric two-player quantum race is a quantum race (Def.~\ref{def:quantumrace}) consisting of two players that is symmetric (Def.~\ref{def:symmetricgame}).
\end{definition}
\noindent
In an S2QR, both players possess identical quantum and classical resources~\cite{LRS19}.
Thus, an S2QR is fully specified by an increasing sequence of probabilities associated with both players~(\ref{eqn:probseq}), and a pair of $K \times K$ payoff matrices, $A$ and $B$.
The player who succeeds first in solving the computational problem receives a payoff of 1, whereas the other player receives a payoff of 0. 
In the case of a tie, the payoff is split evenly between them. 
If the payoff is not split in the case of a tie, then the S2QR is called stingy.

We describe the structure of Alice's and Bob's payoff matrices in a symmetric two-player quantum race. 
The $(a, b)$ element of Alice's payoff matrix~$A$ is her payoff in the case where Alice measures after~$a$ units of time or computational steps and Bob measures after~$b$ units of time or computational steps.
The payoff is determined by the conditional probabilities of each player's success and failure, depending on when each player measures relative to the other.
If Alice measures before Bob ($a<b$),
her payoff corresponds to the probability~$p(a)$ that she succeeds whereas Bob has not yet measured. 
In contrast, if Bob measures first $(a > b)$, Alice's payoff is
\begin{equation}
\label{eqn:QRAlicePayoff}
p(a)(1 - p(b)),
\end{equation}
the joint probability that Alice's measurement is successful and Bob's fails. 
If Alice and Bob measure at the same time ($a=b$), the term 
\begin{equation}
\label{eqn:payoffsplitterm}
\frac1{2} p(a)^2
\end{equation}
is added to Alice's payoff, representing the scenario where both Alice's and Bob's quantum measurements yield a valid proof-of-work and the reward is split.
Alice's payoff matrix is thus conveniently expressed as 
\begin{equation}
\label{eqn:alicepayoffLRS19}
A:\left[k_\text{opt}\right] \times \left[k_\text{opt}\right] \to[0,1]:
(a, b)\mapsto
\begin{cases}
    p(a), & a < b \\
    p(a)(1 - p(b)) + \frac1{2} p(a)^2, & a=b \\
    p(a)(1 - p(b)), & a > b,
\end{cases}
\end{equation}
with~$(a, b)$ being~$(s_1, s_2)$ and the right hand side being~$u_1(s_1, s_2)$ (Def.~\ref{def:payoffmatrix}).
As the game is symmetric, Bob's payoff matrix is $B=A^\intercal$. 
An extension to the S2QR model, studying the case where each player performs two quantum measurements between successive blocks, is presented in Maharshi Ray's PhD thesis~\cite{Ray2020}.

Now we discuss how the S2QR model is applied to the context of mining Bitcoin.
In this setting, Alice and Bob each represents a quantum miner, who are each equipped with identical quantum computers, and compete to solve Problem~\ref{prob:boundedproofofwork} before each other.
The set of pure strategies available to each quantum miner 
is the set of Grover iterations that can be performed before measurement
\begin{equation}
\label{eqn:Sikopt}
S_i=\{1, 2, \dots,k_\text{opt}\},
\end{equation}
where $k_\text{opt}$~(\ref{eqn:koptbitcoin}) is the optimal number of Grover iterations to perform for a given network difficulty~$D$, and~$i$ is the index representing each player.
In this quantum race, Alice and Bob each independently execute Alg.~\ref{alg:grover} with their own choice of~$k \in S_i$ Grover iterations, obtaining a nonce~$x$ as a result of the quantum measurement~(\ref{eqn:measurement}).
Alice or Bob win the quantum race if their own~$x$ satisfies the PoW condition~(Def.~\ref{def:powcondition}). 
As both Alice and Bob have identical resources in this setting, 
the success probability~(\ref{eqn:groversuccessbitcoin}) for each pure strategy $s_i \in S_i$ is the probability that the FQS algorithm returns a nonce constituting a valid proof-of-work given a number of Grover iterations and network difficulty~$D$.
These probabilities form an increasing sequence 
\begin{equation}
\label{eqn:seqkopt}
p(1) < p(2) < \dots  < p(k_{\text{opt}}) \leq 1,
\end{equation}
reflecting the quadratically increasing likelihood of success as the number of Grover iterations increases.
The quantum miner who first finds a valid block header receives a payoff of 1, representing the mining reward, whereas all other quantum miners receive a payoff of 0.
The S2QR model for Bitcoin is based on the assumption that the quantum miners employ the PQMS.
LRS19 model Bitcoin mining as a non-stingy quantum race, where the payoff is split in the case of a tie~(\ref{eqn:alicepayoffLRS19}).
However, the bulk of their analysis considers the stingy variant, which is more tractable.
In contrast, our work permits the use of Sattath's AQMS, so that in the case of a tie, Alice's block wins with probability~$\gamma_\text{S}$ (Def.~\ref{def:sattathspropparameter}).
\section{Approach}
\label{sec:approach}

Now we discuss our approach for determining the impact that quantum miners would have on Bitcoin's security.
We begin by discussing our model, where we consider two quantum miners in the Bitcoin network who compete against each other while mining.
Next, we explain the mathematics for our model, paying special attention to the strategies available to the quantum miners and the construction of their payoff matrices.
Finally, we describe the methods used to compute optimal quantum mining strategies and simulate their impact on Bitcoin's security by estimating the resulting stale rate of the Bitcoin network.
\subsection{Model}
\label{subsec:model}

We present our model as a concept with the mathematics relegated to the next subsection.
We begin by laying the foundations of our model, specifying our extensions to the S2QR model for Bitcoin, as well as the relevant parameters of the Bitcoin network.
Then, we discuss the agents in our model, Alice and Bob.
Finally, we build up the adversarial model we use to study how the agents affect Bitcoin's security. 
\subsubsection{Model foundations}

We first lay the necessary foundations of our model.
We begin by giving an overview of our extensions to the S2QR model for Bitcoin
and explain the connection between our model and Problem~\ref{prob:quantumthreat}.
Then, we list the key Bitcoin network parameters that parameterize our model.
Finally, we explain the regimes for the network difficulty our model considers.

Following the LRS19 approach to restrict to two quantum miners~(Remark~\ref{remark:twoplayers}), we extend the S2QR model for Bitcoin in two ways. 
First, we consider a setting where the quantum miners perform multiple quantum measurements (Def.~\ref{def:qmeas}) between successive blocks, rather than perform a single quantum measurement. 
This choice is motivated by the idea that in a practical setting, quantum miners would likely perform multiple quantum measurements until one of the quantum miners wins (Remark~\ref{remark:weakplayers}).
Second, we consider quantum miners who employ Sattath's AQMS~\cite{Sat20}, which is motivated by the fact that a regulation that all quantum miners in the network must be peaceful is not enforceable~(Remark~\ref{remark:peacefulregulation}).
These extensions to the S2QR model for Bitcoin capture quantum mining behaviour that closer reflects how quantum mining would be performed in reality.
Our goal is to determine the optimal strategies for the quantum miners in this new setting, and study how these optimal strategies affect Bitcoin's security against a 51\% attack.

We discuss the key parameters of the Bitcoin network that we employ in our model. 
Our model accepts the network difficulty~$D$, which influences the success probability for the FQS algorithm~(\ref{eqn:groversuccessbitcoin}) and the optimal number of Grover iterations to perform~$k_\text{opt}$~(\ref{eqn:koptbitcoin}).
Furthermore, we observe how the stale rate~$p_\text{stale}$ changes as a result of the optimal strategies for the quantum miners.
Finally, we assume that the hash rate~$q$ of the Bitcoin network is constant for a given network difficulty~$D$, so that the block arrival times are modelled as a homogeneous Poisson process.

We consider three regimes for the network difficulty in our model.
Recognising that
\begin{equation}
\label{eqn:lowD}
D_0=2^{32},
\end{equation}
for which Problem~\ref{prob:boundedproofofwork} has at least one expected solution,
we now define the difficulty regimes for the network.
\begin{definition}[Network difficulty regimes]
\label{def:networkdifficultyregimes}
The three difficulty regimes for the Bitcoin network are low if $D \leq D_0$, high if $D > D_0$, and ideal if $D > D_0$ with a perfect quantum-phase oracle (Def.~\ref{def:quantumoracle}) encoding the entire block header. 
\end{definition}
\noindent In the ideal regime, the quantum-phase oracle is replaced with a 640-qubit analogue that encodes all the necessary logic required for block validation including Merkle root computation, transaction reordering, \texttt{extraNonce}, timestamp verification, and difficulty target updates. 
By studying this ideal regime, we can model an upper bound on the capabilities of the quantum miners and the theoretical limits of quantum Bitcoin mining.
\subsubsection{Alice and Bob}

Now we discuss Alice and Bob, the quantum miners in our model.
We begin by discussing the roles Alice and Bob play in our model, highlighting their goals.
Then, we go over the resources available to Alice and Bob.
Finally, we lay out and justify some fundamental assumptions about Alice's and Bob's behaviour in our model.

We conceptually explain the roles of Alice and Bob in our model.
Alice and Bob are rational and intelligent quantum miners who are competing in a quantum race~\cite{LRS19} with the goal of successfully mining a block before each other.
For each quantum race, Alice and Bob execute the FQS algorithm over the space of possible nonces for a block using a quantum-phase oracle (Def.~\ref{def:quantumoracle}) encoding the PoW condition~(\ref{eqn:miningcondition}).
After performing the quantum measurement (Def.~\ref{def:qmeas}), if the returned nonce satisfies the PoW condition, then that quantum miner wins the quantum race.
Alice and Bob operate within an otherwise classical Bitcoin network that is unaware of their quantum capabilities.
To maximize their individual rewards, Alice and Bob employ Sattath's AQMS~\cite{Sat20} while competing in this quantum race.
Alice and Bob neither cooperate nor communicate with each other.
To avoid a situation where the quantum miners repeatedly undermine each other's chosen strategy~$s_i \in S_i$~(\ref{eqn:Sikopt}), Alice and Bob compute and employ the same Nash equilibrium~$\sigma^*$ satisfying condition~(\ref{eqn:nasheq}) for the two-player game representing their quantum race using publicly available information. 
We refer to the Nash equilibrium~$\sigma^*$ for competing quantum miners in a two-player game as an optimal quantum mining strategy.
A full mathematical description of Alice's and Bob's quantum race is deferred to~\S\ref{subsec:mathematics}.

We discuss the resources available to Alice and Bob.
Alice and Bob are each equipped with an identical quantum computer that is purpose-built for executing the FQS algorithm for mining Bitcoin.
For each of Alice's and Bob's candidate blocks, they are provided with the corresponding Grover iterate (Def.~\ref{def:groveriterate}), so that the cost of operating the quantum-phase oracle is reduced to queries.
For each candidate block, Alice and Bob execute the FQS algorithm exactly 
\begin{equation}
    \label{eqn:c}
    c \in \mathbb{N}
\end{equation}
times each in search of a valid block header.
The cumulative number of Grover iterations that each quantum miner could perform over all measurements is bounded above by the Grover iteration budget
\begin{equation}
\label{eqn:groveriterbudget}
\kappa := \min\left(k_\text{opt}, k_\text{opt}^{32}\right),
\end{equation}
where~$k_\text{opt}$~(\ref{eqn:koptbitcoin}) is the optimal number of Grover iterations to perform for network difficulty~$D$, and~$k_\text{opt}^{32}$~(\ref{eqn:kopt32}) is the optimal number of Grover iterations to perform for a 32-qubit quantum-phase oracle~(Def.~\ref{def:quantumoracle}).
This choice of~$\kappa$ is motivated by the fact that the 32-qubit quantum-phase oracle constrains the number of Grover iterations performed for~$D > 2^{32}$~(Remark~\ref{remark:maxiterations}).
Furthermore, $\kappa$ is selected such that the problem of how Alice and Bob should optimally allocate their resources over~$c$ quantum measurements is nontrivial.
A resource budget that is too low would concentrate all resources into a single quantum measurement, whereas a resource budget that is too high would allocate~$k_\text{opt}^{32}$ into each of their~$c$ quantum measurements independently.
The propagation parameter
\begin{equation}
\label{eq:propparameter}
\gamma \in [0, 1]
\end{equation}
represents the probability that Alice's block is included in the longest blockchain in the event of a temporary fork, and $1 - \gamma$ represents this same probability for Bob.
Finally, we assume that the cost of all non-query operations, such as state preparation, measurement, and constructing the relevant quantum-phase oracle is negligible.

We make the following assumption about Alice's and Bob's behaviour in our model.
Although employing Sattath's AQMS decreases the effective threshold for a 51\% attack (Remark~\ref{remark:reducedhashrate}), we assume that this does not influence Alice's and Bob's actions.
Instead, we model Alice and Bob as being only interested in beating the other in the quantum race. 
A more complete problem would be to model the interaction between the quantum miners and the Bitcoin network using a game-theoretic framework.
However, such an analysis would require modelling the full state of the Bitcoin network, which includes the behaviour of all nodes, adjustments in the network difficulty, and the many factors that affect block propagation times. 
Thus, we model the direct competition between the quantum miners to study their strategic interactions, which serves as a step towards understanding the full problem.
Furthermore, as the number of pure strategies for Alice and Bob grows combinatorially with~$c$, the evaluation of Alice's and Bob's payoff matrices is computationally expensive. 
Therefore, we introduce a heuristic to reduce the cardinality of their pure strategy spaces, which we formalize in~\S\ref{subsec:mathematics} and outline in~\S\ref{subsec:methods}.
\subsubsection{Adversarial model}

We formalize the adversarial model used in our study.
We first define what a Bitcoin adversary is, as we use this definition to form the foundation of our adversary model. 
Then, we specify the condition under which we consider the Bitcoin network to have failed.
Finally, we identify the adversaries in our model and explain how they affect Bitcoin's security.

To formalize our adversarial model, we first define what a Bitcoin adversary is in the context of our model. 
Whereas standard textbooks define an adversary in general as an agent that generates worst-case inputs to maximize the cost of executing an algorithm~\cite{MotRaj95}, this definition does not capture how an agent might cause the Bitcoin network to fail on a task it was designed to perform. 
Thus, we define a Bitcoin adversary as follows.
\begin{definition}[Bitcoin Adversary ]
\label{def:adversary}
A Bitcoin adversary is an agent who forces the Bitcoin network to fail on some task.
\end{definition}
\noindent We use this definition of a Bitcoin adversary as the foundation for our adversarial model, clarifying the specific task we consider, what it means for the Bitcoin network to fail on performing that task, and how an agent can force such failure.

We identify the task performed on the Bitcoin network relevant to our adversarial model and the conditions for deeming this task to have failed. 
The specific task that we consider is the prevention of a 51\% attack.
We consider the Bitcoin network to be vulnerable to a 51\% attack if the effective threshold for a 51\% attack decreases to 40\%~(\ref{eqn:ghashfrachashrate}) or below.
This threshold corresponds to a stale rate satisfying~(\ref{eqn:stalerate51attack})
\begin{equation}
    \label{eqn:proposition}
    p_\text{stale} \geq \nicefrac1{3},
\end{equation}
which addresses the GHash.io incident~\cite{guardian2014}.

We identify the adversaries in our model and assess how they affect Bitcoin's security.
The adversaries in our model are Alice and Bob, the quantum miners. 
Although Alice and Bob are rational and intelligent quantum miners who are competing in a quantum race (Def.~\ref{def:quantumrace}), the optimal quantum mining strategies they employ increase the Bitcoin network's stale rate via Sattath's AQMS. 
This is not a deliberate attack by Alice and Bob, but is an emergent effect of their optimal quantum mining strategies.
We formalize Alice's and Bob's impact on the Bitcoin network with the following proposition.
\begin{proposition}
\label{prop:quantumminingthreshold}
Given the Grover iteration budget~$\kappa$ (\ref{eqn:groveriterbudget}), two rational and intelligent quantum miners in the Bitcoin network employing optimal quantum mining strategies cannot increase~$p_\text{stale}$ to satisfy~(\ref{eqn:proposition})
\begin{equation}
\label{eqn:stalethird}
p_\text{stale} \geq \nicefrac13
\end{equation}
over the duration of a 51\% attack.
\end{proposition}
\noindent We prove this proposition in~\S\ref{sec:results}.
\subsection{Mathematics}
\label{subsec:mathematics}

We discuss the mathematics required to describe our model within a game-theoretic framework.
We begin by discussing the pure strategies available to Alice and Bob and how these pure strategies label Alice's and Bob's payoff matrices.
Then, we provide the mathematics required to derive Alice's and Bob's payoff matrices in the case where they each perform multiple quantum measurements between successive blocks.
Finally, we explain the mathematics required to include the terms related to Sattath's AQMS into Alice's and Bob's payoff matrices.
\subsubsection{Pure strategies}

We discuss the pure strategies available to the quantum miners in our model.
We begin by outlining the set of actions available to the quantum miners for a single quantum measurement. 
Then, we explain how Alice's and Bob's pure strategies are derived from the single quantum measurement case. 
Next, we show how Alice's and Bob's pure strategies label their payoff matrices.

We discuss Alice's and Bob's pure strategies for a single quantum measurement.
Alice and Bob perform~$c$ quantum measurements (Def.~\ref{def:qmeas}) for each instance of the quantum race they are participating in.
For each of these measurements, the set of actions available to Alice and Bob,
\begin{equation}
\label{eqn:strategysetindiv}
S_0 := \left[k_\text{opt}^{32}\right]
,\end{equation}
comprises all possible numbers of Grover iterations that can be performed before the quantum measurement.
Note that this set differs from Eq.~(\ref{eqn:Sikopt}) by including 0, which corresponds to the action of doing nothing, i.e., performing zero Grover iterations.
Each action label~$k \in S_0$ has an associated probability~$p(k)$~(\ref{eqn:groversuccessbitcoin}), which is the probability that measuring the quantum state after performing~$k \in S_0$ Grover iterations is successful.
These probabilities form an increasing sequence
\begin{equation}
\label{eqn:purestrategyincreasingseq}
p(0) < p(1)< \dots < p\left(k_\text{opt}^{32}\right),
\end{equation}
which is a requirement of the quantum algorithms studied by the quantum races model~\cite{LRS19}.

Now we derive Alice's and Bob's pure strategies based on the set of available actions for a single quantum measurement~(\ref{eqn:strategysetindiv}).
Both Alice and Bob perform~$c$ quantum measurements for each instance of the quantum race they are competing in. 
For each quantum race and given the number of quantum measurements~$c$~(\ref{eqn:c}), the set of pure strategies is
\begin{equation} 
\label{eqn:strategysetrestarts}
S_0^c
:=\left\{\bm{k}\in\left[k_\text{opt}^{32}\right]^c:
\mathds1\cdot \bm k
\leq\kappa\right\},\,
\mathds 1 \cdot \bm k := 
\sum_{i=1}^c k_i,
\end{equation}
which represents all~$c$-tuples of numbers of Grover iterations that could be performed by Alice and Bob across all~$c$ quantum measurements, where~$\kappa$ is the cumulative Grover-iteration limitation~(\ref{eqn:groveriterbudget}).
As both Alice and Bob have identical resources, they share the same set of pure strategies~$S_0^c$.
Now we express a special case of the standard stars-and-bars result~\cite{Feller65}
as a lemma for clarity,
noting that this lemma is quite straightforward to prove.
\begin{lemma}[Cardinality of~$S_0^c$]
\label{lemma:cardinality}
The cardinality of~$S_0^c$ is
\begin{equation}
    \label{eqn:Sccardinality}
    |S_0^c|=\binom{\kappa + c}{c}.
\end{equation}
\end{lemma}
\begin{proof}
Our proof proceeds by the following combinatorial argument. 
For
\begin{equation}
    k_{c+1} := \kappa - \sum_{i=1}^c k_i, \; k_{c+1} \geq 0,
\end{equation}
a slack variable,
\begin{equation}
    \label{eqn:inequalitySc}
    \sum_{i=1}^c k_i \leq \kappa
\end{equation}
is replaced by the equality
\begin{equation}
    \label{eqn:equalitySc}
    \sum_{i=1}^{c+1} k_i=\kappa.
\end{equation}
Each $c$-tuple~$\bm k$ satisfying Eq.~(\ref{eqn:inequalitySc}) corresponds to a unique $(c + 1)$-tuple satisfying Eq.~(\ref{eqn:equalitySc}).
The number of such $(c + 1)$-tuples is
\begin{equation}
    \label{eqn:S0ccardinality}
    |S_0^c|=\binom{\kappa + c + 1 - 1}{c + 1  - 1}=\binom{\kappa + c}{c},
\end{equation}
which follows from the standard stars-and-bars result~\cite{Feller65}.
\end{proof}

We now explain how we manage the computational cost imposed from the combinatorial growth of~$|S_0^c|$~(\ref{eqn:S0ccardinality}).
As~$|S_0^c|$~(\ref{eqn:S0ccardinality}) determines the dimension of Alice's and Bob's payoff matrices, this combinatorial growth renders the evaluation of Alice's and Bob's payoff matrices computationally expensive.
To address this issue, we introduce a heuristic in~\S\ref{subsec:methods} that restricts the dimension of Alice's and Bob's payoff matrices. 
The motivation for this heuristic comes from the success probability~$p(k)$~(\ref{eqn:groversuccessbitcoin}) of the FQS algorithm.
Since~$p(k)$ grows quadratically with~$k$ and is maximized with~$\kappa$ Grover iterations, performing a quantum measurement after executing the FQS algorithm with~$k \ll \kappa$ Grover iterations produces~$p(k) \ll 1$.
This low success probability suggests that if a quantum miner measures early~$(k < \kappa)$ to gain a temporal advantage, then that quantum miner will still perform close to~$k \lesssim \kappa$ Grover iterations. 
Motivated by this, we restrict our attention to a subset of pure strategies for Alice and Bob such that the dimension is bounded by 
\begin{equation}
    \label{eqn:dimensionheuristic}
    \left|\left(S_0^{c}\right)^{\prime}\right| \leq \binom{n_q + 1+ c}{c} < \left|S_0^c\right|,
\end{equation}
which remains combinatorial in~$c$ but is substantially smaller for~$n_q \ll \kappa$.
This heuristic significantly lowers the computational cost of evaluating Alice's and Bob's payoff matrices when computing optimal quantum mining strategies.

Now we explain how Alice's and Bob's pure strategies label their payoff matrices.
Each pure strategy
\begin{equation}
    \label{eqn:purestrategyk}
    \bm{k}=(k_1, k_2, \dots, k_c) \in S_0^c
\end{equation}
labels the rows and columns of Alice's and Bob's payoff matrices.
This labelling implies that Alice's and Bob's payoff matrices take the form of a higher-dimensional tensor.
As the Lemke-Howson algorithm~\cite{LH64} requires payoff matrices as input, we label the rows and columns of Alice's and Bob's payoff matrices in the following way to preserve the matrix structure. 
Let
\begin{equation}
    \label{eqn:S0clexicographic}
    S_0^c=\left\{\bm{k}^{(1)}, \bm{k}^{(2)}, \dots, \bm{k}^{\left|S_0^c\right|} \right\}
\end{equation}
be the elements of~$S_0^c$~(\ref{eqn:strategysetrestarts}) in lexicographic ordering,
with
\begin{equation}
    \label{eqn:kj}
    \bm{k}^{(j)}
\end{equation}
being the~$j$th pure strategy~(\ref{eqn:purestrategyk}) of~$S_0^c$ in this ordering.
We index the rows and columns of Alice's and Bob's payoff matrices by
\begin{equation}
    \label{eqn:alphabeta}
    \alpha, \beta \in \left[\left|S_0^c\right|\right], \; \alpha \mapsto \bm{k}^{(\alpha)}, \; \beta \mapsto \bm{k}^{(\beta)}
\end{equation}
so that 
\begin{equation}
    \label{eqn:Aalphabetarowcol}
    A(\alpha, \beta)=A\left(\bm{k}^{(\alpha)}, \bm{k}^{(\beta)}\right)
\end{equation}
is the payoff Alice receives when she plays strategy~$\bm{k}^{(\alpha)}$ and Bob plays~$\bm{k}^{(\beta)}$.
Alice is the row player in her own payoff matrix, and Bob is the column player. 
Conversely, in Bob's payoff matrix, he is the row player, and Alice is the column player.
As the elements of~$S_0^c$ label the rows and columns of Alice's and Bob's payoff matrices, both Alice's and Bob's payoff matrices have dimension
$\left|S_0^c\right| \times \left|S_0^c\right|$~(\ref{eqn:Sccardinality}).
\subsubsection{Payoff matrices}

We describe the structure of Alice's and Bob's payoff matrices for our model.
We start by discussing how we express Alice's and Bob's total payoff matrices as the sum of payoff matrices corresponding to each of their~$c$ quantum measurements. 
Next, we derive the structure of Alice's and Bob's payoff matrices for a single quantum measurement. 
Then, we extend this logic to derive the structure of Alice's and Bob's payoff matrices for~$c$ quantum measurements.

Now we describe the structure of Alice's and Bob's payoff matrices based on their~$c$ quantum measurements.
We first describe this structure at a high level and then later specify the payoff functions that determine the payoff matrix entries. 
First consider that Alice and Bob play the strategies
\begin{equation}
\label{eqn:abtuples}
\underbrace{\left(a_1, a_2, \dots, a_c\right)}_a, \;\; \underbrace{\left(b_1, b_2, \dots, b_c\right)}_b,
\end{equation}
respectively, 
meaning that Alice executes Alg.~\ref{alg:grover} with~$a_1$ Grover iterations and gets~$x_1$~(\ref{eqn:measurement}), then runs Alg.~\ref{alg:grover} again for~$a_2$ Grover iterations to get~$x_2$ and so on, until the~$c$th quantum measurement.
For that final measurement, she gets~$x_c$.
Thus, she obtains the string of measurement outcomes
\begin{equation}
\label{eqn:alicemeasurementoutcomes}
\left(x_i^A(a)\right)_{i=1}^c
:=\left(x_1^\text{A}(a), x_2^\text{A}(a), \dots ,x_c^\text{A}(a)\right)
\end{equation}
with~$x_i(a)$
the $i^\text{th}$ measurement outcome depending on action~$a_i$ of strategy~$a$~(\ref{eqn:abtuples}).
Similarly, Bob executes Alg.~\ref{alg:grover} under his strategy~$b$~(\ref{eqn:abtuples}) to obtain the string of measurement outcomes 
\begin{equation}
\label{eqn:bobmeasurementoutcomes}
\left(x_i^B(b)\right)_{i=1}^c
:=
\left(x_1^\text{B}(b), x_2^\text{B}(b), \dots, x_c^\text{B}(b)\right),
\end{equation}
which is a tuple of the same length as Eq.~(\ref{eqn:alicemeasurementoutcomes}).
\par 
Here, we refer to the $i^\text{th}$ quantum measurement as Alice and Bob executing~$a_i \in a$ and $b_i \in b$ Grover iterations, respectively, and obtaining~$x_i^\text{A}$ and~$x_i^\text{B}$ for $i \in [c]$~(\ref{eqn:[x]}).  
For each~$i \in [c]$,  
\begin{equation}
\label{eqn:randomvariables}  X_i^\text{A}(a) := 
\begin{cases}
    1 & x_i^A(a) \in M \\
    0 & \text{otherwise}
\end{cases}, \;\;
X_i^\text{B}(b) := 
\begin{cases}
    1 & x_i^B(b) \in M \\
    0 & \text{otherwise}
\end{cases},
\end{equation}
are random variables indicating whether Alice's or Bob's $i^\text{th}$ quantum measurement yields a valid proof-of-work. 
In our model, Alice's quantum measurements are independent of each other, and the same is true for Bob. 
That is, 
\begin{equation}
X_i^\text{A}(a) \indep X_j^\text{A}(a), \; X_i^\text{B}(b)\indep X_j^\text{B}(b) \; \text{for }i \neq j,
\end{equation}
where~$\indep$ denotes probabilistic independence.
In general,
Alice's and Bob's measurements are not independent of each other;
i.e., 
\begin{equation}
    X_i^\text{A}(a) \not\!\indep X_j^\text{B}(b), \;\; X_i^\text{B}(b) \not\!\indep X_j^\text{A}(a),
\end{equation}
as the outcomes of Alice's and Bob's quantum measurements can be correlated through Sattath's AQMS. \par
We now describe how Alice's and Bob's payoff matrices are expressed as sums of payoff matrices corresponding to each of their~$c$ quantum measurements.
\begin{example}[Additive structure of payoff matrices]
Suppose that for~$c=2$ Alice and Bob play the strategies
\begin{equation}
    \label{eqn:alicebobstratexample}
    a=(10, 20) \;\; b=(15, 25)
\end{equation}
respectively. 
Excluding Sattath's AQMS, Alice's payoff is determined from two disjoint events. 
Alice either succeeds after executing 10 Grover iterations, or she fails after 10 but succeeds after 20 whereas Bob fails after 15. 
These two events represent Alice's first and second quantum  measurements, respectively, and the payoff Alice receives from each contributes to her total payoff for the given strategy profile~$(a, b)$.
Thus, Alice's total payoff can be expressed as a sum of the payoffs she receives from each quantum measurement. 
\end{example}
Given that Alice's quantum measurements are independent of each other, and the same is true for Bob, we associate to each quantum measurement~$i \in [c]$~(\ref{eqn:[x]}) a pair of payoff matrices (Def.~\ref{def:payoffmatrix}), $A_i$ and $B_i$, that give Alice's and Bob's respective payoffs for their~$i^\text{th}$ quantum measurement. 
In other words, we express Alice's and Bob's total payoff matrices as 
\begin{equation}
    \label{eqn:payoffmatrixsplit}
    A := \sum_{i=1}^c A_i, \;\; B := \sum_{i=1}^c B_i.
\end{equation}
Furthermore, we express each~$A_i$ as 
\begin{equation}
    \label{eqn:payoffmatrixdecomposition}
    A_i := A_i^{\text{meas}} + A_i^\text{AQMS}, \;\; B_i := B_i^{\text{meas}} + B_i^\text{AQMS},
\end{equation}
where
\begin{equation}
\label{eqn:ABmeas}
A_i^\text{meas},\,
B_i^\text{meas}
\end{equation}
are the payoff matrices encoding Alice's and Bob's payoffs from their~$i^\text{th}$ quantum measurements alone, and 
\begin{equation}
\label{eqn:ABAQMS}
A_i^\text{AQMS},\,
B_i^\text{AQMS}
\end{equation}
are the payoff matrices encoding Alice's and Bob's payoffs from their deployment of Sattath's AQMS~\cite{Sat20} on their~$i^\text{th}$ quantum measurement.
Consequently, the payoff functions~$u(a, b), v(a, b)$ (Remark~\ref{remark:uvnotation}) that determine the entries of~$A$ and~$B$ (Def.~\ref{def:payoffmatrix}) have the same additive structure. 
That is, 
\begin{equation}
    \label{eqn:payofffunctionsplit}
    u(a, b; D, \gamma_\text{S}) := \sum_{i=1}^c u_i(a, b; D, \gamma_\text{S}), \;\;
    v(a, b; D, \gamma_\text{S}) := \sum_{i=1}^c v_i(a, b; D, \gamma_\text{S}),
\end{equation}
with 
\begin{equation}
    \label{eqn:uivi} 
    u_i(a, b; D, \gamma_\text{S}) := u_i^\text{meas}(a, b;D) + u_i^\text{AQMS}(a, b; D, \gamma_\text{S}), \;\; v_i(a, b; D, \gamma_\text{S}) := v_i^\text{meas}(a, b;D) + v_i^\text{AQMS}(a, b; D, \gamma_\text{S}),
\end{equation}
giving Alice's and Bob's payoffs from their~$i^\text{th}$ quantum measurement and from the deployment of Sattath's AQMS, respectively, for the strategy profile~$(a, b)$~(\ref{eqn:abtuples}).

We discuss the mathematics of incorporating Alice's and Bob's~$c$ quantum measurements into their payoff matrices.
To formulate the payoff matrices for Alice and Bob in this setting, we extend the framework for a single quantum measurement established for S2QR (Def.~\ref{def:symmetric2QR})~\cite{LRS19}.
For a single quantum measurement, Alice's and Bob's payoffs depend on the number of Grover iterations performed by each quantum miner.
For example,
if Alice performs fewer Grover iterations than Bob, her payoff is simply the probability that her quantum measurement is successful
\begin{equation}
    \label{eqn:alicemeasurementsuccess}
    P(X_\text{A}=1),
\end{equation}
with
\begin{equation}
    \label{eqn:P}
    P(X)
\end{equation}
denoting the probability that event~$X$ occurs. 
In the context of the FQS algorithm, this probability is~$p(k)$~(\ref{eqn:groversuccessbitcoin}), which we remind is implicitly dependent on~$D$.
Conversely,
if Bob performs the same number of Grover iterations as Alice or fewer, then Alice's payoff is the joint probability that her measurement is successful whereas Bob's measurement failed,
\begin{equation}
\label{eqn:successfail}
P(X_i^\text{A}(a)=1, X_i^\text{B}(b)=0)=P(X_i^\text{A}(a)=1)P(X_i^\text{B}(b)=0),
\end{equation}
where we factorise the probabilities because Sattath's AQMS does not contribute to these terms.
We note that as these probabilities are  dependent on the value of the network difficulty~$D$~ (\ref{eqn:groversuccessbitcoin}), $u_i^\text{meas}(a, b)$, $u_i^\text{AQMS}(a, b)$, $v_i^\text{meas}(a, b)$, and~$v_i^\text{AQMS}(a, b)$ are  functions of~$D$ too.

For multiple quantum measurements, we extend the logic of a single quantum measurement to account for all possible outcomes. 
For Alice playing strategy~$a$ and Bob playing strategy~$b$~(\ref{eqn:abtuples}), we consider the entries of~$A_i^\text{meas}$ and $B_i^\text{meas}$ to determine~$A^\text{meas}$ and $B^\text{meas}$~(\ref{eqn:ABmeas}).
The cumulative number of Grover iterations performed by Alice and Bob up to their~$i^\text{th}$ quantum measurement is 
\begin{equation}
\label{eqn:ABcumulative}
\mathcal{A}_i := \sum_{j=1}^i a_j, \;\;
\mathcal{B}_i := \sum_{j=1}^i b_j, \;\; \mathcal A_0 := 0, \;\; \mathcal B_0 := 0,
\end{equation}
respectively.
Furthermore, 
\begin{equation}
\label{eqn:boblastmeasurement}
    [c]_0 \ni \ell_\text{A}(a, b) := \max\{j \in [c]_0: \mathcal{A}_j < \mathcal{B}_i\}, \;\; [c]_0 \ni \ell_\text{B}(a, b) := \max\{j \in [c]_0: \mathcal{B}_j < \mathcal{A}_i\}
\end{equation}
are the indices in Alice's and Bob's strategies~(\ref{eqn:abtuples}) that indicate the last quantum measurement performed by Alice or Bob before the other's~$i^\text{th}$ quantum measurement.
For Alice's~$i^\text{th}$ quantum measurement, the entry of~$A_i^\text{meas}$ corresponding to the strategy profile
$(a, b)$~(\ref{eqn:abtuples}) is
\begin{equation}
    \label{eqn:alicepayoffrestartpart}
    u_i^\text{meas}(a,b; D)
    :=P\left(X_i^\text{A}(a)=1\right) 
    \prod_{r=1}^{i-1} P\left(X_r^\text{A}(a)=0\right)\prod_{r=1}^{\ell_\text{B}(a, b)} P\left(X_r^\text{B}(b)=0\right)
\end{equation}
with~$\ell_\text{B}$
dependent on $b$~(\ref{eqn:boblastmeasurement}) and~$P$ implicitly dependent on~$D$. \par
The entry of $A_i^\text{meas}$~(\ref{eqn:alicepayoffrestartpart}) is the joint probability that Alice's $i^\text{th}$ quantum measurement is successful, but failed her first $i-1$ quantum measurements, and Bob failed his first~$\ell_\text{B}$ quantum measurements.
The analogous entry of~$B_i^\text{meas}$,
\begin{equation}
\label{eqn:bobpayoffrestartpart}  
    v_i^\text{meas}(a, b; D) := 
    P\left(X_i^\text{B}(b)=1\right) \prod_{r=1}^{i-1} P\left(X_r^\text{B}(b)=0\right) \prod_{r=1}^{\ell_A(a, b)} P\left(X_r^\text{A}(a)=0\right),
\end{equation}
reflects the same logic applied to Bob's~$i^\text{th}$ quantum measurement. 
Thus, our formulation captures the cumulative nature of our model and extends the S2QR framework~\cite{LRS19} to multiple quantum measurements.
\subsubsection{AQMS}

Now we discuss the mathematics required to describe the entries~$A_i^{\text{AQMS}}$ and~$B_i^\text{AQMS}$.
We first explain how Sattath's AQMS affects Alice's and Bob's payoffs based on the timings of Alice's and Bob's most recent quantum measurements relative to each other. 
Next, we consider the special case where Alice and Bob perform a quantum measurement at the same time. 
Then, we discuss additional AQMS terms and provide our justification for their omission from Alice's and Bob's total payoff.

Now we explain the mathematics required to describe the entries of~$A_i^\text{AQMS}$ and~$B_i^\text{AQMS}$~(\ref{eqn:payoffmatrixdecomposition}).
When a quantum miner employs Sattath's AQMS, they perform a quantum measurement with fewer Grover iterations~$k$ than in their chosen strategy.
\begin{fact}
\label{fact:lB<c}
As Alice and Bob perform only~$c$ quantum measurements per our model, Bob cannot respond to Alice with Sattath's AQMS if he has exhausted all~$c$ quantum measurements, so~$\ell_\text{B}<c$.
\end{fact}
\noindent
For Alice's~$i^\text{th}$ quantum measurement, the entry of~$A_i^\text{AQMS}$ corresponding to the strategy profile~$(a, b)$
\begin{equation}
    \label{eqn:alicepayoffAQMSterm}
    u_i^\text{AQMS}(a, b; \gamma_\text{S}) := 
    \gamma_\text{S} P\left(X_i^\text{A}(a)=1\right) \left(\prod_{r=1}^{i-1} P\left(X_r^\text{A}(a)
   =0\right)\right)\left(\prod_{r=1}^{\ell_\text{B}(a, b)} P\left(X_r^\text{B}(b)=0\right)\right) P\left(X^\text{B}_{\text{AQMS}}(b)=1\right),
\end{equation}
for~$\gamma_\text{S}$
Sattath's propagation parameter (Def.~\ref{def:sattathspropparameter}),
is the joint probability that Alice succeeds on her~$i^\text{th}$ quantum measurement after failing all previous~$i-1$ quantum measurements, 
whereas Bob fails his first~$\ell_\text{B}$ quantum measurements but then adopts Sattath's AQMS and succeeds after performing~$\mathcal{A}_i - \mathcal{B}_{\ell_\text{B}}$ Grover iterations.
Here,
\begin{equation}
    \label{eqn:XBAQMS}
    X^\text{B}_\text{AQMS}(b)
\end{equation}
is a random variable denoting the outcome of Bob's measurement after employing Sattath's AQMS.

We consider now the special case for which Alice and Bob perform a measurement at the same time as a result of their chosen strategies.
In other words, $\mathcal{A}_i=\mathcal{B}_j$, so
\begin{equation}
\label{eqn:alicepayoffAQMStermsametime}
u_i^\text{AQMS}(a, b, \gamma_\text{S})=
\gamma_\text{S} P\left(X_i^\text{A}(a)=1\right) P\left(X_j^\text{B}(b)=1\right)\left(\prod_{r=1}^i P(X^\text{A}_r(a)=0)\right)\left(\prod_{r=1}^j P(X^\text{B}_r(b)=0)\right) 
\end{equation}
is the entry of~$A_i^\text{AQMS}$ corresponding to the strategy pair~$(a, b)$ in this case.  
The entries of~$B_i^\text{AQMS}$ are derived analogously by swapping the roles of Alice and Bob~(\ref{eqn:alicepayoffAQMSterm}, \ref{eqn:alicepayoffAQMStermsametime}) and replacing~$\gamma_\text{S}$ with $1 - \gamma_\text{S}$.

We explain why certain payoffs related to Sattath's AQMS are omitted in our model.
We note that, in theory, the entries of~$A_i^\text{AQMS}$~(\ref{eqn:alicepayoffAQMSterm}) should also include the payoff Alice receives after employing Sattath's AQMS~\cite{Sat20} when Bob yields a proof-of-work. 
This additional payoff, 
\begin{equation}
\label{eqn:additionalAQMSterm}
\gamma_\text{S} P\left(X_i^\text{B}(b)=1\right) \left(\prod_{r=1}^{i-1} P(X_r^\text{B}(b)=0)\right) \left(\prod_{r=1}^{\ell_\text{A}(a, b)}P(X_r^\text{A}(a)=0)\right) P(X_\text{AQMS}^\text{A}(a)=1),
\end{equation}
is structurally identical to Eq.~(\ref{eqn:alicepayoffAQMSterm}) but with the labels for Alice and Bob reversed.
However, incorporating this payoff into the entries of~$A_i^\text{AQMS}$ (and similarly for~$B_i^\text{AQMS}$) would require knowing the total ordering of Alice's and Bob's quantum measurements.
This knowledge is required because the number of Grover iterations Alice performs when deploying Sattath's AQMS depends on the timing of Bob's quantum measurements relative to all of Alice's. 
The structure of Alice's and Bob's payoff matrices~(\ref{eqn:payoffmatrixsplit}) either encodes the relative timings of Alice's quantum measurements with respect to Bob's in Alice's payoff matrix, or vice versa for Bob, but not both. 
As the total number of~$c$-tuple orderings grows as $\binom{2c}{c}$ for each of Alice and Bob, a closed-form expression for Alice's and Bob's payoff matrices is complicated and unnecessary, so we instead focus on a numerical solution.

We now discuss the impact of omitting these additional payoffs from Alice's and Bob's payoff matrices on our results.
Omitting these payoffs alters the specific optimal quantum mining strategies we compute. 
The entries of~$A_i^\text{AQMS}$ and~$B_i^\text{AQMS}$ account for Bob deploying Sattath's AQMS in Alice's payoff matrix, and Alice deploying Sattath's AQMS in Bob's payoff matrix, but not the reverse in either case. 
We expect these additional payoffs to be higher-order corrections, and,
although we have not formally studied their impact, their exclusion allows us to write a compact closed-form expression for Alice's and Bob's payoff matrices. 
We leave an analysis including all terms in Alice's and Bob's payoff matrices for future work.
\subsection{Methods}
\label{subsec:methods}

We outline the methods used to compute optimal mining strategies for Alice and Bob and the subsequent 51\% threshold.
We begin by discussing how the entries of Alice's and Bob's payoff matrices are computed and how we manage the computationally prohibitive dimension of these matrices.
Next, we explain how the optimal quantum mining strategies are computed. 
Then, we describe the simulations used to estimate how the employment of these optimal quantum mining strategies affect Bitcoin's security.
\subsubsection{Computing payoff matrices}

We describe how we compute Alice's and Bob's payoff matrices. 
We begin by explaining how the individual entries of each matrix are computed.
Next, we describe the heuristic used to sample Alice's and Bob's pure strategies. 
Finally, we explain how Alice's and Bob's payoff matrices are computed based on sampling.

We first describe how we compute the individual entries of Alice's and Bob's payoff matrices.
Each entry of Alice's payoff matrix corresponds to a strategy profile~$(a, b) \in S_0^c \times S_0^c$~(\ref{eqn:strategysetrestarts}) and is computed by evaluating a function
\begin{equation}
    \label{eqn:f_A}
    f_A : S_0^c \times S_0^c \times \mathbb{R}^+ \times [0, 1] \to \mathbb{R} : \left(a, b, D, \gamma_\text{S}\right) \mapsto \sum_{i=1}^c \left(u_i^\text{meas} \left(a, b; D\right) + u_i^\text{AQMS}\left(a, b; D, \gamma_\text{S}\right)\right)
\end{equation}
that returns a 64-bit float representing Alice's payoff for the corresponding strategy profile~$(a, b)$ given network difficulty~$D$ and Sattath's propagation parameter~$\gamma_\text{S}$. 
The corresponding entry for Bob's payoff matrix is computed by evaluating~$f_A(b, a, D, 1 - \gamma_\text{S})$.

Now we discuss how we sample Alice's and Bob's pure strategies when evaluating their payoff matrices.
As the dimension of Alice's and Bob's payoff matrices grows combinatorially with~$c$~(Lemma~\ref{lemma:cardinality}), evaluating Alice's and Bob's payoff matrices in their entirety is computationally expensive. 
To reduce their dimensionality, we apply a heuristic to the set of pure strategies for Alice and Bob~(\ref{eqn:strategysetrestarts}).
Specifically, we construct a subset of Grover iteration counts from~$S_0$~(\ref{eqn:strategysetindiv}) using a quadratic spacing, 
\begin{equation}
    \label{eqn:groveriterquadrsubset}
    S_0^\prime := \left\{\left\lfloor 1 - \left(1 - \frac{i}{n_q}\right)^2 k_\text{opt}^{32}\right\rfloor : i \in [n_q]_0\right\} \subseteq S_0,
\end{equation}
with~$n_q$ being the number of Grover iteration counts sampled from~$S_0$.
With this, Alice's and Bob's reduced pure strategies,
\begin{equation}
    \label{eqn:purestrategyquadrsubset}    
    (S_0^c)^\prime := \left\{\bm k \in (S_0^\prime)^c : \mathds{1} \cdot \bm k  \leq \kappa\right\} \subseteq S_0^c,
\end{equation}
are the subset of~$c$-tuples of Grover iterations drawn from~$S_0^\prime$~(\ref{eqn:groveriterquadrsubset}) that are permitted under the cumulative Grover iteration budget~$\kappa$.
To further reduce the dimensionality of Alice's and Bob's payoff matrices, we also exclude pure strategies that are strictly dominated (Def.~\ref{def:dominatedstrategy}), as a rational player will not play them.

Now we describe how the sampled pure strategy spaces are used to compute Alice's and Bob's payoff matrices.
Using the sampled pure strategy spaces, we evaluate $f_A(a, b; D, \gamma_\text{S})$ for each strategy profile
\begin{equation}
\label{eq:strategyprofile}
(a, b) \in (S_0^c)^\prime \times (S_0^c)^\prime,
\end{equation}
and similarly for Bob.
We note that we decrease the number of sampled Grover iteration counts~$n_q$ as~$c$ increases to ensure that evaluating Alice's and Bob's payoff matrices is computationally manageable. 
The resulting payoff matrices have dimensions that are small enough such that the payoff matrices can be stored entirely in memory (requiring only a few gigabytes), further allowing the Lemke-Howson algorithm~\cite{LH64, NRTV07} to finish executing within a practical amount of time (on the order of hours). 
\subsubsection{Optimal quantum mining strategy computation}

We explain how the optimal quantum mining strategies are computed using Alice's and Bob's payoff matrices.
We start by describing how the Lemke-Howson algorithm is implemented using the Gambit library.
Next, we discuss the specific parameters used for computing optimal quantum mining strategies via this algorithm.
Then, we explain how we improve the precision of the computed Nash equilibria via Gambit's \texttt{Rational} class.

We discuss how the Lemke-Howson algorithm is implemented in our code.
To compute the optimal quantum mining strategies for Alice and Bob, we use the Lemke Howson algorithm implemented as the \texttt{gambit-lcp} method in the \texttt{pygambit} library~\cite{SavTur25}.
We use Alice's and Bob's payoff matrices as input into this method.
Due to the exponential worst-case runtime of the Lemke-Howson algorithm~\cite{LH64}, we limit the number of Nash equilibria computed with the \texttt{stop\_after} parameter and compute only one exact Nash equilibrium for each value of~$c$. \par

We now describe the parameters used in computing the optimal quantum mining strategies.
We restrict ourselves to~$c \leq 4$, since for larger values, the runtime of the Lemke-Howson algorithm becomes prohibitively long, often taking days to compute a single Nash equilibrium.
For computing the optimal quantum mining strategies, the number of sampled Grover iteration counts~$n_q$~(\ref{eqn:groveriterquadrsubset}) was chosen so that Alice's and Bob's payoff matrices had similar dimensions for~$c \in [2, 3, 4]$. 
A smaller dimension was used for~$c=1$, which was sufficient for comparing optimal quantum mining strategies to this base case.
Computing Nash equilibria for~$c > 4$ requires decreasing the value of~$n_q$ to a degree such that~$S_0^\prime$~(\ref{eqn:groveriterquadrsubset}) is no longer representative of the full pure strategy space.
We compute Nash equilibria for the simplest case for~$c=1$ to serve as a baseline for comparison with the multiple-measurement cases. 
For the low-difficulty regime, we assign
\begin{equation}
\label{eqn:assignDlow}
D\gets3\,651\,011.631,
\end{equation}
which is the value of the network difficulty from 26 February 2013~\cite{difficulty}.
For the high-difficulty and ideal regimes, we assign
\begin{equation}
\label{eqn:assignDhigh}
D\gets119\,116\,256\,505\,723.5,
\end{equation}
which is the value from 9 May 2025~\cite{difficulty}.
We note that it is expected that the network difficulty will continue to increase beyond this value due to technological advancements in mining hardware~\cite{Ant14}.
For all computed Nash equilibria, we assign
\begin{equation}
\gamma_\text{S} \gets 0.5,
\end{equation}
meaning that Alice and Bob have equal probability of their blocks being appended to the longest branch of the blockchain during a temporary fork.
This choice is motivated by the fact that Alice and Bob have identical resources.

We explain how Gambit's \texttt{Rational} class is used to reduce floating-point error in Nash equilibria computation.
To improve the precision of the computed Nash equilibria,
we set the parameter \texttt{rational=True} in the \texttt{gambit-lcp} solver, enabling the solver to use Gambit's built-in \texttt{Rational} class in the computation~\cite{SavTur25}.
Gambit's \texttt{Rational} class implements an exact representation of rational floating-point numbers by storing them as numerator-denominator pairs, eliminating the rounding errors inherent to their binary representations.
Although using the \texttt{Rational} class increases the execution time of the Lemke-Howson algorithm, it significantly improves the precision of the computed Nash equilibria.
\subsubsection{Quantum mining simulation}

We discuss the quantum mining simulation used to determine how the optimal quantum mining strategies affect Bitcoin's security against a 51\% attack.
We begin by giving an overview of the simulation and its parameters.
Then, we discuss how the simulation handles the timings of quantum measurement and classical block-arrival events.
Finally, we discuss our method for studying how the network difficulty affects~$p_\text{stale}$ under the deployment of optimal quantum mining strategies.

We explain how we simulate the effect of Alice's and Bob's optimal quantum mining strategies on~$p_\text{stale}$. 
We perform a Monte Carlo simulation parameterized by a Nash equilibrium computed by the Lemke-Howson algorithm, as well as~$\gamma_\text{S}$ and~$D$.
The simulation begins at time zero and draws the next classical block arrival time from an exponential distribution with rate~$\lambda=\nicefrac1{600}$ seconds (Remark~\ref{remark:arrivaltimes}).
While searching for a valid block, the number of Grover iterations performed by Alice and Bob for each of their~$c$ quantum measurements is sampled from their respective mixed strategies in the Nash equilibrium.
We record the timing of quantum measurements, which occur either according to Alice's and Bob's chosen strategies or as a result of Sattath's AQMS, along with the classical block arrival times.
For every new block that is appended to the blockchain, we repeat the process of simulating the mining of the next block by resampling Alice's and Bob's strategies from the Nash equilibrium and generating the next classical block arrival time. 
This cycle of simulating the next block in the blockchain is repeated until 86400 simulation seconds (1 day) have elapsed. \par
We simulate~$n_s=1\,000\,000$ days.
For each simulated day~$i \in [n_s]$~(\ref{eqn:[x]}), let~$b_i$ denote the number of blocks mined on day~$i$, $f_i$ denote the number of forks observed on day~$i$, and 
\begin{equation}
    \label{eqn:pstalei}
    p_\text{stale}^i := \nicefrac{f_i}{b_i} 
\end{equation}
the stale rate on day~$i$.
From the empirically-observed distribution of~$p_\text{stale}^i$, we compute the 95th and 99th percentiles,
\begin{equation}
    \label{eqn:P95P99}
    P_{95}, \;\; P_{99},
\end{equation}
and the empirical probability
\begin{equation}
    \label{eqn:empprob}
    P_{>\nicefrac13}
    :=
    \frac{1}{n_s} \sum_{i=1}^{n_s} \mathbf{1}_{p_\text{stale}^i > \nicefrac{1}{3}}  
\end{equation}
that any~$p_\text{stale}^i$ exceeds the~$\nicefrac{1}{3}$ benchmark, with~$\mathbf{1}_C$ being the indicator function for condition~$C$.
Furthermore, we also compute the empirical probability
\begin{equation}
    \label{eqn:pdet}
    P_\text{det} := \frac{1}{n_s} \sum_{i=1}^{n_s} \mathbf{1}_{q_i < 0.05},
\end{equation}
of detecting the presence of quantum mining, where
\begin{equation}
    \label{eqn:qi}
    q_i := P\left(X \geq f_i\right), \;\; X \sim \operatorname{Binomial}(b_i, 0.0024)
\end{equation}
is the p-value for day~$i$ under the null hypothesis that no quantum miners are present, with 0.0024 being the classical stale rate~\cite{SSJ+18} and~$X \sim \operatorname{Binomial}(b_i, 0.0024)$ following from block arrival times forming a Poisson process (Remark~\ref{remark:binomial}).

We describe how we handle the timing of quantum measurements and classical block arrival events in the Monte Carlo simulation. 
Alice's and Bob's chosen strategies are converted into quantum measurement times by assuming that each Grover iteration takes~$\nicefrac{600}{\kappa}$ seconds to execute.
We assume that the time required for state preparation, measurement, and constructing the relevant quantum-phase oracle is negligible and that Alice and Bob begin executing the FQS algorithm again immediately after a failed measurement. 
The quantum measurement times and the classical block arrival times are sorted to form a sequence of events.
We record which events belong to Alice, Bob, or the classical network, and handle each event accordingly.
If a quantum miner measures and is successful, Sattath's AQMS is triggered for the other quantum miner who is forced to measure with the Grover iterations they have accumulated up to that point.
If Sattath's AQMS results in a successful measurement, then a fork is recorded.
If the event is a classical block arrival, a block is appended to the blockchain unconditionally, and both quantum miners independently employ Sattath's AQMS.

To further study the impact that optimal quantum mining strategies would have on~$p_\text{stale}$, we also plot how different values of~$D$~(\ref{eqn:networkdifficulty}) affect~$p_\text{stale}$ for the values of~$c$ considered.
For each value of
\begin{equation}
\label{eqn:valuesD}
D \in \{2^{30}, 2^{31}, \dots, 2^{45}\}=:2^{[30,45]},
\end{equation}
and for each value of~$c \in [4]$, we compute an optimal quantum mining strategy via the Lemke-Howson algorithm and then simulate the deployment of this strategy within the Bitcoin network.
Here, we simulate the deployment of the optimal quantum mining strategies over 50\,000 days to reduce the overall computation time, and report both~$P_{95}$~(\ref{eqn:P95P99}) and~$P_\text{det}$~(\ref{eqn:pdet}).
We plot~$P_{95}$ as a function of~$D$ to study the~$p_\text{stale}$ that is exceeded on 5\% of the worst days of quantum mining, and~$P_\text{det}$ as a function of~$D$ to study how the network difficulty affects the probability of detecting quantum mining.
\section{Results}
\label{sec:results}

Now we present our results for how Alice and Bob should play the game including examples and show how this game affects Bitcoin's security against a 51\% attack.
First we construct closed-form expressions for Alice's and Bob's payoff matrices. 
Next, we present examples of optimal quantum mining strategies for Alice and Bob, which correspond to Nash equilibria.
Finally, we report the implications for Bitcoin's security if Alice and Bob employ an optimal strategy for each quantum race. 
\subsection{Payoff matrices for quantum miners}

Now we obtain the closed-form expressions for the payoff matrices used by Alice and Bob.
First, for didactic reasons, we derive the expressions for Alice's and Bob's payoff matrices using only the terms that account for their~$c$ quantum measurements and excluding the terms from Sattath's AQMS.
Then, we extend these expressions to incorporate the terms corresponding to Sattath's AQMS and derive expressions for Alice's and Bob's payoff matrices in our model.
\subsubsection{Payoff matrices for multiple measurements}

Having developed the mathematical machinery in~\S\ref{subsec:mathematics}, we now derive the expressions giving the entries of~$A_i^\text{meas}$ and~$B_i^\text{meas}$~(\ref{eqn:ABmeas}) and present an example.
We first discuss preliminaries explaining how we use the mathematics introduced in~\S\ref{subsec:mathematics} to obtain these two matrices.
Then, we derive a closed form of these payoff matrices. 
Finally, for didactic reasons, we illustrate an example of these entries for~$c=2$.

First, as a preliminary, we briefly clarify how the mathematics introduced in~\S\ref{subsec:mathematics}
relates to the general form of the entries of~$A_i^\text{meas}$ and~$B_i^\text{meas}$~(\ref{eqn:ABmeas}).
The rows and columns of Alice's and Bob's payoff matrices are indexed by~$\alpha$ and~$\beta$~(\ref{eqn:alphabeta}).
Thus, the~$(\alpha, \beta)$ entry~$A_i^\text{meas}$ and~$B_i^\text{meas}$~(\ref{eqn:ABmeas}) is given by evaluating the payoff functions~$u_i^\text{meas}\left(\bm{k}^{(\alpha)}, \bm{k}^{(\beta)}; D\right)$~(\ref{eqn:alicepayoffrestartpart}) and~$v_i^\text{meas}\left(\bm{k}^{(\alpha)}, \bm{k}^{(\beta)}; D\right)$~(\ref{eqn:bobpayoffrestartpart}) for the strategy profile~$\left(\bm{k}^{(\alpha)},\bm{k}^{(\beta)}\right)$ corresponding to~$(\alpha, \beta)$~(\ref{eqn:alphabeta}).  
Furthermore, as we enumerate all possible quantum measurement orderings from the strategy profile~$(a, b)$, the values of~$\ell_\text{A}(a, b)$ and~$\ell_\text{B}(a, b)$~(\ref{eqn:boblastmeasurement}) are determined by the quantum measurement ordering being considered.

We now derive a closed-form expression for the entries of Alice's and Bob's payoff matrices that correspond to their~$i^\text{th}$ quantum measurements,
$A_i^{\text{meas}}$ and $B_i^{\text{meas}}$~(\ref{eqn:payoffmatrixdecomposition}).
The~$(\alpha, \beta)$~(\ref{eqn:Aalphabetarowcol}) entry is
\begin{equation}
\label{eqn:alicepayoffrestartsmapping}
    A_i^\text{meas} : \left[\left|S_0^c\right|\right] \times \left[\left|S_0^c\right|\right]  \to [0, 1] :  \left(\alpha, \beta\right) \mapsto A_i^\text{meas}\left(\alpha, \beta\right),
\end{equation}
with the integers
\begin{equation}
    \label{eqn:kalphakbeta}
    \alpha \mapsto \bm{k}^{(\alpha)}=(a_1, a_2, \dots, a_c), \; \beta \mapsto \bm{k}^{(\beta)}=(b_1, b_2, \dots, b_c),
\end{equation}
mapping to a pair of pure strategies~(\ref{eqn:alphabeta}), and
\begin{equation}
    \label{eqn:payoffrestartgeneral}
    A_i^\text{meas}(\alpha, \beta) =
    \left(p(a_i) \prod_{r=1}^{i-1} (1 - p(a_r))\right)
    \begin{cases}
        1 & \mathcal{A}_i < b_1 \\
        (1 - p(b_1)) & b_1 \leq \mathcal{A}_i < b_1 + b_2 \\
        (1 - p(b_1))(1 - p(b_2)) & b_1 + b_2 \leq \mathcal{A}_i < b_1 + b_2 + b_3 \\
        \vdots \\
        \prod_{r=1}^{c-1} (1 - p(b_r)) & \mathcal{B}_{c-1} \leq \mathcal{A}_{i} < \mathcal{B}_c \\
        \prod_{r=1}^{c} (1 - p(b_r)) & \mathcal{A}_i \geq \mathcal{B}_c. \\
    \end{cases}
\end{equation}
Here, we remind that~$\mathcal{A}_i$ and $\mathcal{B}_i$~(\ref{eqn:ABcumulative}) are the cumulative number of Grover iterations performed by Alice and Bob up to their~$i^\text{th}$ quantum measurement, respectively.
As Alice and Bob have identical quantum and classical resources, the quantum race is symmetric~(Def.~\ref{def:symmetric2QR}).
Consequently, 
\begin{equation}
    \label{eqn:Bobpayoffrestart}
    B_i^\text{meas}=(A_i^\text{meas})^\intercal
\end{equation}
is Bob's payoff matrix from his~$i^\text{th}$ quantum measurement alone.

We illustrate the terms contributing to Alice's and Bob's payoff from their~$c$ quantum measurements~($A^\text{meas}, B^\text{meas}$) by examining the case where each quantum miner performs two quantum measurements. 
In this setting, Alice's and Bob's pure strategies are elements of~$S_0^2$~(\ref{eqn:strategysetrestarts}), which are 2-tuples~$(a_1, a_2)$ and~$(b_1, b_2)$. 
Here,
\begin{equation}
\label{eqn:ABmeasc2}
A^\text{meas}=A_1^\text{meas} + A_2^\text{meas},\;\; B^\text{meas}=B_1^\text{meas} + B_2^\text{meas}
\end{equation}
are the components contributing to Alice's and Bob's payoff matrix from their first and second quantum measurements, respectively. 
Focusing on Alice's payoff matrix, 
\begin{equation}
    \label{eqn:alicepayoffrestartA1}
    A_1^\text{meas}(\alpha, \beta) =
    \begin{cases}
    p(a_1) & a_1 < b_1 \\
    p(a_1)(1 - p(b_1)) & b_1 \leq a_1 < b_1 + b_2 \\
    p(a_1)(1-p(b_1))(1-p(b_2)) & a_1 \geq b_1 + b_2,
    \end{cases}
\end{equation}
is the~$(\alpha, \beta)$ element of~$A_1^\text{meas}$, whereas 
\begin{equation}
\label{eqn:alicepayoffrestartA2}
    A_2^\text{meas}(\alpha, \beta) =
    \begin{cases}
    p(a_2)(1 - p(a_1)) & a_1 + a_2 < b_1 \\
    p(a_2)(1 - p(a_1))(1 - p(b_1)) & b_1 \leq a_1 + a_2 < b_1 + b_2 \\
    p(a_2)(1 - p(a_1))(1-p(b_1))(1-p(b_2)) & a_1 + a_2 \geq b_1 + b_2.
    \end{cases}
\end{equation}
is the~$(\alpha, \beta)$ element of~$A_2^\text{meas}$.
As the quantum race is symmetric, the terms contributing to Bob's payoff are computed by taking the transpose of Alice's terms~(\ref{eqn:Bobpayoffrestart}).
\subsubsection{Payoff entries with Sattath's AQMS}

We build on the previous result~(\ref{eqn:payoffrestartgeneral}) by including Sattath's AQMS, which gives Alice and Bob another chance to win the quantum race in the event that their opponent wins first.
We first derive the expressions for Alice's and Bob's payoff matrices in this case.
Then, we show how the example is modified by choosing a particular value of~$c$.

As a preliminary, we describe how we build on the previous result to derive the general form of the entries of Alice's and Bob's payoff matrices.
The~$(\alpha, \beta)$ entry of~$A_i$~(\ref{eqn:payoffmatrixdecomposition}) is obtained by evaluating the payoff functions~$u_i^\text{meas}(\alpha, \beta; D)$~(\ref{eqn:alicepayoffrestartpart}) and~$u_i^\text{AQMS}(\alpha, \beta; D,\gamma_\text{S})$~(\ref{eqn:alicepayoffAQMSterm},\ref{eqn:alicepayoffAQMStermsametime}) with $\gamma_\text{S}$ being Sattath's propagation parameter (Def.~\ref{def:sattathspropparameter}) for the strategy profile~$(\bm k^{(a)}, \bm k^{(b)})$ corresponding to~$(\alpha, \beta)$~(\ref{eqn:kalphakbeta}) for each~$i \in [c]$~(\ref{eqn:[x]}) and then computing the full entry of~$A$~(\ref{eqn:payoffmatrixsplit}).
The entries of~$B$ are computed analogously.

We now present the general form of Alice's and Bob's payoff matrices.
We incorporate Sattath's AQMS~(\ref{eqn:alicepayoffAQMSterm}) into the entries of~$A_i^\text{meas}$~(\ref{eqn:payoffrestartgeneral}) and~$B_i^\text{meas}$~(\ref{eqn:Bobpayoffrestart}) to obtain the full payoff matrices~$A$ and~$B$.
Focusing on Alice, the~$(\alpha, \beta)$ entry~(\ref{eqn:alphabeta}) of~$A_i$ takes the form
\begin{equation}
\label{eqn:alicepayofffullmapping}
    A_i : \left[\left|S_0^c\right|\right] \times \left[\left|S_0^c\right|\right]  \to [0, 1 + \gamma_\text{S}] :  \left(\alpha, \beta\right) \mapsto A_i\left(\alpha, \beta; \gamma_\text{S}\right),
\end{equation}
where 
\begin{equation}
    \label{eqn:alicepayoffgeneralAS}
    A_i(\alpha, \beta; \gamma_\text{S})=
    \begin{cases}
        p(a_i) F_A^{i-1} + \gamma_\text{S} p(a_i) F_A^{i-1}  p(\mathcal{A}_i) & \mathcal{A}_i < b_1 \\
        p(a_i) F_A^{i-1} F_B^1 + \gamma_\text{S} p(a_i) F_A^{i-1}  p(b_1) & \mathcal{A}_i=b_1 \\
        p(a_i) F_A^{i-1} F_B^1 + \gamma_\text{S} p(a_i) F_A^{i-1}  F_B^1 p(\mathcal{A}_i - b_1) & b_1 < \mathcal{A}_i < b_1 + b_2  \\
        p(a_i) F_A^{i-1} F_B^2 + \gamma_\text{S} p(a_i) F_A^{i-1} F_B^1 p(b_2) &  \mathcal{A}_i=b_1 + b_2  \\
        \vdots \\
        p(a_i) F_A^{i-1}F_B^{c-1}  + \gamma_\text{S} p(a_i) F_A^{i-1}F_B^{c-1} p(\mathcal{A}_i - \mathcal{B}_{c-1}) &  \mathcal{B}_{c-1} < \mathcal{A}_i < \mathcal{B}_c  \\
        p(a_i) F_A^{i-1}F_B^{c} + \gamma_\text{S} p(a_i) F_A^{i-1}F_B^{c-1} p(b_c) &  \mathcal{A}_i=\mathcal{B}_c  \\
        p(a_i) F_A^{i-1}F_B^{c} & \mathcal{A}_i > \mathcal{B}_c \\
    \end{cases}
\end{equation}
and
\begin{equation}
    \label{eqn:failureterm}
    F_A^i := \prod_{k=1}^i (1 - p(a_k)), \; F_B^i:= \prod_{k=1}^i (1 - p(b_k))
\end{equation}
denotes the probabilities that Alice and Bob fail their first~$i$ quantum measurements, respectively. 
As the game is symmetric (Remark~\ref{remark:payofftwoplayer}), Bob's $i^\text{th}$ payoff is
\begin{equation}
    \label{eqn:bobpayoffgeneral}
    B_i=A_i^\intercal, 
\end{equation}
with $\gamma_\text{S}$ replaced by~$1 - \gamma_\text{S}$.
We note that, for~$c = 1$, Alice's and Bob's payoff matrices reduce to those of the S2QR model for Bitcoin~(\ref{eqn:alicepayoffLRS19})~\cite{LRS19}.
Assigning~$\gamma_\text{S} \gets 0$ recovers the stingy variant~\cite{LRS19}, whereas assigning~$\gamma_\text{S} \gets \nicefrac{1}{2}$ recovers the non-stingy variant~(\ref{eqn:alicepayoffLRS19}).

Continuing the example with~$c=2$, we illustrate the general structure of Alice's and Bob's payoff matrices.
We incorporate the terms from Sattath's AQMS~(\ref{eqn:alicepayoffAQMSterm}, \ref{eqn:alicepayoffAQMStermsametime}) into~$A_i^\text{meas}$~(\ref{eqn:alicepayoffrestartA1}, \ref{eqn:alicepayoffrestartA2}) and~$B_i^\text{meas}$ in this example.
Focusing on Alice's payoff matrix again, 
\begin{equation}\label{eqn:alicepayofftwochanceA1AS}
    A_1(\alpha, \beta; \gamma_\text{S})=\begin{cases}
        p(a_1) + \gamma_\text{S} p(a_1)p(a_1), & a_1 < b_1 \\
        p(a_1) (1 - p(b_1)) + \gamma_\text{S} p(a_1)p(b_1) & a_1=b_1 \\
        p(a_1)(1 - p(b_1)) + \gamma_\text{S} p(a_1) (1 - p(b_1))p(a_1 - b_1) & b_1 < a_1 < b_1 + b_2 \\ 
        p(a_1)(1 - p(b_1))(1 - p(b_2)) + \gamma_\text{S} p(a_1) (1 - p(b_1)) p(b_2) & a_1=b_1 + b_2 \\
        p(a_1)(1 - p(b_1))(1 - p(b_2)) & a_1 > b_1 + b_2,
    \end{cases}
\end{equation}
is the~$(\alpha, \beta)$ element of~$A_1$, while 
\begin{equation}\label{eqn:alicepayofftwochanceA2AS}
    A_2(\alpha, \beta; \gamma_\text{S})=\begin{cases}
        p(a_2)(1 - p(a_1)) + \gamma_\text{S} p(a_2)(1 - p(a_1))p(a_1 + a_2) & a_1 + a_2 < b_1 \\
        p(a_2)(1 - p(a_1)) (1 - p(b_1)) + \gamma_\text{S} p(a_2)(1-p(a_1))p(b_1) & a_1 + a_2=b_1 \\ 
        p(a_2)(1 - p(a_1))(1 - p(b_1)) + \gamma_\text{S} p(a_2)  (1-p(a_1)) (1 - p(b_1))p(a_1 + a_2 - b_1) & b_1 < a_1 + a_2 < b_1 + b_2 \\
        p(a_2)(1 - p(a_1))(1 - p(b_1))(1 - p(b_2)) + \gamma_\text{S} p(a_2)(1-p(a_1))(1-p(b_1)) p(b_2) & a_1 + a_2=b_1 + b_2 \\
        p(a_2)(1 - p(a_1))(1 - p(b_1))(1 - p(b_2)) & a_1 + a_2 > b_1 + b_2,
    \end{cases}
\end{equation}
is the~$(\alpha, \beta)$ element of~$A_2$, with 
\begin{equation}
    A=A_1 + A_2
\end{equation}
being Alice's full payoff matrix for our model in this setting.
\subsection{Optimal quantum mining strategies}

We compute examples of optimal quantum mining strategies.
Solving these examples is achieved by assigning numerical values to~$c$ and~$D$, evaluating Alice's and Bob's payoff matrices, and computing an optimal quantum mining strategy via the Lemke-Howson algorithm~\cite{LH64}.
We choose enough examples to illustrate our results, and the values of~$c$ and~$D$ are chosen such that the computations remain feasible for our computational resources.
We first present optimal quantum mining strategies for the~$c=2$ case.
Then, we consider the cases for
\begin{equation}
\label{eqn:c=3c=4}
c=3,\,c=4,
\end{equation}
to understand how Alice and Bob allocate their resources for higher-order measurements. 
Next, we consider the~$c=1$ case to establish a baseline used to compare with the multiple measurement cases. 
Finally, we report the expected payoffs of all computed optimal quantum mining strategies.
\subsubsection{Two measurements }

Now we turn to optimal quantum mining strategies for~$c=2$.
We first list the parameters used in the computation.
Then, we report the optimal quantum mining strategies computed in this case.

We compute optimal quantum mining strategies for~$c=2$ using the parameters outlined in~\S\ref{subsec:methods}.
Specifically, we use~
\begin{equation}
    \label{eqn:gamma0.5}
    \gamma_\text{S} \gets 0.5
\end{equation}
with the values of the network difficulty given in Eqs.~(\ref{eqn:assignDlow}) and~(\ref{eqn:assignDhigh}) for the low- and high-difficulty regimes, respectively.
For the ideal regime, we use the same value of~$D$ as the high-difficulty regime.
Finally, we assign
\begin{equation}
    \label{eqn:assignnqc2}
    n_q \gets 45
\end{equation}
for this computation.

We present examples of optimal quantum mining strategies for all three network difficulty regimes in the~$c=2$ case. 
Figure~\ref{fig:nasheqheatmaps_2c} shows a Nash equilibrium for the low-difficulty regime, whereas Fig.~\ref{fig:nasheqheatmaps_2c_ideal} displays a Nash equilibrium for the ideal regime. 
For the high-difficulty regime, we find that the only optimal quantum mining strategy for Alice and Bob is to allocate all of their resources into their final quantum measurements. 
\begin{figure}[ht]
    \centering
    \includegraphics[width=\textwidth]{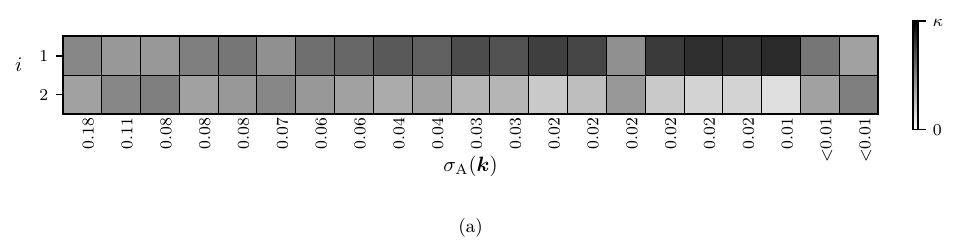}
    \includegraphics[width=\textwidth]{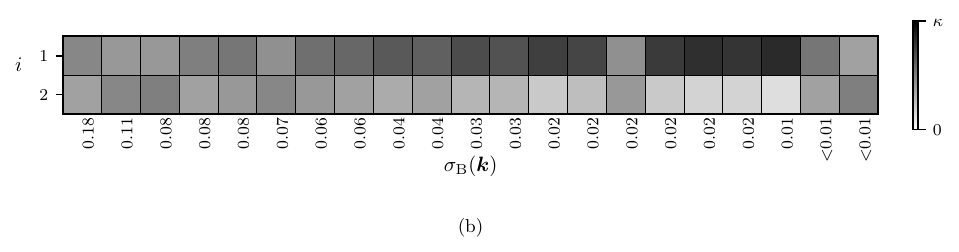}
    \caption{ The support of the mixed strategies for (a) Alice and (b) Bob in the Nash equilibrium for~$c=2$ in the low-difficulty regime.
    Each column corresponds to a pure strategy~$\bm{k}$ in the support of the corresponding mixed strategy for Alice or Bob, $\sigma_{\text{A,B}}(\bm{k})$, with the probability of playing each strategy shown on the abscissa with the probabilities adding to one.
    Each row corresponds to the measurement index~$i$ of the pure strategies, shown on the ordinate. 
    The shade of each cell indicates the number of Grover iterations to perform for the specific measurement index and pure strategy combination, where white represents~$0$ Grover iterations and black represents~$\kappa$ Grover iterations.
    }
    \label{fig:nasheqheatmaps_2c}
\end{figure}
\begin{figure}[ht]
    \centering
    \includegraphics[width=\textwidth]{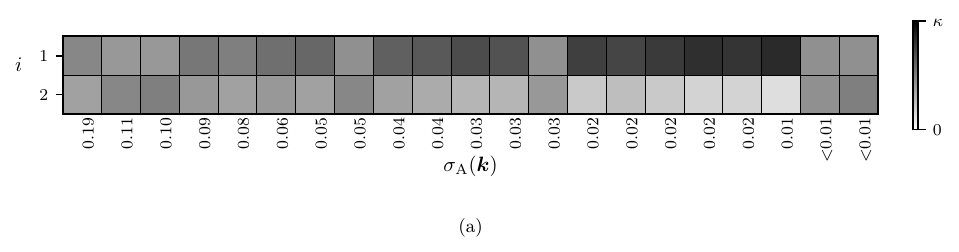}
    \includegraphics[width=\textwidth]{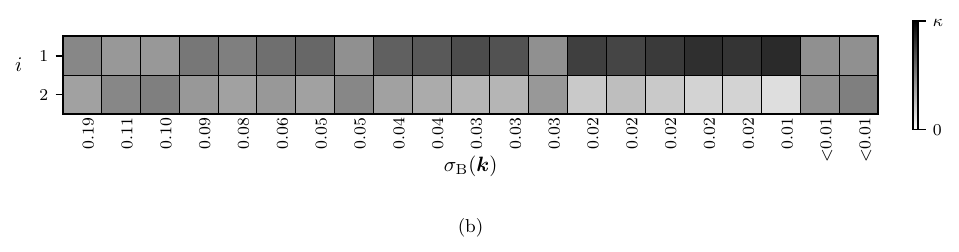}
    \caption{ The support of the mixed strategies for (a) Alice and (b) Bob in the Nash equilibrium for~$c=2$ in the ideal regime, shown in the same format as Fig.~\ref{fig:nasheqheatmaps_2c}.
    }
    \label{fig:nasheqheatmaps_2c_ideal}
\end{figure}
\subsubsection{Three and four measurements}

We present optimal quantum mining strategies for the case where both quantum miners perform $c \in \{3, 4\}$~(\ref{eqn:c=3c=4}) quantum measurements between successive blocks.
We first report the parameters used in this computation.
Then, we present the optimal quantum mining strategies.
Finally, we illustrate that these optimal quantum mining strategies are not unique by providing an example of a second optimal strategy computed in the low-difficulty regime for~$c=3$.

For these computations, we use the same parameters as the computation for~$c=2$. 
However, we assign
\begin{equation}
\label{eqn:assignnqc2low}
n_q \gets 21
\end{equation}
for the~$c=3$ low-difficulty regime, 
\begin{equation}
\label{eqn:assignnqc3low}
n_q \gets 35
\end{equation}
for~$c=3$ in the ideal and high-difficulty regimes, and 
\begin{equation}
\label{eqn:assignnqc4}
n_q \gets 25
\end{equation}
for~$c=4$.

Now we present the computed optimal quantum mining strategies for~$c=3$ and~$c=4$~(\ref{eqn:c=3c=4}).
Figures~\ref{fig:nasheqheatmaps_3c_low} 
\begin{figure}[ht]
    \centering
    \includegraphics[width=0.49\textwidth]{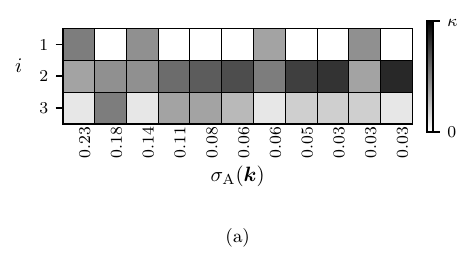}
    \includegraphics[width=0.49\textwidth]{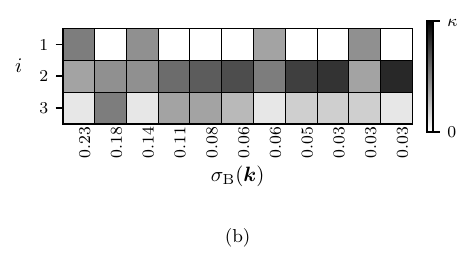}
    \caption{The support of the mixed strategies for (a) Alice and (b) Bob in the Nash equilibrium for~$c=3$ in the low-difficulty regime, shown in the same format as Fig.~\ref{fig:nasheqheatmaps_2c}.
    }
    \label{fig:nasheqheatmaps_3c_low}
\end{figure}
and~\ref{fig:nasheqheatmaps_4c_low}
\begin{figure}[ht]
    \centering
    \includegraphics[width=0.49\textwidth]{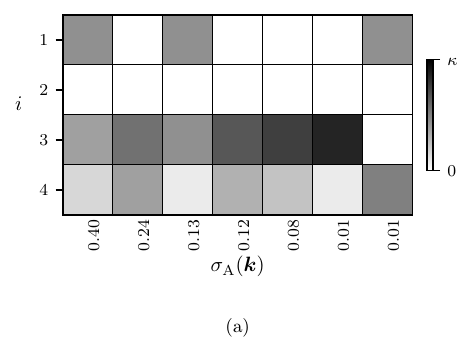}
    \includegraphics[width=0.49\textwidth]{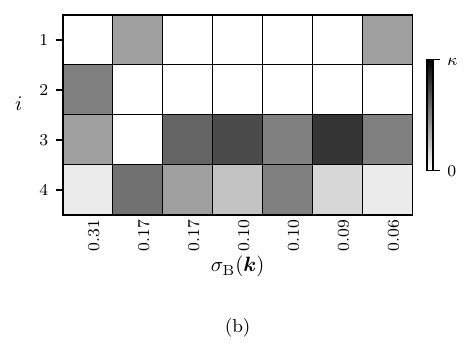}
    \caption{ The support of the mixed strategies for (a) Alice and (b) Bob in the Nash equilibrium for~$c=4$ in the low-difficulty regime, shown in the same format as Fig.~\ref{fig:nasheqheatmaps_2c}.
    }
    \label{fig:nasheqheatmaps_4c_low}
\end{figure}
display an optimal quantum mining strategy for the low-difficulty regime for~$c=3$ and $c=4$, respectively, while Figs.~\ref{fig:nasheqheatmaps_3c_ideal}
\begin{figure}[ht]
    \centering
    \includegraphics[width=\textwidth]{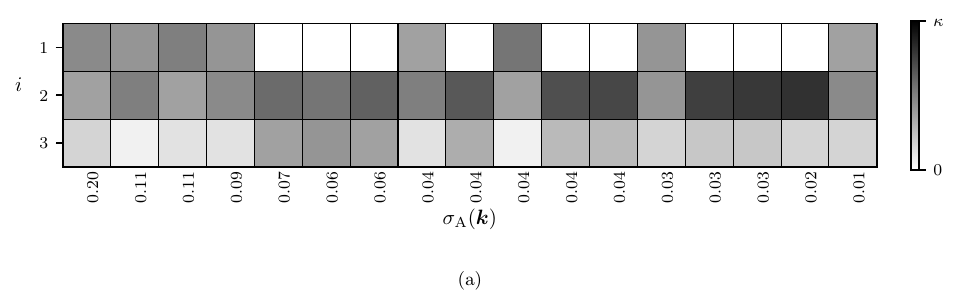}
    \includegraphics[width=\textwidth]{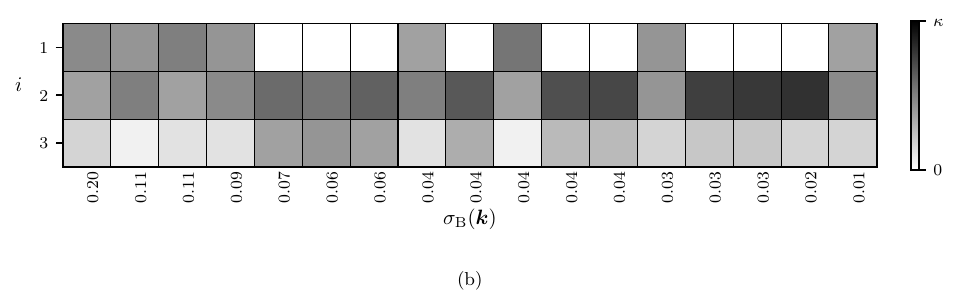}
    \caption{ The support of the mixed strategies for (a) Alice and (b) Bob in the Nash equilibrium for~$c=3$ in the ideal regime, shown in the same format as Fig.~\ref{fig:nasheqheatmaps_2c}.
    }
    \label{fig:nasheqheatmaps_3c_ideal}
\end{figure}
and~\ref{fig:nasheqheatmaps_4c_ideal}
\begin{figure}[ht]
    \centering
    \includegraphics[width=0.49\textwidth]{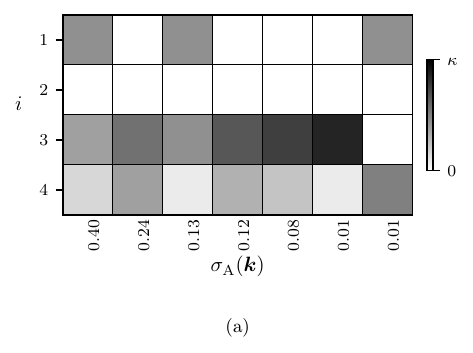}
    \includegraphics[width=0.49\textwidth]{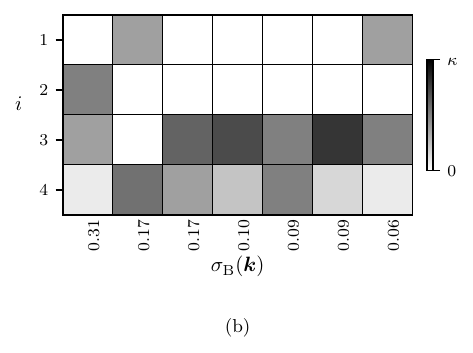}
    \caption{ The support of the mixed strategies for (a) Alice and (b) Bob in the Nash equilibrium for~$c=4$ in the ideal regime, shown in the same format as Fig.~\ref{fig:nasheqheatmaps_2c}.
    }
    \label{fig:nasheqheatmaps_4c_ideal}
\end{figure}
display an optimal quantum mining strategy for the ideal regime for~$c=3$ and $c=4$, respectively. 
Similar to the~$c=2$ case, we find that for the high-difficulty regime, the only optimal quantum mining strategy for Alice and Bob is to allocate all of their resources into their final quantum measurements.

To illustrate that the optimal quantum mining strategies are not unique, we compute an additional Nash equilibrium for~$c=3$ in the low-difficulty regime. 
This strategy is displayed in Fig.~\ref{fig:nasheqheatmaps_3c_low2}.
\begin{figure}[ht]
    \centering
    \includegraphics[width=0.49\textwidth]{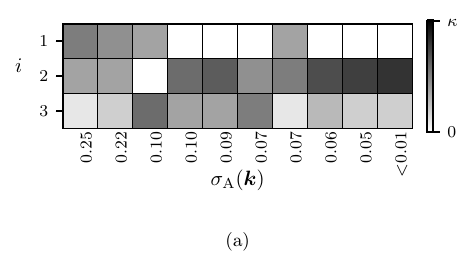}
    \includegraphics[width=0.49\textwidth]{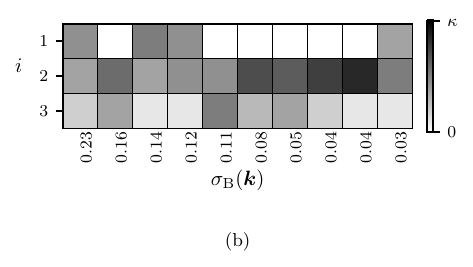}
    \caption{ The support of the mixed strategies for (a) Alice and (b) Bob in the second Nash equilibrium for~$c=3$ in the low-difficulty regime, shown in the same format as Fig.~\ref{fig:nasheqheatmaps_2c}.
    }
    \label{fig:nasheqheatmaps_3c_low2}
\end{figure}
\subsubsection{One measurement and expected payoffs}

Now we turn to the~$c=1$ case, which serves as a baseline used to compare how the quantum miners should optimally allocate their resources for multiple measurements. 
We first describe the parameters used in this computation and then report the computed optimal quantum mining strategies. 
Then, we report the expected payoffs for all the previously computed optimal quantum mining strategies.

For computing these optimal quantum mining strategies, we use the same parameters as the computation for~$c=2$. 
However, here we assign
\begin{equation}
\label{eqn:assignnqc1}
n_q \gets 200,
\end{equation}
which in this case also gives the dimension of Alice's and Bob's payoff matrices.

We report the computed optimal quantum mining strategies for~$c=1$. 
As the optimal quantum mining strategies for~$c=1$ contain many pure strategies in their support, we present the distribution of Grover iterations performed in the support of these optimal strategies as a histogram rather than a heatmap as for~$c \geq 2$.
Figure~\ref{fig:nasheqheatmaps_1c_low} 
\begin{figure}[ht]
    \centering
    \includegraphics[width=\textwidth]{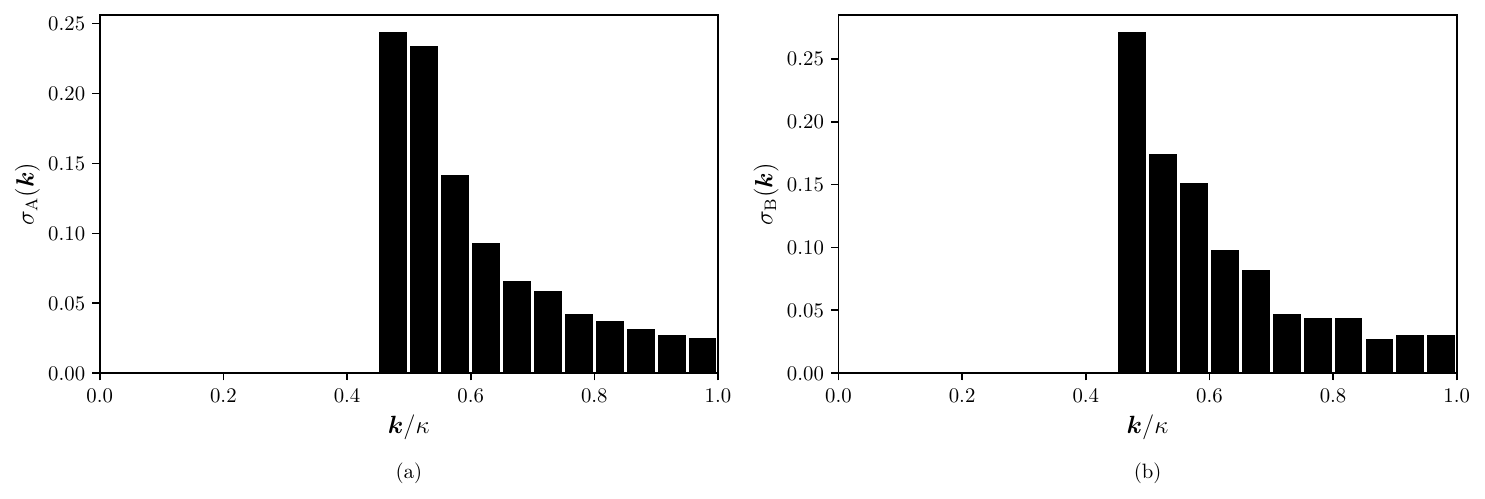}
    \caption{ The distribution of Grover iterations in the support of the mixed strategies for (a) Alice and (b) Bob in the Nash equilibrium for~$c=1$ in the low-difficulty regime. 
    The normalized number of Grover iterations~$\bm{k} /\kappa$ is shown on the abscissa, divided into equal-width bins.
    The probability of performing the number of Grover iterations corresponding to each bin in the support of the Nash equilibrium~$\sigma_{\text{A,B}}(\bm{k})$, is shown on the ordinate, with the probabilities adding to one.
    }
\label{fig:nasheqheatmaps_1c_low}
\end{figure}
shows an optimal quantum mining strategy for the low-difficulty regime, whereas Fig.~\ref{fig:nasheqheatmaps_1c_ideal}
\begin{figure}[ht]
    \centering
    \includegraphics[width=\textwidth]{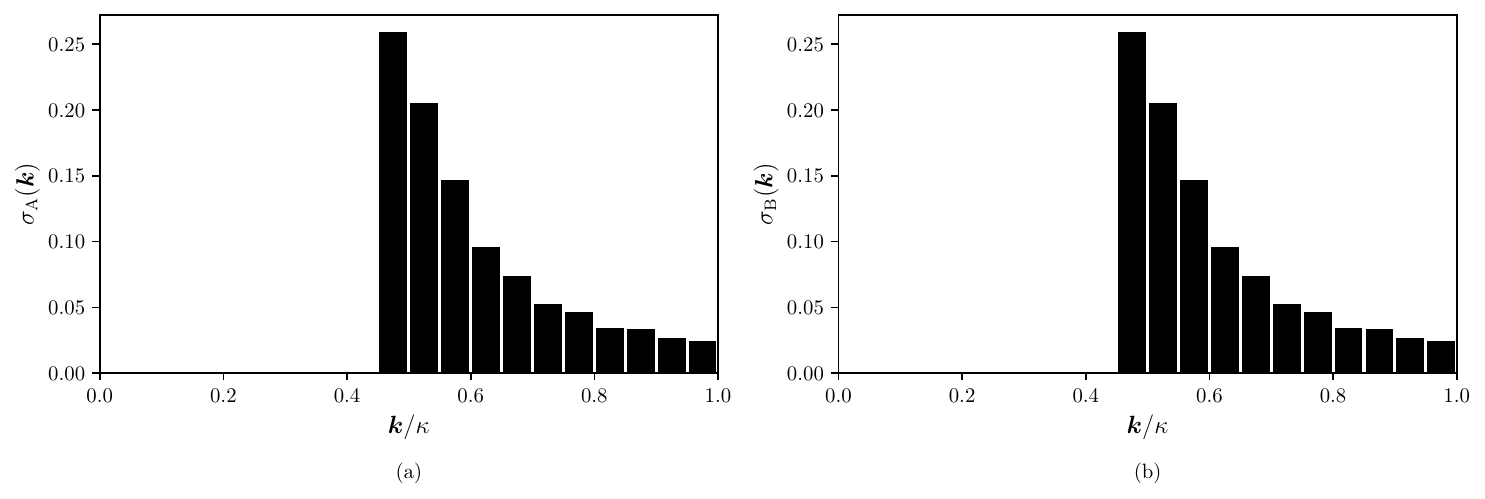}
    \caption{The distribution of Grover iterations in the support of the mixed strategies for (a) Alice and (b) Bob in the Nash equilibrium for~$c=1$ in the ideal regime, shown in the same format as Fig.~\ref{fig:nasheqheatmaps_1c_low}.
    }
\label{fig:nasheqheatmaps_1c_ideal}
\end{figure}
shows an optimal quantum mining strategy for the ideal regime. 
For the high-difficulty regime, we find that the only optimal quantum mining strategy is where Alice and Bob perform~$k_\text{opt}^{32}$~(\ref{eqn:kopt32}) Grover iterations before measuring.

Now we report the quantum miners' expected payoffs for each of the previously computed optimal quantum mining strategies (Remark~\ref{remark:expectedpayoff}).
Tables~\ref{table:expectedpayoffs_low} and~\ref{table:expectedpayoffs_ideal} summarize Alice's and Bob's expected payoffs for the low-difficulty and ideal regimes, respectively.
The expected payoffs for all difficulty regimes are reported to six decimal places, as any differences beyond this precision occur beyond the numerical precision of Python's 64-bit floating-point arithmetic. 
In the high-difficulty regime, all computed optimal strategies yield an expected payoff of $8.9 \times 10^{-5}$.
\begin{table}[H]
    
    \centering
      \begin{tabular}{c c c}
          $c$ & Quantum Miner & $U(\sigma^*)$ \\
          1 & Alice & 0.537708 \\
          1 & Bob   & 0.534115 \\
          \hline
          2 & Alice & 0.625250 \\
          2 & Bob   & 0.625250 \\
          \hline
          $3^{(1)}$ & Alice & 0.617527 \\
          $3^{(1)}$ & Bob   & 0.617527 \\
          $3^{(2)}$ & Alice & 0.607639 \\
          $3^{(2)}$ & Bob   & 0.596411 \\
          \hline
          4 & Alice & 0.639309 \\
          4 & Bob   & 0.580668 \\
          \hline
      \end{tabular}
    \caption{ Alice's and Bob's expected payoffs~$U(\sigma^*)$ for each computed optimal quantum mining strategy in the low-difficulty regime.
    The superscripts $(1)$ and $(2)$ denote the two distinct Nash equilibria computed for $c=3$.}
    \label{table:expectedpayoffs_low}
\end{table}
\begin{table}[H]
    \centering
      \begin{tabular}{c c c}
          $c$ & Quantum Miner & $U(\sigma^*)$ \\
          1 & Alice & 0.532788 \\
          1 & Bob   & 0.532788 \\
          \hline
          2 & Alice & 0.624900 \\
          2 & Bob   & 0.624900 \\
          \hline
          3 & Alice & 0.633620 \\
          3 & Bob   & 0.633620 \\
          \hline
          4 & Alice & 0.639528 \\
          4 & Bob   & 0.580464 \\
          \hline
      \end{tabular}
    \caption{ Alice's and Bob's expected payoffs~$U(\sigma^*)$ for each computed optimal quantum mining strategy in the ideal regime.}
    \label{table:expectedpayoffs_ideal}
\end{table}
\FloatBarrier
\subsection{Decrease of 51\% attack threshold}

We report our results for estimating~$p_\text{stale}$ from simulating the deployment of the computed optimal quantum mining strategies in the Bitcoin network. 
We begin by presenting the estimated values of~$p_\text{stale}$ for all previously computed optimal quantum mining strategies.
These estimates are obtained by simulating the deployment of the optimal quantum mining strategies in the Bitcoin network by the quantum miners.
Then, we show estimates of $p_\text{stale}$ over a range of network difficulties for~$c \in [4]$.

First we report the results of performing a Monte Carlo simulation to estimate~$p_\text{stale}$.
Table~\ref{table:stalerates} reports the 95th and 99th percentiles~$P_{95}$ and~$P_{99}$~(\ref{eqn:P95P99}), the maximum observed per-day stale-rate, the empirical probability~$P_{>\nicefrac13}$~(\ref{eqn:empprob}), the empirical detection probability~$P_\text{det}$~(\ref{eqn:pdet}), and the average blocks and forks per day for all computed optimal quantum mining strategies, organized by the number of quantum measurements~$c$.
\begin{table}[ht]
    \centering
\begin{tabular}{| c | l | d{1.4} | d{1.4} | d{1.4} | d{1.6} | d{1.6} | d{3.2} | d{2.3} |}
\hline
$c$ & diff & \multicolumn{1}{c|}{$P_{95}$} & \multicolumn{1}{c|}{$P_{99}$} & \multicolumn{1}{c|}{max} & \multicolumn{1}{c|}{$P\left(p_\text{stale}>\nicefrac13\right)$} & \multicolumn{1}{c|}{$P_\text{det}$} & \multicolumn{1}{c|}{abd} & \multicolumn{1}{c|}{afd} \\
\hline
1 & High  & 0.0000 & 0.0000 & 0.0135 & 0.000000 & 0.000001 & 144.02 & 0.004 \\
1 & Low   & 0.2899 & 0.3011 & 0.3387 & 0.000008 & 1.000000 & 376.74 & 98.731 \\
1 & Ideal & 0.2894 & 0.3005 & 0.3350 & 0.000001 & 1.000000 & 376.47 & 98.464 \\
\hline
2 & High  & 0.0000 & 0.0000 & 0.0146 & 0.000000 & 0.000004 & 144.01 & 0.004 \\
2 & Low   & 0.2639 & 0.2753 & 0.3126 & 0.000000 & 1.000000 & 388.05 & 91.659 \\
2 & Ideal & 0.2644 & 0.2757 & 0.3232 & 0.000000 & 1.000000 & 388.19 & 91.873 \\
\hline
3 & High    & 0.0000 & 0.0000 & 0.0142 & 0.000000 & 0.000005 & 144.03 & 0.004 \\
3 & Low$^1$ & 0.2722 & 0.2834 & 0.3271 & 0.000000 & 1.000000 & 393.31 & 96.167 \\
3 & Low$^2$ & 0.2588 & 0.2703 & 0.3067 & 0.000000 & 1.000000 & 386.28 & 89.167 \\
3 & Ideal   & 0.2600 & 0.2714 & 0.3143 & 0.000000 & 1.000000 & 387.36 & 89.945 \\
\hline
4 & High  & 0.0000 & 0.0000 & 0.0154 & 0.000000 & 0.000005 & 144.01 & 0.004 \\
4 & Low   & 0.2543 & 0.2656 & 0.3198 & 0.000000 & 1.000000 & 384.22 & 87.123 \\
4 & Ideal & 0.2545 & 0.2659 & 0.3063 & 0.000000 & 1.000000 & 384.29 & 87.205 \\
\hline
\end{tabular}

    \caption{Monte Carlo simulation for all computed optimal quantum mining strategies with averages taken over 1\,000\,000 simulated days.
    All values are reported to the fewest decimal digits to ensure distinguishability. 
    Each row corresponds to a specific combination of the number of quantum measurements~$c$ and network difficulty regime (diff), where High/Low/Ideal denotes the difficulty regime and the superscripts distinguish the first and second Nash equilibria for the~$c=3$ low-difficulty case. 
    The reported quantities are the 95th and 99th percentiles~$P_{95}$ and $P_{99}$, the maximum observed stale rate, the empirical probability~$P_{>\nicefrac13}$ that any given day's stale rate exceeds the $\nicefrac13$ benchmark, the empirical detection probability~$P_\text{det}$, and the average blocks (abd) and forks per day (afd).
    }
    \label{table:stalerates}
\end{table}

We now report on how~$p_\text{stale}$ is affected by the value of the network difficulty~$D$.
Our simulations provide insight into how Bitcoin's security is affected as the quantum miners allocate their resources over additional quantum measurements.
For each~$c\in[4]$, we compute an optimal quantum mining strategy and estimate~$P_{95}$~(\ref{eqn:P95P99}) and~$P_\text{det}$~(\ref{eqn:pdet}) resulting from the quantum miners employing a corresponding optimal quantum mining strategy over 50\,000 days. 
Figure~\ref{fig:difficultyplots} shows~$P_{95}$ and~$P_\text{det}$ as a function of~$D \in 2^{[30,45]}$~(\ref{eqn:valuesD})
in increments of powers of two, for each value of~$c$.
\begin{figure}
    \centering
    \includegraphics[width=\textwidth]{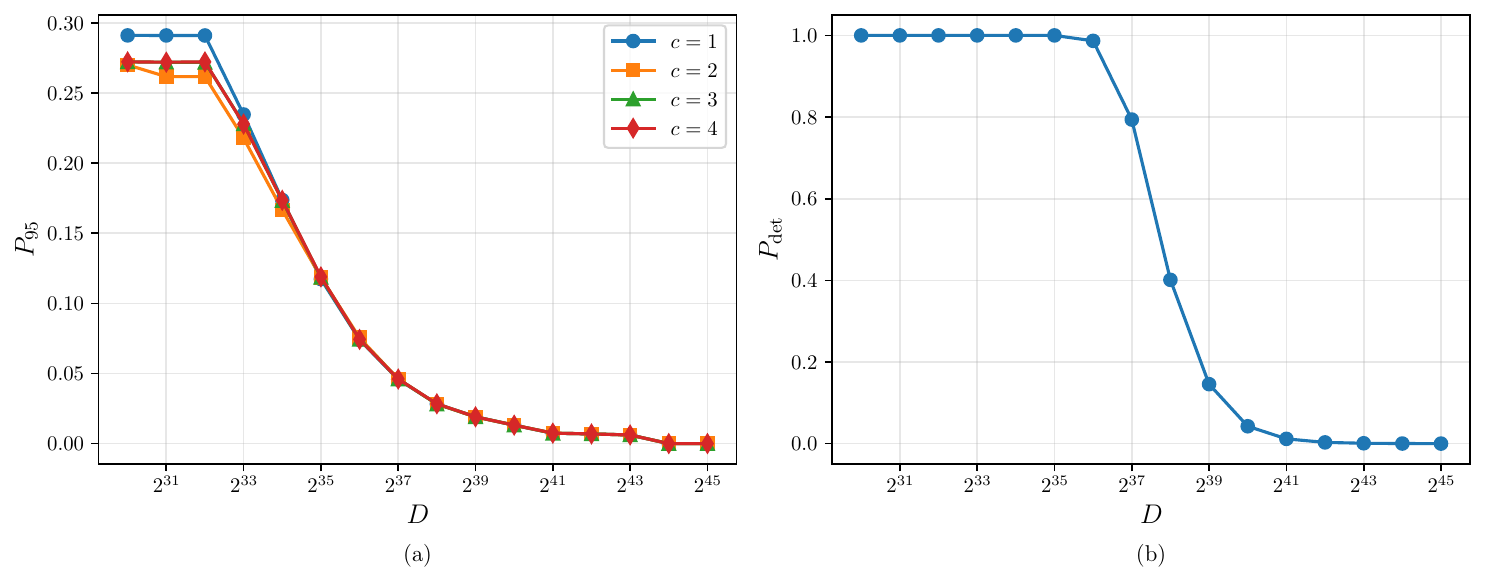}
\caption{ Estimated $P_{95}$ (a) and~$P_\text{det}$ (b) as functions of~$D$ for quantum miners employing an optimal quantum mining strategy over 50\,000 simulated days.
    Four curves are shown on each plot, each representing a specific number of quantum measurements~$c$ over which the quantum miners allocate their resources.
    In (b), all four curves lie on top of each other.}
    \label{fig:difficultyplots}
\end{figure}
\section{Discussion}
\label{sec:discussion}

Now we discuss our results.
We first discuss the payoff derived for quantum miners.
Next, we provide insight into the optimal quantum mining strategies computed for~$c \in [4]$~(\ref{eqn:[x]}) across the three network difficulty regimes~(Def.~\ref{def:networkdifficultyregimes}), and discuss the expected payoffs associated with these optimal quantum mining strategies. 
Subsequently, we discuss our simulation results for estimating~$p_\text{stale}$ when the computed optimal quantum mining strategies are deployed and discuss the broader security implications for the Bitcoin network.

Here we discuss our result that Alice and Bob each have constructed their respective payoff matrices for the case that each only performs one quantum measurement.
Alice and Bob then use their knowledge of their payoff matrices to compete in a quantum race with the goal of finding a valid proof-of-work before the other.
Alice's and Bob's payoff matrices represent their preferences over every combination of pure strategies and fully summarize the strategic game (Def.~\ref{def:strategicgame}) representing the quantum race.
In the S2QR model for Bitcoin~\cite{LRS19}, Alice's and Bob's payoff matrices were determined for the setting where they perform only one quantum measurement between consecutive blocks and are peaceful~(Remark~\ref{remark:peacefulregulation}).
The measurement-only terms in Alice's and Bob's payoff matrices~(\ref{eqn:payoffrestartgeneral}) are similar to the S2QR model for Bitcoin: each entry corresponds to the joint probability that a quantum miner succeeds in their current quantum measurement, given that all previous quantum measurements from either quantum miner failed.
For~$c=1$ and~$\gamma_\text{S} \to 0$, we recover the payoff matrices of the S2QR model for Bitcoin~\cite{LRS19}.

Extending this single-quantum-measurement case for payoff matrices to the general case of~$c$ multiple quantum measurements each, the entries of~$A^\text{meas}$ and~$B^\text{meas}$ capture every possibility in which Alice and Bob can win the quantum race respectively across their~$c$ quantum measurements each.
Likewise, the entries of~$A^\text{AQMS}$ and~$B^\text{AQMS}$ reflect the number of Grover iterations that Alice or Bob have accumulated when the other yields a valid proof-of-work, triggering the employment of Sattath's AQMS, depending on the relative timing of Alice's and Bob's~$c$ quantum measurements. 
We validate our result for~$c=2$ by showing that~$A_1^\text{meas}$~(\ref{eqn:alicepayoffrestartA1}) and~$A_2^\text{meas}$~(\ref{eqn:alicepayoffrestartA2}) agree with the result derived in Ray's PhD thesis~\cite{Ray2020},
which is the sole prior case of studying a quantum race between two competing quantum Bitcoin miners with more than one quantum measurement.

Alice and Bob use their payoff matrices to compute optimal quantum mining strategies corresponding to Nash equilibria.
Given the identical resources of the quantum miners, we expect the optimal quantum mining strategies to be symmetric equilibria.
Of our 13 computed optimal strategies for~$c \in [4]$, the equilibria are symmetric in nine cases.
We observe asymmetric equilibria for~$c \in \{1, 4\}$ in the low-difficulty regime, $c=4$ in the ideal regime, and the second optimal strategy computed for~$c=3$ in the low-difficulty regime.
As symmetric equilibria are guaranteed to exist in symmetric games (Theorem~\ref{theorem:nashsymm}), these cases of asymmetric equilibria likely reflect the path-dependence of the Lemke-Howson algorithm.

In the low-difficulty and ideal regimes, Alice and Bob should allocate their fixed resources over multiple quantum measurements. 
Specifically, for~$c=2$, Alice and Bob should allocate their resources across both quantum measurements to balance the high success probability of the FQS algorithm in these regimes against the temporal advantage of one measuring before the other.
However, for~$c \geq 3$, sometimes Alice and Bob should perform zero Grover iterations before a quantum measurement (Remark~\ref{remark:zeroiterations}) to maximize the success probability of other quantum measurements.
For all computed optimal strategies with~$c \geq 3$, the strategies involve pure strategies in which Alice and Bob concentrate their resources into fewer quantum measurements.
This outcome is influenced by the Grover iteration budget~$\kappa$~(\ref{eqn:groveriterbudget}) and reflects that spreading this budget over many quantum measurements is suboptimal, as the success probability of the FQS algorithm grows quadratically with the number of Grover iterations performed. 
For~$c=1$, Figs.~\ref{fig:nasheqheatmaps_1c_low} and~\ref{fig:nasheqheatmaps_1c_ideal} show that, most of the time, Alice and Bob should perform approximately half of their budgeted Grover iterations before measuring.
We note that this distribution is specific to the computed optimal quantum mining strategy and may not generalize.

In the high-difficulty regime, the optimal strategy is for the quantum miners to allocate all of their resources into their final quantum measurement.
This result is counterintuitive, as one might expect that it would be more advantageous for Alice or Bob to allocate all of their resources into their first quantum measurement and gain a temporal advantage.
This outcome is likely a result of the combination of the near-zero success probability of the FQS algorithm in this regime and payoff contributions from Sattath's AQMS~(\ref{eqn:alicepayoffAQMSterm}), which could slightly favour later quantum measurements.
Although we do not have a rigorous explanation of this behaviour, including Sattath's AQMS terms omitted from Alice's and Bob's payoff~(\ref{subsec:mathematics}) could remove this bias and result in the more intuitive strategy of allocating all resources to the first quantum measurement.

With respect to the expected payoff of the computed optimal quantum mining strategies, allocating resources into two quantum measurements each yields the highest expected payoff (Tables~\ref{table:expectedpayoffs_low} and \ref{table:expectedpayoffs_ideal}).
However, the computed optimal quantum mining strategies are not unique, as we demonstrate by computing two distinct optimal quantum mining strategies for~$c=3$ in the low-difficulty regime~(Figs.~\ref{fig:nasheqheatmaps_3c_low}
and~\ref{fig:nasheqheatmaps_3c_low2}).
The Lemke-Howson algorithm traces a singular path through the space of strategy profiles of Alice and Bob, returning the first Nash equilibrium found~\cite{LH64}.
Therefore, other optimal quantum mining strategies with higher expected payoffs might exist but were not returned by the Lemke-Howson algorithm.

Alice and Bob deploy their optimal quantum mining strategies to beat each other at the race without regard to beating the classical Bitcoin network.
Their quantum race has the side effect of increasing~$p_\text{stale}$ due to the employment of Sattath's AQMS~\cite{Sat20}.
This increased stale rate poses a threat to the entire network by creating stale blocks through forks,
but the significance of this threat decreases and even becomes negligible in the high-difficulty regime.
For all cases in the high-difficulty regime, 99\% of the simulated days exhibit no forks at all, as the 95th and 99th percentiles are zero up to machine precision.
Similarly, $P_\text{det}$ is negligible in this regime, meaning that the presence of the quantum miners is statistically indistinguishable from the classical network.
The maximum observed~$p_\text{stale}$ in the high-difficulty regime across all cases is $0.0154$,
which is approximately 20 times higher than the classical stale rate but still not significant in the sense that the 51\% attack is still infeasible.
The empirical probability that~$p_\text{stale} > \nicefrac13$ is zero up to machine precision.
Despite that two quantum miners are racing each other,
the Bitcoin network produces the expected 144 blocks per day (Remark~\ref{remark:averageblocks}), confirming the negligible impact of quantum miners in this regime.
This result that quantum miners have negligible effect in the high-difficulty regime is expected, as the success probability of the FQS algorithm is near-zero in this regime. \par

Contrariwise, in the ideal and low-difficulty regimes, a single execution of the FQS algorithm is sufficient to yield a proof-of-work with high probability.
A single execution suffices in the low-difficulty regime as the expected number of nonces to check is fewer than $2^{32}$, and, in the ideal regime, the quantum oracle encodes the entire block header. 
As a general trend, increasing the value of~$c$ decreases~$P_{95}$ and~$P_{99}$, as allocating resources over more quantum measurements decreases the success probability of those quantum measurements.
Across all nine cases in the ideal and low-difficulty regimes, the empirical probability that~$p_\text{stale}$ exceeds the~$\nicefrac13$ benchmark (Prop.~\ref{prop:quantumminingthreshold}) is zero empirically in all but two cases.
These two exceptions to being zero arise for either $c=1$ in the low-difficulty regime with probability~$8 \times 10^{-6}$
or~$c=1$ in the ideal regime with probability~$1 \times 10^{-6}$.
These cases do not contradict Prop.~\ref{prop:quantumminingthreshold}, as a 51\% attack requires~$p_\text{stale}$ to exceed the threshold~(\ref{eqn:proposition}) throughout the duration of the attack, not on a single day (Remark~\ref{remark:51duration}). 
Even then, these cases correspond to a threat emerging once every~342 or~2738 years, respectively, meaning that the Bitcoin network remains secure against this quantum attack.
Finally, we find~$P_\text{det}=1$ across all nine of these cases, meaning that it is statistically certain to detect the presence of quantum miners in the ideal and low-difficulty regimes.

Finally, to assess the impact of optimal quantum mining strategies on Bitcoin's security, we study how~$P_{95}$ and~$P_\text{det}$ vary as functions of~$D$ for~$c \in [4]~(\ref{eqn:[x]})$ shown in  Fig.~\ref{fig:difficultyplots}.
For~$D \leq 2^{32}$, $P_{95}$ plateaus, as a single execution of the FQS algorithm is sufficient to yield a proof-of-work with high probability.
For~$D > 2^{32}$, $P_{95}$ approaches zero, which is consistent with the fact that the success probability~(\ref{eqn:groversuccessbitcoin}) decreases for increasing~$D$ whereas the number of Grover iterations is restricted to at most~$k_\text{opt}^{32}$~(\ref{eqn:kopt32}).
$P_\text{det}$ follows the same trend, but drops off at a higher value of~$D$ than~$P_{95}$. 
This offset is because even a small increase in~$p_\text{stale}$ above the classical baseline is statistically detectable by the network due to the small value of that baseline.
For all values of~$D$ and~$c$ considered, ~$P_{95} < \nicefrac1{3}$, reinforcing that the Bitcoin network remains secure against this quantum threat for our model.

We now explain that the threat posed to the Bitcoin network by a pair of quantum miners racing each other is negligible for the foreseeable future. 
In the high-difficulty regime, which reflects both current and projected values of the network difficulty~$D$, the optimal quantum mining strategies of Alice and Bob produce a~$p_\text{stale}$ that is statistically indistinguishable from the classical value of 0.0024~\cite{SSJ+18}, meaning that the Bitcoin network remains secure in this regime.
In the low-difficulty and ideal regimes, $p_\text{stale}$ increases significantly, approaching the~$\nicefrac{1}{3}$ benchmark on some days, which is evidenced by the nonzero~$P_{>\nicefrac{1}{3}}$ for~$c=1$.
However, neither of these regimes is likely to hold in practice.
The low-difficulty regime requires a network difficulty~$D$ at values not seen since the early days of Bitcoin, and the ideal regime requires a quantum-phase oracle capable of encoding the entire 640-bit block header (instead of just the 32-bit nonce), which is unrealistic given that even the 32-bit quantum-phase oracle is known to be expensive to construct~\cite{ABL+17}.
Furthermore, even if these conditions were met in reality, we find that~$P_\text{det}=1$ in both regimes, meaning that one could detect quantum mining with statistical certainty.
\section{Conclusions}
\label{sec:conclusion}

We solved Problem~\ref{prob:quantumthreat} by determining optimal quantum mining strategies for two quantum miners and estimating the impact of these strategies on Bitcoin's security against a 51\% attack.
We extended the S2QR model for Bitcoin~\cite{LRS19} to include multiple quantum measurements between successive blocks and Sattath's AQMS~\cite{Sat20}.
We derived the payoff matrices of the two quantum miners in this novel setting, computed optimal quantum mining strategies across three network difficulty regimes, and estimated the effect of these optimal strategies on the stale rate~$p_\text{stale}$.

We find that two quantum miners cannot raise~$p_\text{stale}$ sufficiently to render the Bitcoin network vulnerable to a 51\% attack according to our model.
In other words,
we show that the Bitcoin network is safe against our adversarial model.
In the high-difficulty regime, which reflects both the current and projected values of the network difficulty~$D$, the effect of quantum miners on~$p_\text{stale}$ is negligible and statistically indistinguishable from that of classical Bitcoin miners. 
In the low-difficulty and ideal regimes, quantum mining produces a significant increase in~$p_\text{stale}$ compared to classical values.
However we find that~$p_\text{stale} < \nicefrac{1}{3}$ on all but a vanishingly-small fraction of simulated days, meaning that the Bitcoin network remains secure according to the historical benchmark.
Furthermore, these regimes are unlikely to hold in practice, and it is statistically certain that the presence of the quantum miners would be detected by the classical network.
The method we develop for studying two quantum miners sets the foundation for studying the threat posed by three or more quantum miners.

We now present our views on next steps for this work.
Our model for two quantum miners only accounts for query-related overhead and not for overhead associated with other tasks such as quantum-state preparation and constructing the quantum-phase oracle for each candidate block.
Thus, we have only modelled a best-case scenario for the quantum miners.
Hence, the quantum miners' combined effect on~$p_\text{stale}$
will be less than we have predicted. 
Directions for future work could include generalizing 
to colluding quantum miners or to more than two non-colluding quantum miners, as well as determining the minimum number of quantum miners racing each other that would create a dangerous~$p_\text{stale}$ given a network difficulty~$D$.
Other future work could 
incorporate the classical Bitcoin network as a single or multiple players joining the quantum race.
Future work could account properly for overhead associated with building and querying the oracle.
Each of these investigations would be valuable to connect the potential threat of quantum mining to the actual Bitcoin network.
Overall, we conclude that the Bitcoin network remains secure against the threat posed by two quantum miners on~$p_\text{stale}$.
\section{Acknowledgments}
\label{sec:acknowledgements}
This work has been supported by Canada's Natural Sciences and Engineering Research Council (NSERC). 
ZM thanks R.\ R.\ Nerem for his insightful conversations in the nascent stage of our work.
We acknowledge the traditional owners of the land on which this work was performed at the University of Calgary: the Treaty 7 First Nations and the M\'{e}tis Nation of Alberta.
\newpage
\bibliographystyle{unsrt}
\bibliography{ref}
\end{document}